\pdfoutput=1
\documentclass[11pt,letterpaper]{article}
\usepackage[T1]{fontenc}
\usepackage{lmodern,amsmath,amssymb,amsthm,mathtools,microtype}
\usepackage[margin=1in]{geometry}
\usepackage{booktabs,tabularx,enumitem,xcolor,longtable,needspace}
\usepackage[colorlinks=true,linkcolor=blue!40!black,citecolor=blue!40!black,urlcolor=blue!40!black]{hyperref}
\newtheorem{theorem}{Theorem}[section]
\newtheorem{lemma}[theorem]{Lemma}
\newtheorem{corollary}[theorem]{Corollary}
\newtheorem{proposition}[theorem]{Proposition}
\theoremstyle{definition}
\theoremstyle{remark}\newtheorem{remark}[theorem]{Remark}
\newcommand{\E}{\mathbb E}
\newcommand{\Prob}{\mathbb P}
\newcommand{\Var}{\operatorname{Var}}
\newcommand{\Poi}{\operatorname{Poi}}
\newcommand{\TV}{\operatorname{TV}}
\newcommand{\one}{\mathbf 1}
\newcommand{\eps}{\varepsilon}

\newcommand{\G}{\mathsf G}
\newcommand{\g}{\mathsf g}
\newcommand{\norm}[1]{\left\lVert#1\right\rVert}

\numberwithin{equation}{section}
\allowdisplaybreaks[2]
\setlist{itemsep=3pt,topsep=4pt}
\hypersetup{pdftitle={Uniform Renyi Entropy Estimation and Logarithmic Additive Functionals},pdfauthor={Arman Adibi and Piotr Krysta},pageanchor=false}

\begin{document}
\begin{titlepage}
\centering
\vspace*{0.45in}
{\LARGE\bfseries Optimal R\'enyi 
and Logarithmic Additive Functional Entropy Estimation\par}
\vspace{0.3in}
{\large Arman Adibi\qquad Piotr Krysta\par}
\vspace{0.12in}
Department of Computer Science\\
School of Computer and Cyber Sciences\\
Augusta University
\vspace{0.25in}
\begin{abstract}
We characterize the minimax sample complexity of estimating R\'enyi
entropy from independent samples on an alphabet of size at most $k$,
when the known order $\alpha>1$ may vary with $k$.\footnote{Lean
verification files of our proofs are available at \url{https://github.com/adibiarman/Optimal-Renyi-and-Logarithmic-Additive-Functional-Entropy-Estimation}.}
For every fixed additive accuracy $\varepsilon\in(0,1/2)$, we obtain
matching bounds with constants independent of both $k$ and $\alpha$.
Below order $\log k$, the sample complexity is
$\Theta_\varepsilon(\alpha k^{1-1/\alpha}X_\alpha(k,\varepsilon))$,
where a monotone scalar inequality determines the noninteger
multiplier. This characterization resolves the transitions as the
order approaches one or an integer, or grows with the alphabet.
Integer orders have complexity
$\Theta_\varepsilon(\alpha k^{1-1/\alpha})$, while orders at least
$\log k$ and min-entropy have complexity
$\Theta_\varepsilon(k\log k)$.

We also develop a framework for estimating logarithms of positive
additive functionals. Exact factorial-count variance identities and
weighted approximation duality retain the multiplicative thresholds
required for logarithmic accuracy. For every fixed continuous function
$f$ satisfying $f(0)=0$, $f(x)>0$ for $x>0$, and
$f(x)\sim ax^r$ as $x\downarrow0$, with $a>0$ and $r>1$, we show that
estimating $-\beta^{-1}\log\sum_i f(p_i)$, for fixed $\beta>0$,
has sample complexity $\Theta(k^{1-1/r})$ when $r$ is an integer
and $\Theta(k/\log k)$ otherwise. The proof uses a replication
reduction that preserves the sampling experiment. Smooth flat
functions demonstrate that regularity alone does not determine
these rates.

The R\'enyi upper bound combines polynomial estimates for small
probabilities, Gamma-function count statistics for larger
probabilities, and an independent pilot. A global affine correction
preserves the known first moment near order one. The lower bound
constructs positive mean-one priors by finite interpolation and
controls both statistical indistinguishability and entropy
separation after normalization.
\end{abstract}
\vfill
\end{titlepage}
\hypersetup{pageanchor=true}
\pagenumbering{roman}
\tableofcontents
\clearpage\pagenumbering{arabic}

\section{Introduction}
\label{sec:intro}

Estimating a distribution and estimating a functional of that distribution
are different statistical tasks. A sample may contain too little
information to recover most small probabilities, while still determining
their aggregate contribution to entropy. This distinction is central to
large-alphabet statistics and distribution property estimation. It also
creates a precise algorithmic question: which nonlinear summaries can be
estimated from substantially fewer observations than the alphabet size?

R\'enyi entropy~\cite{Renyi61} provides a family in which the answer
depends sharply on a real parameter. For a distribution
$p=(p_1,\ldots,p_k)$ and order $\alpha>1$, it is
$H_\alpha(p)=-(\alpha-1)^{-1}\log\sum_i p_i^\alpha$.
At integer orders, collisions provide unbiased power-sum estimators.
At noninteger orders, the behavior of $x^\alpha$ near zero creates an
approximation problem. The distinction becomes more delicate when
$\alpha$ changes with $k$: a fixed-order constant can conceal a transition
near an integer, near one, or near the high-order regime.

We resolve these transitions for all regimes of $\alpha$ through a uniform sample-complexity
characterization. We also ask which parts of our arguments extend beyond
powers. For a known positive function $f$, the logarithmic target
$-\beta^{-1}\log\sum_i f(p_i)$ requires relative accuracy in its
additive sum. Its statistical difficulty can therefore differ from that
of estimating the same sum to fixed absolute accuracy. We develop the
count and approximation identities for this loss and prove a transfer
theorem for every fixed function with a finite leading power at zero.

\paragraph{Contributions.}
Our first result gives matching bounds for every real $\alpha>1$, with
the accuracy fixed before the alphabet size and order. We define a scalar root that
retains the dependence on the distance to adjacent integers and on the
requested accuracy. The result includes the collision scale
$\alpha k^{1-1/\alpha}$ at integer orders and the saturation scale
$k\log k$ at large orders. At every fixed noninteger order it gives
$k/\log k$.

Our second result separates the general functional framework from its
quantitative evaluation for powers. We give exact count-variance and
likelihood identities, an approximation duality with the exact
logarithmic threshold, and explicit conditions needed to turn these
identities into statistical bounds. Our third result shows that, for
fixed function $f(x)\sim ax^r$, the leading order determines the alphabet
dependence of the logarithmic additive estimation. The reduction assumes
only continuity and positivity, and retains its accuracy loss explicitly.

\subsection{Prior work and the remaining uniformity question}

\paragraph{Entropy and R\'enyi estimation.}
Valiant and Valiant~\cite{VV11} established the optimal
$\Theta(k/\log k)$ sample complexity for constant-accuracy Shannon
entropy estimation. Acharya, Orlitsky, Suresh, and Tyagi~\cite{AOST17}
identified the integer/noninteger distinction for R\'enyi entropy:
fixed integer orders admit the collision scale
$k^{1-1/\alpha}$, whereas noninteger orders have near-linear alphabet
dependence. Their polynomial estimator gives an $O(k/\log k)$ upper
bound for fixed noninteger orders above one. Obremski and
Skorski~\cite{OS17} refined R\'enyi bounds, including
entropy-dependent integer-order estimates and falling-factorial
variance calculations. Adibi and Krysta~\cite{AK26} obtained matching
bounds for growing integer orders and min-entropy. Those results supply
important parts of the present picture; they do not give the
noninteger transition uniformly as the order approaches one or an integer.

\paragraph{Polynomial methods and additive functionals.}
Jiao, Venkat, Han, and Weissman~\cite{JVHW15} developed a general
method that estimates polynomial approximations in nonsmooth regions
and uses bias correction in smooth regions. Wu and Yang~\cite{WY16}
proved sharp entropy risk bounds through best polynomial approximation.
Their later work~\cite{WY19,WY20} develops the associated connections
between approximation, moment matching, and estimation of unseen
probabilities. Fukuchi and Sakuma~\cite{FS18} study general additive
functionals under squared-error loss, with rates governed by derivative
behavior. Our functional conclusions concern logarithmic accuracy and
require no differentiability under the finite-leading-power assumption.
The underlying use of polynomial estimates and moment matching is an
established methodology.

\paragraph{Approximation and statistical duality.}
Polyanskiy and Wu~\cite{PW} connect functional estimation to convex
duality and optimized testing bounds. We use a weighted form of
approximation duality that records an exact multiplicative ratio.
The identity alone does not evaluate a minimax rate: uniform variance,
likelihood tails, concentration of the realized target, and the sampling
reduction must also be controlled. Our main technical work here is the
order-uniform evaluation of these quantities for powers at the original
accuracy. In particular, a constant-factor reduction in entropy accuracy
cannot be absorbed into a sample-size constant when the order varies.

\paragraph{Organization.}
Section~\ref{sec:results} states our two main statistical results.
Section~\ref{sec:framework-main} explains our common functional
framework. Section~\ref{main:upper} gives the i.i.d. proof and its
principal quantitative estimates; Section~\ref{sec:functional-main}
explains our functional reduction. The appendices give the full versions of our count
identities, polynomial constructions, sampling bounds, lower-bounds, transfer proofs, and approximation dualities.

\section{Models and main results}
\label{sec:models}\label{app:iid}\label{sec:results}

All logarithms are natural. Write
$\Delta_k=\{p\in[0,1]^k:\sum_i p_i=1\}$; zero probabilities are
allowed. The observations are independent samples from an unknown
$p\in\Delta_k$, and the alphabet bound $k$ is known. For any target
$T$, define
\begin{equation}\label{main:minimax}
 S_T(k,\varepsilon,\zeta)=\min\left\{n\in\mathbb N_0:
 \exists\widehat T\quad\sup_{p\in\Delta_k}
 \Prob_p\{|\widehat T-T(p)|>\varepsilon\}\le\zeta\right\}.
\end{equation}
Estimators may use independent randomization. For real $\alpha>1$, set
\begin{equation}\label{main:renyi}
 P_\alpha(p)=\sum_i p_i^\alpha,\qquad
 H_\alpha(p)=-\frac{\log P_\alpha(p)}{\alpha-1},
 \qquad S_\alpha=S_{H_\alpha},
\end{equation} where $H_{\alpha}$ is the R\'enyi entropy of order $\alpha$. Jensen's inequality gives
\begin{equation}\label{intro:range}
 k^{1-\alpha}\le P_\alpha(p)\le1,\qquad
 0\le H_\alpha(p)\le\log k.
\end{equation}
Min-entropy is $H_\infty(p)=-\log\max_i p_i$, with complexity
$S_\infty$ defined by the same loss. We suppress the failure parameter
when it equals $1/3$.

\subsection{A rate uniform over all real orders above one}

Fix $0<\varepsilon<1/2$. For noninteger $1<\alpha<\log k$, put
\begin{align}
 \beta&=\alpha-1,&m&=\lfloor\alpha\rfloor,&
 \gamma&=\alpha-m,&\delta&=1-\gamma,&h&=\gamma\delta,\nonumber\\
 E&=2\varepsilon\beta,&
 \Psi&=\gamma^{1+1/\delta}\delta^{1+1/\gamma},&&&
 \mathcal Q^2&=\frac{\Psi}{\alpha^2\beta^2}
   \exp\!\left(\frac{\log k}{\alpha}-\frac{E}{h}\right).
 \label{main:root-data}
\end{align}
The letters $\gamma,\delta$ denote distances to the adjacent integers,
not failure probabilities. For $[t]_+=\max\{t,0\}$, define
\begin{equation}\label{main:root}
 X_\alpha^{\mathrm{root}}=
 \inf\left\{x\ge1:\frac{\alpha h}{\beta}x^\alpha
 \left[1+\log\frac{x(1+\log x)}{\mathcal Q^2}\right]_+\ge1\right\}.
\end{equation}
The defining function is continuous and nondecreasing, and strictly
increasing wherever positive. Thus the infimum is attained and, unless
it equals one, is the unique crossing of level one. It can be computed
by monotone search in $\log x$.

For $\log k>2$, extend this multiplier to all orders by
\begin{equation}\label{main:global-x}
 X_\alpha(k,\varepsilon)=
 \begin{cases}
 X_\alpha^{\mathrm{root}},&1<\alpha<\log k,\ \alpha\notin\mathbb N,\\
 \max\{1,(2\varepsilon r)^{-1/r}\},
     &\alpha=r\in\mathbb N,\ 2\le r<\log k,\\
 (\log k)k^{1/\alpha}/\alpha,&\alpha\ge\log k,
 \end{cases}
\end{equation}
and set
\begin{equation}\label{main:rate}
 R_\alpha(k,\varepsilon)
   =\alpha k^{1-1/\alpha}X_\alpha(k,\varepsilon).
\end{equation}
The integer multiplier is bounded above and below by positive
accuracy-dependent constants. The high-order branch gives
$R_\alpha=k\log k$ exactly; its multiplier need not be at least one.

\begin{theorem}[Uniform R\'enyi sample complexity]\label{main:theorem}
For every $0<\varepsilon<1/2$ there are constants
$c_\varepsilon,C_\varepsilon>0$ and $K_\varepsilon\in\mathbb N$,
depending only on $\varepsilon$, such that for every
$k\ge K_\varepsilon$ and every real $\alpha>1$,
\[
 c_\varepsilon R_\alpha(k,\varepsilon)
 \le S_\alpha(k,\varepsilon,1/3)
 \le C_\varepsilon R_\alpha(k,\varepsilon).
\]
Moreover, $S_\infty(k,\varepsilon,1/3)
=\Theta_\varepsilon(k\log k)$.
\end{theorem}

The constants are chosen before both $k$ and $\alpha$. The estimator
may depend on the known order; different ranges use different auxiliary
parameters. Table~\ref{tab:rates} records useful consequences.
\begin{table}[t]
\centering
\caption{Consequences at fixed additive accuracy and failure probability
$1/3$. Constants are uniform in $\alpha$ only in the first, second,
and fourth rows.}\label{tab:rates}
\begin{tabular}{@{}ll@{}}
\toprule
Order range & Sample complexity\\\midrule
$1<\alpha\le3/2$ & $\Theta_\varepsilon(k/\log k)$\\
Integer $2\le\alpha<\log k$ &
 $\Theta_\varepsilon(\alpha k^{1-1/\alpha})$\\
Fixed noninteger $\alpha>1$ &
 $\Theta_{\alpha,\varepsilon}(k/\log k)$\\
$\alpha\ge\log k$, or $\alpha=\infty$ &
 $\Theta_\varepsilon(k\log k)$\\\bottomrule
\end{tabular}
\end{table}
The root itself is needed for moving noninteger orders outside the
first row. Its exponential term retains the original accuracy:
replacing $\varepsilon$ by a constant fraction can change the rate.

\subsection{Logarithmic additive functionals}

Let $f\in C([0,1])$ satisfy $f(0)=0$ and $f(x)>0$ for $x>0$.
For known $\beta_f>0$, define
\begin{equation}\label{main:additive}
 F_f(p)=\sum_i f(p_i),\qquad
 H_{f,\beta_f}(p)=-\beta_f^{-1}\log F_f(p).
\end{equation}
The scale $\beta_f$ is independent of any leading exponent of $f$.

\begin{theorem}[Finite leading power]\label{main:functional}
Suppose $f(x)/(ax^r)\to1$ as $x\downarrow0$, for fixed $a>0$
and $r>1$. For every fixed $\beta_f,\varepsilon>0$ and
$0<\zeta<1/2$, as $k\to\infty$,
\[
 S_{H_{f,\beta_f}}(k,\varepsilon,\zeta)=
 \begin{cases}
 \Theta_{f,r,\beta_f,\varepsilon,\zeta}(k^{1-1/r}),
    &r\in\{2,3,\ldots\},\\
 \Theta_{f,r,\beta_f,\varepsilon,\zeta}(k/\log k),
    &r>1\text{ noninteger}.
 \end{cases}
\]
\end{theorem}

For example, $x^2+x^{5/2}$, $x^2e^x$, and $e^x-1-x$ have
complexity $\Theta(\sqrt{k})$, whereas $x^{3/2}(1+x)$ has complexity
$\Theta(k/\log k)$. No differentiability is required. The function,
its leading exponent, and the accuracy are fixed here; the theorem
does not claim uniform constants over changing functions or orders.
The reduction proving it is stated in
Theorem~\ref{thm:real-leading-order-transfer} independently of any
formula for the R\'enyi rate. It incurs an explicit deterministic
accuracy loss while preserving the sampling experiment exactly.

\section{The general functional framework}
\label{sec:framework-main}

Three identities explain the common structure of the proofs. They
separate deterministic approximation, sampling noise, and statistical
indistinguishability. Their quantitative evaluation for powers is the
additional step needed for Theorem~\ref{main:theorem}.

\paragraph{Polynomial approximation and exact count noise.}
In the Poisson experiment, counts are independent with
$N_i\sim\Poi(np_i)$. For a polynomial
$Q(x)=\sum_{j=0}^L c_jx^j$, factorial substitution gives
\[
 Z_Q=\sum_i\sum_{j=0}^L c_j(N_i)_j/n^j,
 \qquad (N)_j=N(N-1)\cdots(N-j+1).
\]
Its expectation is $\sum_iQ(p_i)$, and
\begin{equation}\label{main:energy}
 \Var Z_Q=\sum_i\sum_{\ell=1}^L
 \frac{p_i^\ell}{\ell!n^\ell}\{Q^{(\ell)}(p_i)\}^2.
\end{equation}
Lemma~\ref{lem:factorial-energy-new} proves the identity, including
all cross terms. Thus approximation error controls the bias, while
derivatives control the noise. If the polynomial has a linear term,
its aggregate is known: $\sum_i c_1p_i=c_1$. Replacing that term
globally by its deterministic value is particularly important near
order one.

\paragraph{Exact logarithmic accuracy.}
For a positive functional $F$ and scale $\beta>0$, set
$\eta=\tanh(\beta\varepsilon)$. If
$|\E Z-F|\le\eta F/2$ and
$\Var Z\le\zeta\eta^2F^2/4$, Chebyshev's inequality gives
$Z/F\in[1-\eta,1+\eta]$ with failure at most $\zeta$.
The centered output
\[
 -\frac{\log Z}{\beta}
 +\frac{\log(1-\eta^2)}{2\beta}
\]
then has error at most $\varepsilon$. The logarithmic half-width
is exactly $\operatorname{arctanh}(\eta)=\beta\varepsilon$.
Proposition~\ref{prop:additive-margin-new} gives the full certificate.

\paragraph{Approximation duality at the same threshold.}
Let $G$ be positive and continuous on a compact parameter space
$\Theta$, with $\min G>0$, and let $V\subset C(\Theta)$ be a
finite-dimensional space containing constants. The weighted
approximation error satisfies
\begin{equation}\label{main:duality}
 \inf_{v\in V}\left\|\frac{G-v}{G}\right\|_\infty
 =\max_{\substack{\int v\,d\omega=0\ (v\in V)\\
                   \int G\,d|\omega|\le1}}
    \int G\,d\omega.
\end{equation}
Theorem~\ref{thm:exact-multiplicative-dual} proves attainment and
shows that, when $d>0$, the normalized Jordan parts of an optimizer have matching
$V$-moments and target-expectation ratio $(1+d)/(1-d)$, where $d$
is the displayed value. Therefore the threshold
$d=\tanh(\beta\varepsilon)$ corresponds exactly to a logarithmic
separation $2\varepsilon$. The regularized version in
Theorem~\ref{thm:regularized-weighted-dual} permits controlled moment
defects.

The sampling counterpart is equally explicit. For independent
$N_e\sim\Poi(\nu_e\theta_e)$ and a positive reference $\theta^*$,
the likelihood ratios obey
\begin{equation}\label{main:kernel}
 \E_{\theta^*}L_\theta L_\vartheta
 =\exp\left\{\sum_e\frac{\nu_e
    (\theta_e-\theta_e^*)(\vartheta_e-\theta_e^*)}{\theta_e^*}\right\}.
\end{equation}
An orthonormal Poisson polynomial basis turns coefficient energy into
reference variance and moment-defect energy into likelihood overlap.
Appendix~\ref{app:duality} derives both identities and the truncation
tail. Reference variance is not a uniform variance bound, and separated
prior expectations do not ensure concentration of the realized target.
These points are verified separately in the i.i.d. proof.

\section{Proof of the uniform i.i.d. theorem}
\label{main:upper}

We describe the construction and its quantitative obligations before
giving the complete proofs in the appendices. Write
$\beta=\alpha-1$ and $P=P_\alpha(p)$.

\subsection{Upper bound: approximation with a sampling margin}

The noninteger estimator targets the exact interval
\begin{equation}\label{intro:accuracy-band}
 e^{-2\varepsilon\beta}P\le\widehat P\le P.
\end{equation}
On this event the output
\begin{equation}\label{intro:entropy-output}
 \widehat H=-\beta^{-1}\log\widehat P-\varepsilon
\end{equation}
has error at most $\varepsilon$. The construction must leave positive
margins at both ends of this interval, so that random fluctuations do
not consume the deterministic approximation tolerance.

\paragraph{Orders bounded away from one.}
For noninteger $3/2\le\alpha<\log k$, use main counts
$N_i\sim\Poi(np_i)$ and independent pilot counts
$Y_i\sim\Poi(c_pnp_i)$, with
$n=D\alpha k^{1-1/\alpha}X$ and $X=X_\alpha^{\mathrm{root}}$.
Set $\chi=\min\{\beta,1\}$ and choose
\[
 L\asymp_\varepsilon\alpha(1+\log X),\quad b=c_bL/n,
 \quad v=(\beta/\alpha)X^{-\alpha},\quad
 u=\min\{v+h,u_0\}.
\]
The exact two-arm envelope of Lemma~\ref{oc:envelope} chooses a
lower cutoff $a$ and a deterministic tolerance
$E_0=E-\min\{v+C_hh,E/2\}$. The root calibration ensures
$a\le c_\varepsilon\beta/(\alpha k)$.

Appendix~\ref{sec:coupled-bias-bound} constructs an explicit
polynomial $Q(x)=\sum_{j=0}^d q_jx^j$, of degree $d=O(L)$,
from two endpoint profiles. Reversing one profile makes its fractional
power cancel against $p^\alpha$, so both profiles yield ordinary
polynomials in $p$. For the larger probabilities use the analytic
count transform
$\mathsf H_\alpha(j)=\Gamma(j+1)/\Gamma(j+1-\alpha)$,
interpreted with the entire reciprocal Gamma function. It can be
signed at small counts; those values must be included in the error
analysis. With $I_i=\one\{Y_i\le c_pB\alpha\}$, define
\begin{equation}\label{main:hybrid}
 \widehat P=2E_aP_*+q_1+
 \sum_i\left[I_i(\widehat Q_i-q_1N_i/n)
 +(1-I_i)(\widehat G_i-q_1N_i/n)\right],
\end{equation}
where $P_*=k^{1-\alpha}$,
$\widehat Q_i=\sum_jq_j(N_i)_j/n^j$,
$\widehat G_i=\omega n^{-\alpha}\mathsf H_\alpha(N_i)$,
and $E_a=\chi d_\varepsilon C_0^m(ka)^\beta$.
The known coefficients, damping factor $\omega$, and constants are
specified in Appendix~\ref{ped:upper-calibration}. The deterministic
$q_1$ in \eqref{main:hybrid} replaces the random global linear total.

\begin{lemma}[Upper-bound budgets]\label{main:upper-budget}
The design constants can be chosen in terms of $\varepsilon$ alone
so that \eqref{main:hybrid} satisfies
\eqref{intro:accuracy-band} with probability at least $3/4$,
uniformly over $p\in\Delta_k$ and noninteger
$3/2\le\alpha<\log k$, for all sufficiently large $k$.
\end{lemma}
\begin{proof}[Proof structure]
The deterministic construction and the correctly selected heavy
aggregate leave margins comparable to
\[
 \mathcal M=E_aP_*+e^{-E}uP+
 \chi e^{-E}\sum_i p_i^\alpha\min\{Lnp_i,1\}.
\]
The derivative energy \eqref{main:energy} bounds the conditional
polynomial noise by a small multiple of $\mathcal M^2$.
The global affine replacement obeys the same variance budget.
Pilot misclassification and the lower cutoff contribute expected
absolute error at most a small multiple of $\mathcal M$.
The heavy-count norm estimate supplies the remaining concentration
event. Assigning failure $1/16$ to each of these four contributions
proves the assertion by Chebyshev, Markov, and a union bound.
The constructions and all four quantitative estimates are proved in
Appendices~\ref{ped:upper-calibration} and
\ref{sec:tutorial-sampling}; Theorem~\ref{cu:theorem} assembles them.
\end{proof}

\paragraph{Orders approaching one.}
For $1<\alpha\le3/2$, the root gives, uniformly,
\begin{equation}\label{main:nearone-rate}
 X\asymp_\varepsilon k^{1/\alpha}/\log k,
 \qquad R_\alpha\asymp_\varepsilon k/\log k.
\end{equation}
We use separate fixed choices of the auxiliary parameters in this
range. The localized polynomial is anchored to the tangent of
$p^{1+\beta}$ and has a factor $\beta$ in every nonlinear
coefficient. After the global linear replacement, the complete variance,
including the pilot and signed low-count corrections, is
$o(\beta^2P^2)$. Its bias is an arbitrarily small multiple of
$\beta P$, and its center is at distance at least
$c_\varepsilon\beta P$ from both ends of
\eqref{intro:accuracy-band}. Chebyshev proves the upper bound.
Appendix~\ref{sec:near-one-proof} gives the construction and proves the
exact affine cancellation.

\paragraph{Integer and high orders.}
For integers $2\le r<\log k$, use
$Z_r=n^{-r}\sum_i(N_i)_r$. Its mean is $P_r$ and
\[
 \frac{\Var Z_r}{P_r^2}
 \le\sum_{j=1}^r\binom rj^2j!
          (nP_r^{1/r})^{-j}.
\]
At $n=C_\varepsilon r k^{1-1/r}$, the correct target is
\begin{equation}\label{main:integer-band}
 e^{-\varepsilon(r-1)}P_r\le Z_r\le e^{\varepsilon(r-1)}P_r,
\end{equation}
rather than a fixed relative-variance bound. The proof splits
coordinates according to whether $np_i$ exceeds a fixed multiple of
$r$, combines count concentration with factorial variance, and controls
both tails (Appendix~\ref{tutorial:integer}).
For $\alpha\ge\log k$, simultaneous count bounds at
$n=C_\varepsilon k\log k$ give a relative-plus-$1/k$ approximation
to each $p_i$. A nonnegative capped-order Gamma transform and an
$\ell_\alpha$ norm bound then estimate $H_\alpha$; the maximum
coordinate gives min-entropy. Proposition~\ref{prop:iid-high}
proves both guarantees.

All these count estimators can be implemented from a fixed sample pool.
Draw a Poisson total of mean $(1+c_p)n$, split it into main and pilot
counts, and use a pool of twice that size, with a fallback on overflow.
The overflow probability is exponentially small. The ideal and implemented
experiments agree off overflow, so no conditioning destroys independence.
Appendix~\ref{sec:count-foundations} gives the coupling.

\subsection{Lower bound: positive priors and observable counts}
\label{main:lower}

\paragraph{The collision obstruction.}
For $\alpha\ge3/2$, a uniformly located heavier coordinate gives
$\Omega_\varepsilon(\alpha k^{1-1/\alpha})$ below saturation and
$\Omega_\varepsilon(k\log k)$ for high orders and min-entropy.
The Poisson likelihood identity \eqref{main:kernel} controls the
mixture distance. This suffices when $X$ is bounded; the unbounded-root
case requires a stronger approximation obstruction.

\paragraph{Constructing positive priors.}
Finite interpolation produces two signed factors, with exponents
$\gamma$ and $1-\gamma$. Splitting their tensor product by parity
of signs, applying a change of variables, and completing at zero
gives genuine laws $U_0,U_1\in\{0\}\cup[1,B]$ with
$\E U_0=\E U_1=1$. If $A_i=\E U_i^\alpha$, their exact ratio is
\begin{equation}\label{main:prior-gap}
 \frac{A_1}{A_0}=\frac{1+R_-R_+}{R_-+R_+},
\end{equation}
where $R_-,R_+$ are the negative masses of the two normalized
factors. Integer moments cancel through the selected degree, except
possibly at the integer nearest $\alpha$; the first moment always
matches exactly. Appendix~\ref{sec:common-noninteger-priors} gives
the nodes, weights, positivity, and moment identities.

\paragraph{Normalization is part of the lower bound.}
Draw $Q\sim\Poi(b)$ weights with $b=k/2$, and put
$S=b^{-1}\sum_j U_j$, $V=b^{-1}\sum_j U_j^\alpha$,
and $W=V-A_iS$. On $S>0$ and $Q\le k$, normalize to
$p_j=U_j/(bS)$. Its entropy center is
$h_i=\log b-(\log A_i)/\beta$, and
\begin{equation}\label{main:normalization}
 \beta(H_\alpha(p)-h_i)
 =\beta\log S-\log\left(1+\frac{W}{A_iS}\right).
\end{equation}
The residual $U_i^\alpha-A_iU_i$ contains a factor $\beta$,
which preserves concentration near order one.

Give each weight an independent count of mean $\theta U_i$.
The numbers observed exactly $j$ times are independent Poisson
variables of means $bq_{ij}$, where
$q_{ij}=\E[e^{-\theta U_i}(\theta U_i)^j/j!]$ for $j\ge1$.
Their total variation is at most $\sqrt{bD}$, with
\[
 D=\sum_{j\ge1}\frac{(q_{1j}-q_{0j})^2}{q_{1j}+q_{0j}}.
\]
A common random labeling and ordering kernel recovers the labeled
sample experiment. Conditional on the weights and observation length,
the ordered labels are i.i.d. from the normalized vector. The transfer
accounts separately for alphabet overflow and insufficient observations.

\paragraph{Calibration at the original accuracy.}
For $X$ above a fixed threshold, put $T=\alpha\log X$,
$v=(\beta/\alpha)e^{-T}$, and $w=h+v$. Choose
$L\asymp_\varepsilon(T+\alpha)$,
$\theta=c_\varepsilon\alpha k^{-1/\alpha}X$, and
$B=c_mL/\theta$. The root equality gives the exact balance
\[
 \frac{\log k}{\alpha}
 =\frac{E-v}{h}+\log X+\log(1+\log X)
       +1-\log\Psi+2\log\alpha+2\log\beta.
\]
It permits separation $\log(A_1/A_0)\ge E+w$ while keeping
$kD$ small. Matched moments cancel in the mixed-Poisson expansion;
actual atoms in the priors bound the denominator for the one possible
moment defect. Normalization fluctuations fit inside the extra margin
$w$, uniformly also in the degree-zero near-one case.
Theorem~\ref{cal:thm:lower} proves these estimates and the testing
contradiction at $n\asymp_\varepsilon R_\alpha$.

\subsection{Completion for general real order}
\label{sec:conclusions}
\begin{proof}[Proof of Theorem~\ref{main:theorem}]
Lemma~\ref{main:upper-budget}, Theorem~\ref{nuc:thm:upper},
Theorem~\ref{tutorial:integer-raw}, and
Proposition~\ref{prop:iid-high} give the four upper-bound ranges.
Theorem~\ref{cal:thm:lower} and the hidden-coordinate construction
give the corresponding lower bounds. Take the minimum of their
positive lower constants, the maximum of their upper constants, and
the maximum of their alphabet cutoffs. Every choice depends only on
$\varepsilon$, giving the asserted uniform quantifiers.
\end{proof}
For fixed noninteger $\alpha$, write
$\mathcal Q^2=c_{\alpha,\varepsilon}k^{1/\alpha}$ and let
$y(1+\log y)=\mathcal Q^2/e$. The root criterion is zero at $y$
and exceeds one at $2y$ for large $k$. Thus
$X\asymp_{\alpha,\varepsilon}k^{1/\alpha}/\log k$ and
$R_\alpha\asymp_{\alpha,\varepsilon}k/\log k$.

\section{Why the leading power determines the functional rate}
\label{sec:functional-main}

The transfer argument is simpler than the uniform-order calculation.
Fix $0<\xi<\beta_f\varepsilon$ and choose $\tau>0$ so that
\[
 e^{-\xi}ax^r\le f(x)\le e^\xi ax^r\qquad(0\le x\le\tau).
\]
For the upper bound, first estimate $F_f$ to a fixed absolute
tolerance using a number of samples independent of $k$.
Lemma~\ref{lem:dimension-free-additive} does this by approximating
the continuous function $f(x)/x$ by a polynomial. If the estimated
sum is substantial, this absolute estimate already gives logarithmic
accuracy. Otherwise positivity implies that every coordinate is below
$\tau$, and the leading-power band reduces the task to estimating
$H_r$ with accuracy $(\beta_f\varepsilon-\xi)/(r-1)$.

For the lower bound, choose a fixed integer $m\ge1/\tau$ and replace
each symbol of $p\in\Delta_{\lfloor k/m\rfloor}$ by $m$ copies,
each of probability $p_i/m$. A sample from this replicated distribution
is exactly a sample from $p$ together with an independent uniform label.
Discarding that label reverses the operation. Every probability now
lies below $\tau$, and
\[
 \left|\log F_f(q)-\log\bigl(am^{1-r}P_r(p)\bigr)\right|\le\xi.
\]
Consequently an estimator of the logarithmic functional would estimate
$H_r$ with accuracy $(\beta_f\varepsilon+\xi)/(r-1)$ and the
same sample size. Theorem~\ref{thm:real-leading-order-transfer}
proves both reductions with explicit failure accounting.

\begin{proof}[Proof of Theorem~\ref{main:functional}]
Apply those reductions with, for example,
$\xi=\beta_f\varepsilon/2$. Use
Theorem~\ref{fixed:thm:renyi} for noninteger $r$ and the direct
integer bounds for integer $r$. These bounds cover every fixed
positive accuracy and confidence, as required by both reductions.
A fixed replication changes neither $k^{1-1/r}$ nor $k/\log k$
by more than a constant factor, and the auxiliary sample cost is constant.
\end{proof}

The finite-leading-power condition has substantive content. For every
$s>0$, a smooth, convex, increasing, unit-Lipschitz function proportional
to $e^{-A/x^s}$ has logarithmic estimation complexity
$\Omega(k^{2s})$, although all its derivatives vanish at zero.
Appendix~\ref{app:additive} proves this and a single smooth example
with an exponential lower bound. Thus continuity, smoothness, or
convexity alone cannot replace the leading-power hypothesis.
\label{main:end}
\clearpage\appendix
\section{The counting identities used throughout the proof}
\label{app:counts}\label{sec:count-foundations}

The upper and lower bounds use the same observation model: independent
Poisson counts. This section derives the count identities before they are
used. There are three distinct operations to keep in mind. A polynomial
in a probability becomes an unbiased statistic by replacing powers with
falling factorials. A real power can also be estimated by a Gamma ratio,
with a small, explicit bias. Finally, the likelihood ratios of Poisson
counts have an exact overlap formula, which will make hard distributions
statistically indistinguishable.

\subsection{Poissonization and factorial counts}
Suppose that the number of observations is a random variable
$M\sim\Poi(n)$, independent of an iid sequence with distribution
$p=(p_1,\ldots,p_k)$. Write $N_i$ for the number of occurrences of symbol
$i$ among the first $M$ observations. Then
\[
 N_1,\ldots,N_k\text{ are independent},\qquad N_i\sim\Poi(np_i).
\]
One can see independence directly from their joint probability-generating
function. For $z_1,\ldots,z_k\ge0$,
\[
 \E\prod_{i=1}^kz_i^{N_i}
 =\E\left(\sum_i p_i z_i\right)^M
 =\exp\left(n\sum_i p_i(z_i-1)\right)
 =\prod_i\exp(np_i(z_i-1)).
\]
The final expression is a product of Poisson generating functions.
Here $n$ is a sample-size \emph{mean}; it need not be an integer. Later,
a Poisson tail bound converts the construction into an estimator using a
fixed number of observations. An independent pilot sample is obtained by
independent thinning, and therefore does not change these identities.

\paragraph{Fixed and Poisson sample sizes have different denominators.}
If the sample size is a fixed integer $N$, then counting ordered choices
of $j$ distinct positions gives $\E(N_i)_j=(N)_j p_i^j$.
The unbiased denominator is $(N)_j$. In the Poisson experiment it is
$n^j$, because $n$ is the Poisson mean. The count formulas below use
the Poisson convention unless a fixed sample size is explicitly stated.

\paragraph{How to implement Poissonization from a fixed sample pool.}
With $\lceil2n\rceil$ available observations, independently draw
$M\sim\Poi(n)$ and use the first $M$ observations if available.
Couple this procedure to an infinite i.i.d. sequence. The ideal Poisson
estimator and the implemented estimator agree unless $M>\lceil2n\rceil$.
Since $\E2^M=e^n$, Markov's inequality yields
\[
 \Prob\{M>\lceil2n\rceil\}\le e^{-(2\log2-1)n}.
\]
An ideal failure probability $\zeta$ therefore becomes at most
$\zeta+e^{-(2\log2-1)n}$. Independent main and pilot lengths of means
$n$ and $c_pn$ have a total of mean $(1+c_p)n$; allocate successive
disjoint blocks and use the same argument for their sum. Independence
is asserted in the ideal experiment. We handle the finite-pool overflow
by coupling, rather than conditioning the counts on no overflow.

For integers $j\ge0$, define
\[
 (N)_j=N(N-1)\cdots(N-j+1),\qquad (N)_0=1.
\]
The product vanishes when $j>N$. If $N\sim\Poi(\lambda)$, then
\begin{equation}\label{main:factorial}
 \E(N)_j
 =\sum_{q=j}^{\infty}\frac{q!}{(q-j)!}
             e^{-\lambda}\frac{\lambda^q}{q!}
 =\lambda^j\sum_{v=0}^{\infty}e^{-\lambda}\frac{\lambda^v}{v!}
 =\lambda^j.
\end{equation}
Consequently $(N_i)_j/n^j$ estimates $p_i^j$ without bias. Ordinary
powers do not have this property: $\E(N_i/n)^2=p_i^2+p_i/n$.

\subsection{From a polynomial approximation to a count estimator}
\label{sec:additive-complete}
Let $Q(x)=\sum_{j=0}^{L}c_jx^j$ approximate $x^\alpha$ on the
probability range of interest. Replacing each $x^j$ by $(N_i)_j/n^j$
produces an estimator whose mean is exactly $Q(p_i)$. Thus its bias is a
\emph{deterministic approximation error}. Its variance is also explicit;
the next identity is why derivatives of the approximation polynomial
appear in the later proof. In the following lemma, $\G_j(N)=(N)_j$ for
integer $j$, with $\G_0(N)=1$.

\begin{lemma}[Factorial derivative energy]\label{lem:factorial-energy-new}
Let $N_i\sim\Poi(n p_i)$ independently, and let
$Q(x)=\sum_{j=0}^L c_jx^j$. The statistic
$Z_Q=\sum_i\sum_{j=0}^L c_j\G_j(N_i)/n^j$, with $\G_0=1$, satisfies
\[
 \E Z_Q=\sum_iQ(p_i),\qquad
 \Var Z_Q=\sum_i\sum_{\ell=1}^L
       \frac{p_i^\ell}{\ell!n^\ell}\{Q^{(\ell)}(p_i)\}^2.
\]
\end{lemma}
\begin{proof}
For integers $r,s\ge0$, counting pairs of ordered injective selections
by the number of shared indices gives the polynomial identity
\[
 (N)_r(N)_s=\sum_{\ell=0}^{\min(r,s)}
       \binom r\ell\binom s\ell\ell!(N)_{r+s-\ell}.
\]
Taking the Poisson expectation and subtracting the $\ell=0$ term
gives the covariance. Expanding $Q$ and collecting the terms with
the same $\ell$ proves the formula. Independent coordinates add their
variances. For a single coordinate, the term with a fixed $\ell$ is
\[
 \sum_{r,s\ge\ell}c_rc_s
       \binom r\ell\binom s\ell\ell!
       \frac{p^{r+s-\ell}}{n^\ell}
 =\frac{p^\ell}{\ell!n^\ell}
       \left(\sum_{r\ge\ell}c_r(r)_\ell p^{r-\ell}\right)^2.
\]
The expression in parentheses is $Q^{(\ell)}(p)$. This proves the
claimed derivative formula, including all cross terms between the
polynomial coefficients.
\end{proof}

\paragraph{The order-two specialization.}
Putting $Q(x)=x^2$ in the derivative-energy identity recovers
\eqref{intro:collision-variance}: the first derivative contributes
$4P_3/n$ and the second contributes $2P_2/n^2$. Thus the elementary
collision calculation is a special case of the general polynomial formula.

\subsection{Why power-sum accuracy must be calibrated before taking a logarithm}
For R\'enyi entropy the target is $H_\alpha(p)=-\beta^{-1}\log P_\alpha(p)$,
where $\beta=\alpha-1$. A power-sum error of relative size $O(\beta\varepsilon)$
is sufficient near order one, but this informal statement loses constants
that matter when the order varies. We will instead use an exact interval.
If
\[
 e^{-2\varepsilon\beta}P_\alpha(p)\le\widehat P\le P_\alpha(p),
\]
then $-\beta^{-1}\log\widehat P-\varepsilon$ differs from
$H_\alpha(p)$ by at most $\varepsilon$. The following symmetric version
is also useful. The following statement applies to every positive additive sum
$F(p)=\sum_i f(p_i)$ and every scale $\beta>0$. The letter $\delta$ in this proposition denotes a failure
probability, not a distance to an integer.

\begin{proposition}[Exact logarithmic accuracy bound]
\label{prop:additive-margin-new}
Let $F(p)=\sum_i f(p_i)>0$, $\beta>0$, and
$\eta=\tanh(\beta\eps)$. If
\[
 \left|\sum_i Q(p_i)-F(p)\right|\le\frac\eta2F(p),
 \qquad
 \Var Z_Q\le\frac{\delta\eta^2}4F(p)^2
\]
uniformly in the parameter class, then
\[
 \widehat H=-\frac{\log Z_Q}{\beta}
             +\frac{\log(1-\eta^2)}{2\beta}
\]
estimates $-\beta^{-1}\log F$ within $\eps$ with failure at most
$\delta$; define it arbitrarily when $Z_Q\le0$.
The same conclusion holds for any real statistic with the stated
bias and variance, including a statistic defined using an independent pilot.
\end{proposition}
\begin{proof}
Chebyshev gives $Z_Q/F\in[1-\eta,1+\eta]$ except with probability
$\delta$. This interval is positive. The midpoint of its logarithms
is $\frac12\log(1-\eta^2)$, and its logarithmic half-width is
$\operatorname{artanh}(\eta)=\beta\eps$. For a pilot, apply this
argument conditionally and account separately for the pilot's failure.
No replacement of $\eps$ by a smaller accuracy occurs.
\end{proof}

\subsection{Extending factorial counts to real powers: the nonnegative transform}
\label{sec:positive-count}
A noninteger power $x^s$ cannot be represented by a finite polynomial
exactly. The Gamma function gives a second extension of factorial counts.
Recall that $\Gamma(j+1)=j!$ for nonnegative integers $j$ and
$\Gamma(x+1)=x\Gamma(x)$. For a real $s\ge1$ and an integer $j\ge0$, define
\begin{equation}\label{eq:gamma-count}
 \G_s(j)=
 \begin{cases}
 \Gamma(j+1)/\Gamma(j+1-s),&j\ge\lceil s\rceil,\\
 0,&j<\lceil s\rceil,
 \end{cases}
 \qquad \g_s(j)=\G_s(j)^{1/s}.
\end{equation}
At an integer $r$, $\G_r(j)=(j)_r$ is the falling factorial.

\begin{lemma}[Mean and deterministic approximation]\label{lem:gamma}
For every $s\ge1$ and integer $j\ge0$,
\begin{equation}\label{eq:gamma-sandwich}
 (j-s)_+\le\g_s(j)\le j,\qquad |\g_s(j)-j|\le s.
\end{equation}
If $N\sim\Poi(\lambda)$, then $\E\G_r(N)=\lambda^r$ for integer $r$.
For noninteger $s=m+\gamma$, $0<\gamma<1$, and $\delta=1-\gamma$,
\begin{equation}\label{eq:gamma-mean}
 \E\G_s(N)=\frac{\lambda^s}{\Gamma(\delta)}
                  \int_0^\lambda t^{\delta-1}e^{-t}\,dt.
\end{equation}
\end{lemma}
\begin{proof}
For $j<\lceil s\rceil$, the sandwich is immediate. Otherwise, write
the logarithm of the Gamma ratio as the integral of the increasing
function $\psi=(\log\Gamma)'$ over $[j+1-s,j+1]$.
The average of an increasing function over a terminal interval decreases
when that interval is lengthened. Since $s\ge1$, this average is at most
the average over $[j,j+1]$, which equals $\log j$.
Convexity of $\log\Gamma$ also gives
$\psi(t)\ge\log(t-1)$ for $t>1$ by comparison with the preceding
unit secant. If $j>s$, the integral average is consequently at least
$\log(j-s)$. The boundary case $j=s$ has a zero lower bound.
Exponentiating proves \eqref{eq:gamma-sandwich}.

For integers, factorial unbiasedness follows by shifting the Poisson
series. For nonintegers, cancellation of $j!$ gives
\[
 \E\G_s(N)=e^{-\lambda}\lambda^{m+1}
       \sum_{\ell\ge0}\frac{\lambda^\ell}{\Gamma(\ell+1+\delta)}.
\]
The function
$e^{-\lambda}\lambda^\delta\sum_{\ell\ge0}
\lambda^\ell/\Gamma(\ell+1+\delta)$ vanishes at zero and has
derivative $e^{-\lambda}\lambda^{\delta-1}/\Gamma(\delta)$.
Termwise differentiation is justified by local uniform convergence;
adjacent terms cancel using the Gamma recurrence. This proves the mean.
\end{proof}

For a noninteger $s$, the integral in \eqref{eq:gamma-mean}, divided by
$\Gamma(\delta)$, is a number in $[0,1]$ that approaches one as
$\lambda$ grows. Thus $\G_s(N)$ is downward biased for $\lambda^s$,
but the bias becomes small for sufficiently frequent symbols. The
pointwise sandwich \eqref{eq:gamma-sandwich} will be particularly useful
at high orders: taking an $s$th root turns a nonlinear statistic into a
quantity differing from the count itself by at most $s$.

\subsection{An analytic count extension with exact order-one cancellation}
\label{ag:sec:extension}
The nonnegative count family $\G_\alpha$ has a cutoff jump when
$\alpha$ crosses an integer. A signed extension removes that jump
and retains the exact factorial members. This section supplies a common
heavy-count construction. The corresponding light-coordinate approximation
and the complete statistical bound are proved in
Sections~\ref{sec:coupled-bias-bound} and~\ref{cu:section}.

The distinction between the next transform and $\G_a$ is important.
The transform $\G_a$ truncates low counts to zero; the analytic transform
below retains them, even when they have a negative sign. This makes its
dependence on the order continuous. For every real $a$ and integer
$j\ge0$, define
\[
 \mathsf H_a(j)=\frac{\Gamma(j+1)}{\Gamma(j+1-a)},
 \qquad \Delta v(j)=v(j+1)-v(j).
\]
The reciprocal function $1/\Gamma(z)$ is entire. In particular its value
is zero at every nonpositive integer, so the ratio is well-defined even
where a literal denominator $\Gamma(j+1-a)$ has a pole.
In particular $\mathsf H_r(j)=(j)_r$ for every nonnegative integer $r$.

\begin{proposition}[Exact analytic count identities]\label{ag:prop:identities}
For $\lambda>0$, let $m_a(\lambda)=\E\mathsf H_a(N)$ with
$N\sim\Poi(\lambda)$. For every real $a$,
\begin{align}
 \Delta^\ell\mathsf H_a(j)&=(a)_\ell\mathsf H_{a-\ell}(j),
       &&\ell\ge0,\label{ag:eq:difference}\\
 m_a(\lambda)&=\lambda^a-
       \frac1{\Gamma(-a)}\int_1^\infty
          e^{-\lambda t}t^{-a-1}\,dt,\label{ag:eq:mean}\\
 \Var\mathsf H_a(N)&=
       \sum_{\ell\ge1}\frac{\lambda^\ell}{\ell!}
                   (a)_\ell^2m_{a-\ell}(\lambda)^2.
       \label{ag:eq:variance}
\end{align}
The variance series converges. At $a=r\in\mathbb N$ it terminates
at $r$, and the mean is exactly $\lambda^r$. For $a>0$,
\begin{equation}\label{ag:eq:bias}
 |m_a(\lambda)-\lambda^a|
 \le \frac{\Gamma(a+1)|\sin(\pi a)|}{\pi}
           e^{-\lambda}\min\{a^{-1},\lambda^{-1}\}.
\end{equation}
\end{proposition}
\begin{proof}
The Gamma recurrence gives
$\Delta\mathsf H_a=a\mathsf H_{a-1}$; induction proves
\eqref{ag:eq:difference}, including the removable zeros.
For $a<0$ the beta integral gives
\[
 \mathsf H_a(j)=\frac1{\Gamma(-a)}
       \int_0^1t^j(1-t)^{-a-1}\,dt.
\]
Take the Poisson expectation and substitute $u=1-t$. Subtracting
the tail from the complete Gamma integral proves \eqref{ag:eq:mean}
for $a<0$. For fixed $\lambda>0$ both sides are entire in complex
$a$. Indeed, the Gamma ratio grows at most polynomially in $j$ on
each compact set of $a$, so its Poisson expectation series is locally
uniformly convergent; the integral on $[1,\infty)$ is also locally
uniformly convergent. Analytic continuation proves the formula for
all $a$. Its value at $\lambda=0$ is instead read directly as
$m_a(0)=1/\Gamma(1-a)$. Bounding the integral both by
$e^{-\lambda}/a$ and by $e^{-\lambda}/\lambda$, and using reflection,
proves \eqref{ag:eq:bias}.

It remains to prove the variance identity. We use an orthogonal
polynomial basis for the Poisson distribution, in the same way that
Fourier coefficients decompose a squared norm. No approximation is
involved in this step. Define the Charlier polynomials by
\[
 e^{-\lambda t}(1+t)^N
       =\sum_{\ell\ge0}C_\ell(N;\lambda)t^\ell/\ell!.
\]
The joint generating function gives
$\E C_\ell C_q=\one\{\ell=q\}\ell!\lambda^\ell$.
Repeated Poisson size bias gives
\[
 \E[v(N)C_\ell(N;\lambda)]
      =\lambda^\ell\E\Delta^\ell v(N).
\]
These polynomials are complete in $L^2(\Poi(\lambda))$: if $v$ is
orthogonal to all of them, then the entire function
$e^{-\lambda(z-1)}\E[v(N)z^N]$ has every derivative zero at $z=1$.
It therefore vanishes, and its coefficients imply $v(j)=0$ for all
$j$. Entirety follows by Cauchy--Schwarz and Poisson exponential
moments. Thus Parseval gives, whenever $v(N)$ is square integrable,
\begin{equation}\label{ag:eq:charlier}
 \Var v(N)=\sum_{\ell\ge1}\frac{\lambda^\ell}{\ell!}
                      (\E\Delta^\ell v(N))^2.
\end{equation}
The polynomial growth of $\mathsf H_a(j)$ ensures square
integrability. Substitution of \eqref{ag:eq:difference} proves
\eqref{ag:eq:variance}.
\end{proof}

\begin{lemma}[Poisson Poincar\'e from the same count expansion]
\label{ag:lem:poincare}
If $N\sim\Poi(\lambda)$ and $v(N)$ is square integrable, then
\[
 \Var v(N)\le\lambda\E\{v(N+1)-v(N)\}^2.
\]
\end{lemma}
\begin{proof}
If the right side is infinite the assertion is immediate. Otherwise
apply the complete Charlier expansion to $\Delta v$. Its squared
norm, multiplied by $\lambda$, equals
\[
 \sum_{\ell\ge1}\ell\frac{\lambda^\ell}{\ell!}
                       (\E\Delta^\ell v(N))^2,
\]
which dominates \eqref{ag:eq:charlier} term by term.
\end{proof}

\paragraph{Removing noise from a known linear term.}
Because $\sum_i p_i=1$, the aggregate of a linear function $cp_i$ is the
known constant $c$. It is wasteful to estimate that constant by
$c\sum_iN_i/n$. The following identity describes exactly how removing
an affine term changes the variance. Only the first term of the Charlier
expansion is affected.
For constants $c,d$, the exact residual formula is
\begin{equation}\label{ag:eq:anchored-energy}
\begin{split}
 \Var\{\mathsf H_a(N)-cN-d\}
   ={}&\lambda\{a m_{a-1}(\lambda)-c\}^2\\
     &+\sum_{\ell\ge2}\frac{\lambda^\ell}{\ell!}
                         (a)_\ell^2m_{a-\ell}(\lambda)^2.
\end{split}
\end{equation}
Here the affine term must be removed globally, or compensated in the
other part of a pilot partition: a random subset of the identity
$\sum_i p_i=1$ is not known.

\paragraph{Why this matters near $\alpha=1$.}
Write $\alpha=1+\beta$. As $\beta$ decreases to zero,
$\mathsf H_{1+\beta}(N)$ approaches $N$, whose aggregate mean is
already known. After removing a suitable linear term, the remaining
fluctuation should be proportional to $\beta$. The next proposition
shows that its variance is proportional to $\beta^2$, uniformly in
small positive $\beta$. This matches the shrinking power-sum accuracy
required before division by $\beta$ in the entropy formula.

\begin{proposition}[Uniform divided-difference noise at order one]
\label{ag:prop:order-one}
Let $0<\beta\le1/2$, $N\sim\Poi(\lambda)$, $\lambda\ge1$, and $B>0$.
Then, for an absolute constant $C$,
\begin{equation}\label{ag:eq:order-one-var}
 \Var\left\{\frac{\mathsf H_{1+\beta}(N)-B^\beta N}{\beta}\right\}
 \le C\lambda(\lambda^{2\beta}+B^{2\beta})
                    \{1+\log^2(\lambda/B)\}.
\end{equation}
The bias of its numerator relative to
$\lambda^{1+\beta}-B^\beta\lambda$ is at most
$C\beta e^{-\lambda}/\lambda$.
\end{proposition}
\begin{proof}
Uniformly for $j\ge0$ and $0\le\beta\le1/2$,
\[
 |\mathsf H_\beta(j)-(j+1)^\beta|
       \le C\beta(j+1)^\beta.
\]
Indeed, the logarithm of their ratio is the integral over
$0\le t\le\beta$ of $\psi(j+1-t)-\log(j+1)$, where
$\psi=\Gamma'/\Gamma$. This integrand is uniformly bounded:
the recurrence $\psi(x+1)=\psi(x)+1/x$, a bound on the compact
interval $[1/2,3/2]$, and comparison of harmonic sums with the
logarithm prove the assertion. Also the mean value theorem gives
\[
 |(j+1)^\beta-B^\beta|
  \le\beta\{(j+1)^\beta+B^\beta\}|\log((j+1)/B)|.
\]
Consequently
\begin{equation}\label{ag:eq:order-one-difference}
 |(1+\beta)\mathsf H_\beta(j)-B^\beta|
  \le C\beta\{(j+1)^\beta+B^\beta\}
                         \{1+|\log((j+1)/B)|\}.
\end{equation}
Apply Lemma~\ref{ag:lem:poincare} and
$\Delta\mathsf H_{1+\beta}=(1+\beta)\mathsf H_\beta$.
For $V=(N+1)/\lambda$, uniformly over $0\le q\le1$ and
$\lambda\ge1$,
\[
 \E V^q+\E[V^q(\log V)^2]\le C.
\]
On $V\ge1$ this follows from
$V^q(\log V)^2\le C(1+V^3)$ and the first three Poisson moments.
On $1/2\le V\le1$ it is immediate. On $V<1/2$, use
$|\log V|\le\log\lambda$ and the Poisson Chernoff bound
$\Prob(N<\lambda/2)\le e^{-c\lambda}$; the bounded range
$1\le\lambda\le4$ is immediate. Expand
$\log((N+1)/B)=\log(\lambda/B)+\log V$ in
\eqref{ag:eq:order-one-difference} to obtain
\eqref{ag:eq:order-one-var}. Finally
$\Gamma(2+\beta)|\sin(\pi\beta)|/\pi\le C\beta$;
\eqref{ag:eq:bias} proves the bias assertion.
\end{proof}

In probability units, with $a=1+\beta$, $\lambda=np$ and $B=nb$,
Proposition~\ref{ag:prop:order-one} gives
\[
 \Var\{n^{-a}\mathsf H_a(N)-b^\beta N/n\}
 \le C\beta^2 n^{-2a}\lambda
       (\lambda^{2\beta}+B^{2\beta})\{1+\log^2(\lambda/B)\}.
\]
Thus the common family preserves the full $\beta^2$ cancellation;
an order-independent pilot-tail estimate alone would not do so.

\paragraph{Controlling the negative values.}
The signed transform will be used on coordinates that the pilot declares
frequent. Nevertheless their main count can occasionally be small. We
must bound the resulting negative contribution, not assume it is absent.
The following lemma supplies an exponentially small bound.

\begin{lemma}[The signed low-count correction]
\label{ag:lem:negative}
Let $a=m+\gamma>1$, $m=\lfloor a\rfloor$, $0<\gamma<1$,
$h=\gamma(1-\gamma)$, and $\lambda\ge2(m+1)$. Then
\[
 \E\bigl[|\mathsf H_a(N)|\one\{N\le m-1\}\bigr]
       \le Ch e^{-\lambda}\lambda^{m-1}.
\]
Every negative value is included in this range; at integer orders
there are no negative values.
\end{lemma}
\begin{proof}
For $j\le m-1$, the reciprocal-Gamma recurrence gives
\[
 \frac1{|\Gamma(j+1-a)|}
 =\frac{\gamma}{\Gamma(1-\gamma)}
       \prod_{\ell=1}^{m-j-1}(\ell+\gamma)
 \le Ch(m-j)!.
\]
Here $\gamma/\Gamma(1-\gamma)=h/\Gamma(2-\gamma)$ and
$\Gamma$ is bounded away from zero on $[1,2]$. Taking the
Poisson expectation cancels the numerator $j!$. Moreover
\[
 \sum_{j=0}^{m-1}\lambda^j(m-j)!
   =\lambda^{m-1}\sum_{\ell=1}^m\frac{\ell!}{\lambda^{\ell-1}}
   \le2\lambda^{m-1},
\]
because successive terms have ratio at most $(m+1)/\lambda\le1/2$.
The denominator of $\mathsf H_a(j)$ is positive for $j\ge m$,
which proves the sign assertion.
\end{proof}

\paragraph{When a floor removes the distinction between the transforms.}
Some estimators replace all sufficiently small count values by a fixed
positive floor. If the floor is at least the order, the signed and
nonnegative transforms give exactly the same result, as the next lemma
shows. This is a deterministic identity, with no probabilistic assumptions.

\begin{lemma}[Exact agreement after a count floor]\label{ag:lem:floor}
Let $a\ge1$ and $b\ge a$. For every integer $j\ge0$,
\[
 \max\{\mathsf H_a(j),b^a\}=\max\{\G_a(j),b^a\}.
\]
\end{lemma}
\begin{proof}
At integer orders the two count families agree. Otherwise write
$a=m+\gamma$, $0<\gamma<1$. They agree when $j\ge m+1$.
For $0\le j\le m$, the Gamma recurrence gives
\[
 |\mathsf H_a(j)|
 =\frac{j!}{\Gamma(1-\gamma)}
             \prod_{\ell=j+1}^m(a-\ell)
 \le j!(m-j)!\le m!\le a^a\le b^a.
\]
The bound $\Gamma(1-\gamma)\ge1$ follows, for example, by
Jensen's inequality applied to $T^{-\gamma}$ with
$T\sim\operatorname{Exp}(1)$. The truncated count is zero in
this range, proving the claim.
\end{proof}

Above the original cutoff, $\mathsf H_a$ equals $\G_a$ and retains
its pointwise lower bound $(j-a)_+^a$. Its additional value at $j=m$
is nonnegative; Lemma~\ref{ag:lem:negative} controls the remaining
signed correction. These identities therefore extend the heavy-count
framework. The subsequent sections construct and calibrate the corresponding
polynomial estimator for rare symbols.

\subsection{The likelihood-overlap identity used in the lower bounds}
For probability laws $P,Q$, with likelihood ratio $L=dP/dQ$, define
$\chi^2(P,Q)=\E_Q(L-1)^2$. Total variation obeys
$\TV(P,Q)\le\tfrac12\sqrt{\chi^2(P,Q)}$. Thus a small second moment
of a mixture likelihood ratio shows that the mixture is difficult to
distinguish from the reference law. The next calculation gives that
second moment explicitly for Poisson observations. In the iid application,
the coordinate index $e$ is a symbol and normally $\nu_e=n$.

\label{sec:overlap-proof}
\begin{lemma}[Exponential overlap and hidden locations]\label{lem:poisson-overlap}
Let independent $N_e\sim\Poi(\nu_e\theta_e^*)$ be the reference
experiment, with $\theta_e^*>0$. For alternatives $\theta,\phi$,
their likelihood ratios satisfy
\begin{equation}\label{eq:overlap}
 \E_*L_\theta L_\phi
 =\exp\left\{\sum_e\nu_e
       \frac{(\theta_e-\theta_e^*)(\phi_e-\theta_e^*)}{\theta_e^*}
        \right\}.
\end{equation}
Suppose $m$ alternatives have disjoint perturbation supports and
$\sum_e\nu_e(\theta_e^{(s)}-\theta_e^*)^2/\theta_e^*\le J$.
Write $P_s$ for the law of alternative $s$ and
$\overline P=m^{-1}\sum_{s=1}^{m}P_s$. Their uniform mixture satisfies
\begin{equation}\label{eq:mixture}
 \chi^2(\overline P,P_*)\le(e^J-1)/m.
\end{equation}
\end{lemma}
\begin{proof}
For one coordinate, the likelihood ratio is
\[
 L_\theta(N)=\exp\{-\nu(\theta-\theta^*)\}
                   (\theta/\theta^*)^N.
\]
The formula includes $\theta=0$ by continuity, with $0^0=1$.
Multiplying two such ratios and using
$\E_*t^N=\exp\{\nu\theta^*(t-1)\}$ gives
\[
 \E_*L_\theta L_\phi
 =\exp\left\{-\nu(\theta+\phi-2\theta^*)
            +\nu\theta^*(\theta\phi/(\theta^*)^2-1)\right\}
 =\exp\left\{\nu\frac{(\theta-\theta^*)(\phi-\theta^*)}{\theta^*}\right\}.
\]
Independence multiplies the coordinate overlaps, proving
\eqref{eq:overlap}. Off-diagonal overlaps of alternatives
equal one by disjointness. Expanding the squared mixture likelihood
and subtracting one proves \eqref{eq:mixture}.
\end{proof}

The same kernel gives the iid moment cancellation. For a finite signed
measure $\omega$ supported on a bounded subset of $[0,\infty)$ and $u_*>0$,
termwise expansion of \eqref{eq:overlap} gives exactly
\begin{equation}\label{eq:moment-expansion}
 \sum_{j\ge0}
 \frac{\bigl[\int e^{-\lambda u}(\lambda u)^j/j!\,d\omega(u)\bigr]^2}
      {e^{-\lambda u_*}(\lambda u_*)^j/j!}
 =\sum_{\ell\ge0}\frac{\lambda^\ell}{\ell!u_*^\ell}
       \left[\int(u-u_*)^\ell\,d\omega(u)\right]^2.
\end{equation}
Bounded support justifies all exchanges by absolute domination by the
exponential series. For instance, if $\omega=\mu_+-\mu_-$ and the probability measures
$\mu_+,\mu_-$ have equal moments through degree $L$, then
$\int(u-u_*)^\ell\,d\omega(u)=0$ for $0\le\ell\le L$.
The right side starts at degree $L+1$. The left side measures a weighted
squared difference between their mixed Poisson count probabilities.
This is the exact reason moment matching hides differences from the data.
Thus the hidden-location calculation and iid moment matching are
evaluations of the same independent-count likelihood kernel.

\clearpage
\subsection{The order-two example}
\label{sec:collision-example}

Imagine first that the number of observations is Poisson with mean $n$.
Then the counts $N_i$ of the symbols are independent and
$N_i\sim\Poi(np_i)$. Section~\ref{sec:count-foundations} explains this
reduction and the return to a fixed sample size.

\subsubsection{An unbiased estimator of the collision probability}
At $\alpha=2$, $P_2(p)=\sum_i p_i^2$ is the probability that two
independent observations are equal. Since
$\E[N_i(N_i-1)]=(np_i)^2$, the estimator
\[
 Z=\frac{1}{n^2}\sum_{i=1}^k N_i(N_i-1)
\]
satisfies $\E Z=P_2(p)$. The product $N_i(N_i-1)$ counts ordered pairs
of distinct observations of symbol $i$.

Why not use $\sum_i(N_i/n)^2$? Because
$\E N_i^2=(np_i)^2+np_i$, its expectation is $P_2+1/n$.
Subtracting $N_i$ removes that bias exactly.

\subsubsection{Its exact variance and a uniform sample bound}
The falling-factorial identity
\[
 (N)_2^2=(N)_4+4(N)_3+2(N)_2
\]
gives, for a Poisson count with mean $\lambda$,
$\Var((N)_2)=4\lambda^3+2\lambda^2$. Independence therefore yields
\begin{equation}\label{intro:collision-variance}
 \Var Z=\frac{4P_3(p)}n+\frac{2P_2(p)}{n^2}.
\end{equation}
Since $\max_i p_i\le\sqrt{P_2}$,
$P_3\le (\max_i p_i)P_2\le P_2^{3/2}$.
Combining this with $P_2\ge1/k$ gives
\begin{equation}\label{intro:collision-relative}
 \frac{\Var Z}{P_2^2}
 \le\frac{4}{n\sqrt{P_2}}+\frac{2}{n^2P_2}
 \le\frac{4\sqrt{k}}n+\frac{2k}{n^2}.
\end{equation}
Fix $\eta=1-e^{-\varepsilon}\in(0,1)$. If $n=C_\varepsilon\sqrt{k}$,
Chebyshev's inequality implies
\[
 \Prob\{|Z-P_2|>\eta P_2\}
 \le\frac{4/C_\varepsilon+2/C_\varepsilon^2}{\eta^2}.
\]
Choose $C_\varepsilon$ to make this bound at most $1/4$.
On the complementary event,
\[
 |\log Z-\log P_2|
 \le\max\{-\log(1-\eta),\log(1+\eta)\}=\varepsilon.
\]
Thus $-\log Z$ estimates $H_2$ to accuracy $\varepsilon$.
Poissonization can be implemented from a fixed pool of a constant multiple
of $n$ observations, with exponentially small overflow probability.
This proves $S_2(k,\varepsilon,1/3)=O_\varepsilon(\sqrt{k})$ for large $k$.
The hidden-coordinate argument in Theorem~\ref{thm:iid-collision-new}
proves the matching lower bound.

\paragraph{What this example teaches.}
The estimator does not estimate each $p_i$ accurately. It estimates a sum
directly, and its variance is compared to the square of that sum.
For a general integer order, replace $(N)_2$ by $(N)_r$.
For a noninteger order, no finite polynomial equals $x^\alpha$;
the next step is polynomial approximation with simultaneous variance
control. That is the central additional difficulty.

\clearpage
\section{A direct proof for every integer order}
\label{tutorial:integer}

Before constructing a polynomial approximation for noninteger orders, it
is useful to solve the integer case directly. The estimator is particularly
simple: count ordered collisions. The subtlety is in the proof when the
integer order grows with the alphabet size.

Throughout this section, $r\ge2$ is an integer,
$P_r(p)=\sum_i p_i^r$, and $H_r(p)=-(r-1)^{-1}\log P_r(p)$, as in
\eqref{main:renyi}. Under Poissonization, the counts are independent with
$N_i\sim\Poi(np_i)$. Define
\begin{equation}\label{tutorial:raw-estimator}
 \widehat P_r=n^{-r}\sum_{i=1}^k(N_i)_r.
\end{equation}
It is nonnegative and unbiased: $\E\widehat P_r=P_r(p)$.
Here $(N_i)_r$ counts ordered selections of $r$ occurrences of symbol $i$.
The estimator does not need to know which symbols are rare or frequent.
The partition used below is only a device for analyzing its lower tail.

\subsection{Why a constant relative-variance bound is too restrictive}
For fixed $r$, one usually proves that $\widehat P_r/P_r$ is close to
one. When $r$ grows, that requirement becomes unnecessarily strong.
The exact condition needed for entropy accuracy $\varepsilon$ is
\begin{equation}\label{tutorial:integer-band}
 e^{-\varepsilon(r-1)}P_r
 \le\widehat P_r\le e^{\varepsilon(r-1)}P_r.
\end{equation}
Indeed, taking logarithms and dividing by $r-1$ gives an error at most
$\varepsilon$. The allowed multiplicative interval widens with $r$.

Here is a concrete reason not to insist on bounded relative variance.
Fix $D>0$, put $a=k^{-1+1/r}$, and consider
\[
 p_1=a,\qquad p_i=\frac{1-a}{k-1}\quad(i>1),
 \qquad n=Dr/a=Drk^{1-1/r}.
\]
Since $a\ge1/k$, the remaining probabilities are at most $1/k$.
Consequently $a^r\le P_r(p)\le2a^r$. The $j=1$ term of the exact
factorial variance formula gives
\begin{equation}\label{tutorial:variance-obstruction}
 \frac{\Var(\widehat P_r)}{P_r(p)^2}
 \ge\frac{r^2a^{2r-1}}{4na^{2r}}
 =\frac{r}{4D}.
\end{equation}
Thus the relative variance diverges along any sequence $r\to\infty$
with fixed $D$. This does not rule out the entropy guarantee
\eqref{tutorial:integer-band}. The proof below uses the full interval:
a second-moment bound controls the upper tail, while positivity and a
weighted Jensen inequality control the lower tail.

\subsection{The estimator and the uniform guarantee}
\begin{theorem}[Raw collision estimator, uniformly over integer orders]
\label{tutorial:integer-raw}
Fix $0<\varepsilon\le1/2$. For every $k\ge2$ and integer $r\ge2$, let
\[
 D=D_\varepsilon=
 \frac{100e^{2\varepsilon}}{\varepsilon}
       \left(\frac{4096}{\varepsilon^2}+1\right),
 \qquad n=Drk^{1-1/r}.
\]
The estimator \eqref{tutorial:raw-estimator}, based on independent
$N_i\sim\Poi(np_i)$, satisfies \eqref{tutorial:integer-band} with
probability at least $11/12$, uniformly over $p\in\Delta_k$.
Therefore $-(r-1)^{-1}\log\widehat P_r$ estimates $H_r(p)$ within
$\varepsilon$ on this event. Define the estimate arbitrarily when
$\widehat P_r=0$.
Using $\lceil2n\rceil$ ordinary iid observations instead gives success
probability at least $2/3$. The constant $D_\varepsilon$ is independent
of both $r$ and $k$; its displayed value is chosen for convenience,
not numerical efficiency.
\end{theorem}

\begin{proof}
\textbf{Step 1: obtain a variance bound, allowing it to grow with $r$.}
Write $P=P_r(p)$ and $K_r=k^{1-1/r}$. Jensen's inequality gives
$P\ge k^{1-r}=K_r^{-r}$. Applying
Lemma~\ref{lem:factorial-energy-new} to $Q(x)=x^r$ yields
\[
 \Var(\widehat P_r)
 =\sum_{j=1}^r\binom rj^2j!\,n^{-j}P_{2r-j}(p).
\]
For $q\ge r$, monotonicity of vector norms gives
$P_q\le P^{q/r}$. Moreover
$\binom rj^2j!=(r)_j^2/j!\le r^{2j}/j!$. Hence
\begin{equation}\label{tutorial:integer-relative-variance}
 \frac{\Var(\widehat P_r)}{P^2}
 \le\sum_{j=1}^r\frac1{j!}
          \left(\frac{r^2}{nP^{1/r}}\right)^j
 \le e^{r/D}-1,
\end{equation}
where $nP^{1/r}\ge Dr$. We will compare this exponential growth to
the exponential tolerance in \eqref{tutorial:integer-band}.

\medskip
\noindent\textbf{Step 2: control the upper tail at the actual tolerance.}
Set $y=\varepsilon(r-1)$. For $u,y\ge0$,
$e^u-1\le ue^u$ and $e^y-1\ge ye^{y/2}$; the latter follows from
$2\sinh(y/2)\ge y$. Chebyshev's inequality therefore gives
\begin{align}
 \Pr\{\widehat P_r>e^yP\}
 &\le\frac{e^{r/D}-1}{(e^y-1)^2}
 \le\frac{r}{Dy^2}e^{r/D-y}
 \le\frac{2}{D\varepsilon^2}
 \le\frac1{100}.\label{tutorial:integer-upper-tail}
\end{align}
Indeed, $D\ge2/\varepsilon$ and $r\le2(r-1)$ imply $r/D\le y$,
while $r/(r-1)^2\le2$. This is why a relative variance growing like
$e^{r/D}$ is harmless for the upper tail.

\medskip
\noindent\textbf{Step 3: separate rare and frequent coordinates in the proof.}
Put $\lambda_i=np_i$ and define
\[
 T=\frac{4096r}{\varepsilon^2},\qquad
 \mathcal L=\{i:\lambda_i\le T\},\qquad
 \mathcal H=\{i:\lambda_i>T\}.
\]
Let $L=\sum_{i\in\mathcal L}p_i^r$ and
$H=\sum_{i\in\mathcal H}p_i^r$, so that $L+H=P$.
Correspondingly write $\widehat P_r=F+J$, where
$F=n^{-r}\sum_{i\in\mathcal L}(N_i)_r$ and
$J=n^{-r}\sum_{i\in\mathcal H}(N_i)_r$.
We control the two contributions by different inequalities.

\medskip
\noindent\textbf{Step 4: the rare contribution has a small additive error.}
For $N\sim\Poi(\lambda)$, the factorial product identity gives
\begin{align*}
 \E(N)_r^2
 &=\sum_{j=0}^r\binom rj^2j!\lambda^{2r-j}
 \le\sum_{j=0}^r\binom rj r^j\lambda^{2r-j}
 =\lambda^r(\lambda+r)^r.
\end{align*}
The inequality uses $\binom rj j!=(r)_j\le r^j$.
Independence and $\lambda_i\le T$ on $\mathcal L$ now imply
\begin{equation}\label{tutorial:integer-light-variance}
 \E F=L,\qquad
 \frac{\Var F}{P^2}
 \le\frac{(T+r)^rL}{n^rP^2}
 \le\left(\frac{T+r}{Dr}\right)^r=A^r,
 \qquad A=\frac{4096/\varepsilon^2+1}{D}.
\end{equation}
The second inequality uses $L\le P$ and $n^rP\ge(Dr)^r$.
Set $\eta=\varepsilon(r-1)/2$ and
$d=e^{-\eta}(1-e^{-\eta})$. Because
$1-e^{-\eta}\ge\eta e^{-\eta}$, another application of Chebyshev gives
\begin{align*}
 \Pr\{F<L-dP\}
 &\le\frac{A^r}{d^2}
 \le\frac{4(Ae^{2\varepsilon})^r}{\varepsilon^2(r-1)^2}
 \le\frac4{10000}.
\end{align*}
For the middle inequality, $d^2\ge\eta^2e^{-4\eta}$ and
$4\eta=2\varepsilon(r-1)\le2\varepsilon r$.
For the last one, $Ae^{2\varepsilon}=\varepsilon/100$ and $r\ge2$.

\medskip
\noindent\textbf{Step 5: a weighted average controls the frequent contribution.}
If $H=0$, then $J=0$ and the claimed bound below is automatic.
Otherwise define weights $w_i=p_i^r/H$ for $i\in\mathcal H$, and put
\[
 Z=\sum_{i\in\mathcal H}w_i\frac{N_i}{\lambda_i}.
\]
The weights are nonnegative and sum to one. Thus
\[
 \E Z=1,\qquad
 \Var Z=\sum_{i\in\mathcal H}\frac{w_i^2}{\lambda_i}
 \le\frac1T.
\]
Except on an event of probability at most $256/(\varepsilon^2T)$,
we have $Z\ge1-\varepsilon/16$. For any integer $x\ge0$,
$(x)_r\ge(x-r)_+^r$. Since $z\mapsto(z_+)^r$ is convex, Jensen's
inequality then gives
\begin{align*}
 \frac JH
 &\ge\sum_{i\in\mathcal H}w_i
        \left(\frac{N_i-r}{\lambda_i}\right)_+^r
 \ge\left(Z-r\sum_{i\in\mathcal H}\frac{w_i}{\lambda_i}\right)_+^r\\
 &\ge(1-\varepsilon/8)^r
 \ge e^{-\varepsilon r/4}\ge e^{-\eta}.
\end{align*}
Here $r/T\le\varepsilon/16$,
$\log(1-\varepsilon/8)\ge-\varepsilon/4$ for
$0<\varepsilon\le1/2$, and $r\ge2$. The exceptional probability is
at most $256/(\varepsilon^2T)=1/(16r)\le1/32$.
Notice that we did not require every frequent count to concentrate:
one weighted average was sufficient.

\medskip
\noindent\textbf{Step 6: combine the lower-tail estimates.}
On the two successful events,
\[
 \widehat P_r=F+J\ge L-dP+e^{-\eta}H
 \ge(e^{-\eta}-d)P=e^{-2\eta}P.
\]
The middle inequality uses $L+H=P$ and $L\ge e^{-\eta}L$.
Since $2\eta=\varepsilon(r-1)$, this is the required lower bound.
Combining it with \eqref{tutorial:integer-upper-tail}, the total failure
probability is at most $1/100+4/10000+1/32<1/12$.

\medskip
\noindent\textbf{Step 7: return to a fixed number of iid observations.}
Take $m=\lceil2n\rceil$ observations, and independently draw
$M\sim\Poi(n)$. If $M\le m$, use the first $M$ observations to form
the counts in \eqref{tutorial:raw-estimator}; otherwise return an arbitrary
answer. Coupling to an unlimited iid sequence adds failure probability
at most
\[
 \Pr\{M>m\}\le\Pr\{M>2n\}\le e^{-n/3}<1/12.
\]
The Poisson Chernoff bound gives the middle inequality; the displayed
choice of $D_\varepsilon$ ensures the last one. Total failure is below
$1/6$, and hence below $1/3$, proving the fixed-sample assertion.
\end{proof}

\paragraph{What this direct argument establishes.}
For every integer $r\ge2$, the same raw collision estimator uses
$O_\varepsilon(rk^{1-1/r})$ samples. This is the relevant upper rate
below the high-order threshold $r\approx\log k$; the lower bound and the
$k\log k$ high-order estimator are proved separately. The argument above
does not claim that raw collisions attain the saturation rate when
$r\gg\log k$. 

\clearpage
\section{Two simpler bounds: large orders and a hidden large coordinate}
\label{sec:endpoints}

Before entering the noninteger approximation argument, we prove the
large-order upper bound and a lower bound that applies to every
$\alpha\ge3/2$. These proofs use only the preceding count identities.
They explain both the integer-type scale and its saturation.

\subsection{Why the upper bound saturates at high orders}
For large $\alpha$, the norm $\|p\|_\alpha$ is driven by the largest
probabilities. At sample size $Ck\log k$, all empirical coordinates
can be controlled to error proportional to $p_i+1/k$. The additive
$1/k$ term is harmless because $\|p\|_\alpha\ge k^{1/\alpha-1}$.
The following proof makes this observation precise and explains why the
Gamma-transform order is capped even when the target order is very large.

\begin{proposition}[The high-order and min-entropy upper bound]\label{prop:iid-high}
For iid $p\in\Delta_k$, real $\alpha\ge\max\{2,\log k\}$,
$L=\log(4k/\delta)$, $s=\min\{\alpha,L\}$, and $n=CkL$, use
\[
 \widehat P=\sum_i[\g_s(N_i)/n]^\alpha,
 \qquad N_i\sim\Poi(np_i).
\]
Its logarithm estimates iid R\'enyi entropy with fixed error and
confidence using $O(k\log k)$ samples. The same construction with
$s=L$ and a maximum in place of the sum estimates min-entropy.
\end{proposition}
\begin{proof}
\noindent\textbf{Step 1: control all counts.}
For $N_i\sim\Poi(np_i)$, a Poisson Bernstein bound gives
\[
 \Prob\{|N_i-np_i|>\sqrt{2np_iL}+2L/3\}\le2e^{-L}.
\]
Since $n=CkL$, division by $n$ and the inequality
$\sqrt{p_i/k}\le(p_i+1/k)/2$ show that, for sufficiently large
$C=C(\eta)$, the right-hand deviation is at most
$\eta(p_i+1/k)$. A union bound over $k$ symbols gives,
with failure probability at most $2ke^{-L}\le\delta/2$,
$|N_i/n-p_i|\le\eta(p_i+1/k)$ for all $i$, after fixing $C$ large.
\noindent\textbf{Step 2: account for the count transform.}
By Lemma~\ref{lem:gamma}, $s/n\le1/(Ck)\le\eta/k$ and thus
\[
 |\g_s(N_i)/n-p_i|\le\eta p_i+2\eta/k.
\]
\noindent\textbf{Step 3: estimate the norm, then take its logarithm.}
Write $\widehat v_i=\g_s(N_i)/n$. By the triangle inequality,
\[
 |\|\widehat v\|_\alpha-\|p\|_\alpha|
 \le\|\widehat v-p\|_\alpha
 \le\eta\|p\|_\alpha+2\eta k^{1/\alpha-1}.
\]
The $\ell_\alpha$ triangle inequality and
$\norm p_\alpha\ge k^{1/\alpha-1}$ imply relative norm error at
most $3\eta$. For $\eta\le1/12$, logarithmic norm error is at
most $6\eta$ and entropy error is at most $12\eta$.
The $\ell_\infty$ argument is identical, using $\norm p_\infty\ge1/k$.
Choose $\eta\le\min\{1/12,\varepsilon/12\}$.
For min-entropy use $-\log\max_i\widehat v_i$; for finite order use
$-\log\widehat P/(\alpha-1)$. The entropy factor
$\alpha/(\alpha-1)$ is at most two. Taking a Poisson sample from a
fixed pool of a constant multiple of $n$ observations adds only an
exponentially small overflow probability, absorbed in the unused
failure budget. Thus a fixed sample size $O_{\varepsilon,\delta}(k\log k)$
suffices at any fixed confidence.
\end{proof}

\subsection{One hidden-coordinate calculation for the iid collision scale}
\label{sec:iid-common-lower}
This lower bound hides one enlarged coordinate among several candidates.
The entropy changes substantially, but the data reveal too little about
which candidate was changed. The reservoir consists of the remaining
coordinates; changing its total mass keeps every vector normalized.
The proof is uniform in the order. For noninteger orders it gives the
collision obstruction only; the additional approximation obstruction
requires the random-prior construction later.

\begin{theorem}[Order-uniform collision and saturation lower bound]
\label{thm:iid-collision-new}
Fix $\eps>0$. For all sufficiently large $k$, uniformly over every
finite real $\alpha\ge3/2$, iid R\'enyi entropy estimation with
accuracy $\eps$ and failure probability at most $1/3$ requires
\[
 \Omega_\eps\!\left(\ell k^{1-1/\ell}\right),
 \qquad \ell=\min\{\alpha,\log k\}.
\]
The same $\Omega_\eps(k\log k)$ conclusion holds for min-entropy.
\end{theorem}
\begin{proof}
\noindent\textbf{Step 1: construct the null and alternatives.}
Set $m=\lfloor e^{\ell/64}\rfloor$, $q=k^{-1+1/\ell}$,
and $B=e^{2\eps+2}$. The null distribution gives probability $q$
to each of $m$ distinguished coordinates and distributes its
remaining mass uniformly among the other $k-m$ coordinates.
Alternative $s\le m$ changes only the $s$th distinguished
probability to $Bq$ and reduces the reservoir total by $(B-1)q$.
All distributions are well defined for large $k$, uniformly in
$\alpha$: $mq\le k^{1/64-1/3}\to0$ and $q\le k^{-1/3}$.

\medskip\noindent\textbf{Step 2: verify entropy separation.}
The reservoir probabilities under the null are at most $1/k$
because $q\ge1/k$. Also $q^\alpha\ge k^{1-\alpha}$ because
$\alpha\ge\ell$. Thus
\[
 P_\alpha(p^{(0)})\le(m+1)q^\alpha,
 \qquad P_\alpha(p^{(s)})\ge B^\alpha q^\alpha.
\]
Since $m+1\le2e^{\alpha/64}$, the entropy separation is at least
\[
 \frac{\alpha\log B-\log(m+1)}{\alpha-1}
 \ge\frac{\alpha(2\eps+2-1/64)-\log2}{\alpha-1}
 >2\eps.
\]
For min-entropy the null and alternative maxima are $q$ and $Bq$,
so their separation is $\log B>2\eps$; here use $\ell=\log k$.

\medskip\noindent\textbf{Step 3: show that the observation laws are close.}
Consider independent Poisson symbol counts with total mean $n$.
Relative to the null, the diagonal and off-diagonal likelihood
overlaps in Lemma~\ref{lem:poisson-overlap} are
\[
 \E_0 L_s^2=e^{u+v},\quad
 \E_0 L_sL_t=e^v\quad(s\ne t),
 \qquad
 u=n(B-1)^2q,\quad
 v=\frac{n(B-1)^2q^2}{1-mq}.
\]
The off-diagonal term accounts exactly for the shared reservoir;
it is not set to one. Therefore the uniform mixture obeys
\[
 \chi^2(\overline P,P_0)
 =e^v\frac{e^u-1}{m}+e^v-1.
\]
Put $y=\log m$. If $n\le c(1+y)/q$, then
$u\le\kappa(1+y)$, where $\kappa=c(B-1)^2$.
For $0<\kappa\le1/2$,
\[
 \frac{e^{\kappa(1+y)}-1}{e^y}
 \le\kappa e^\kappa(1+y)e^{-(1-\kappa)y}\le4\kappa.
\]
Moreover $v\le\kappa(1+y)q/(1-mq)\to0$ uniformly over the
orders, since $y\le\log k/64$ and $q\le k^{-1/3}$.
Choose $c$ sufficiently small, then $k$ sufficiently large, to
make the total variation at most $\frac12\sqrt{\chi^2}<1/8$.
For any test between a null law and an alternative mixture, the sum
of its two error probabilities is at least $1-\TV$. An entropy estimator
accurate to $\varepsilon$ would provide such a test because the target
intervals are disjoint. Thus its worst-case failure probability is at
least $(1-1/8)/2=7/16$, contradicting a uniform bound of $1/3$.

\medskip\noindent\textbf{Step 4: read off the sample scale.}
Since $1+\log\lfloor e^{\ell/64}\rfloor\asymp\ell$ for
$\ell\ge3/2$, the Poisson lower scale is
$\ell k^{1-1/\ell}$. A fixed-sample estimator using $s$ observations
can be run on a $\Poi(2s)$ sample except with probability at most
$e^{-s/4}$; this transfers the lower scale to fixed sample sizes.
More explicitly, apply the comparison at a fixed small constant
multiple of the diverging lower scale and use monotonicity to
cover all smaller fixed sample sizes.
\end{proof}

\clearpage
\section{Choosing the light-symbol polynomial and its accuracy margins}
\label{ped:upper-calibration}

The elementary count identities explain how to estimate a polynomial. We
must now choose the polynomial. This is the most technical deterministic
part of the i.i.d. proof. Its output is a polynomial whose expectation lies
strictly inside the required multiplicative accuracy interval. The
probabilistic section will show that the observed statistic stays close
to this expectation.

\paragraph{The four steps in this section.}
First, solve a two-variable optimization problem exactly. Second, use that
solution to choose the sample size and two probability cutoffs. Third,
construct polynomial approximations that behave correctly near both
cutoffs. Fourth, add the pointwise bounds over the alphabet. Keeping these
steps separate distinguishes an algebraic identity from a statistical
sample-complexity guarantee.

\paragraph{Notation and the target.}
Throughout, $\alpha>1$ is noninteger, $\beta=\alpha-1$,
$m=\lfloor\alpha\rfloor$, $\gamma=\alpha-m$, $\delta=1-\gamma$, and
$h=\gamma\delta$. Thus $\gamma$ and $\delta$ are the distances to the
adjacent integers. They are not error probabilities. Write
$P=\sum_i p_i^\alpha$ and $P_*=k^{1-\alpha}\le P$.
The full logarithmic tolerance is $E=2\varepsilon\beta$ and
$\kappa=e^{-E}$. We seek a corrected mean $\mu$ satisfying
$\kappa P<\mu<P$, with explicit positive gaps from both endpoints.
These gaps will later absorb sampling noise and pilot-selection errors.

\begin{center}
\begin{tabular}{@{}p{0.16\textwidth}p{0.77\textwidth}@{}}
\toprule
Quantity & Role in the construction\\\midrule
$a$ & Lower probability cutoff; the aggregate contribution below $a$ is paid by a deterministic correction.\\
$b$ & Upper end of the polynomial-approximation interval for light symbols.\\
$L$ & Degree scale of the polynomial for light symbols.\\
$n$ & Mean number of observations used for the main Poisson counts.\\
$E_0<E$ & Tolerance used by the deterministic approximation; $E-E_0$ is reserved for subsequent errors.\\
$X\ge1$ & Dimensionless factor multiplying the integer-type sample scale $\alpha k^{1-1/\alpha}$.\\
$u$ & A small positive regularization parameter that provides additional room for fluctuations.\\
\bottomrule
\end{tabular}
\end{center}
The symbol $Q^2$ below is a scalar appearing in the root equation.
Later $Q(p)$ denotes a polynomial; the argument $p$ distinguishes these
uses. Likewise, the profile $u_s(x)$ is a function and is distinct from
the regularization parameter $u$.

\subsection{How the polynomial becomes an observable estimator}
\label{tutorial:upper-guide}\label{tutorial:upper}

Before constructing the profiles, it is useful to see what they are
supposed to accomplish. The unknown quantity is
$P_\alpha(p)=\sum_i p_i^\alpha$. A polynomial
$Q(x)=\sum_jq_jx^j$ is useful because its expectation can be reproduced
from counts: replacing $p_i^j$ by $(N_i)_j/n^j$ gives an unbiased
estimate of $Q(p_i)$. The difficulty is therefore to approximate the
fractional power using integer powers, while keeping their statistical
fluctuations under control.

\paragraph{1. Why use the two adjacent integers?}
Write $\alpha=m+\gamma=m+1-\delta$, where $0<\gamma,\delta<1$.
Consider, first, two illustrative profiles
\[
 F_-(x)=x^\gamma R_-(x),\qquad
 F_+(x)=x^\delta R_+(x),
\]
where $R_-$ and $R_+$ are ordinary polynomials. If
$R_-(x)=\sum_{j=0}^{m-1}a_jx^j$ and
$R_+(x)=\sum_{j=0}^{d_+}b_jx^j$, then direct expansion gives
\begin{align}
 p^\alpha F_-(a/p)
   &=\sum_{j=0}^{m-1}a_j a^{\gamma+j}p^{m-j},
   \label{tutorial:upper-reverse}\\
 p^\alpha F_+(p/b)
   &=\sum_{j=0}^{d_+}b_j b^{-\delta-j}p^{m+1+j}.
   \label{tutorial:upper-forward}
\end{align}
For the first identity, the exponent of $p$ is
$\alpha-\gamma-j=m-j$; for the second it is
$\alpha+\delta+j=m+1+j$. Thus the first block supplies degrees
$1,\ldots,m$ and the second supplies degrees $m+1,m+2,\ldots$.
The fractional exponents cancel exactly. The estimator will use
ordinary integer-order falling factorials, not a hypothetical
unbiased statistic for a fractional collision.

These two identities explain the construction, but the actual
polynomial needs an additional transition between the blocks.
Moreover, the reverse profile in the full construction may have a
polynomial part of degree $m$, which also produces a constant term.
This is intentional: for $m=1$, the near-one profile gives an affine
polynomial with a nonzero constant. The statistic therefore retains
$q_0(N_i)_0=q_0$, including for unseen symbols. The known alphabet
bound makes these zero-count contributions observable.

\paragraph{2. The actual coupled polynomial.}
The expression used in the proof is the polynomial from
\eqref{eq:coupled-single-polynomial}, here called $Q_0$ to distinguish
it from its final rescaling:
\begin{equation}\label{tutorial:upper-actual-polynomial}
\begin{split}
 Q_0(p)=p^\alpha\bigl[&A'u_\gamma(a/p)\{1-J_1(p/b)\}\\
 &+B'v_{\delta,\omega_0}(p/b)
   +\omega_0(1-B')J_\delta(p/b)\bigr].
\end{split}
\end{equation}
Here $u_\gamma(x)=x^\gamma R_-(x)$ is the reverse profile;
$J_\delta(x)=x^\delta R_\delta(x)$ and
$v_{\delta,\omega_0}(x)=x^\delta V_{\delta,\omega_0}(x)$
have ordinary polynomial factors. Also $J_1$ is an ordinary
polynomial. Hence every term in
\eqref{tutorial:upper-actual-polynomial} is an ordinary polynomial,
with its value at $p=0$ defined by that polynomial extension.
Multiplication by $1-J_1$ can create degrees above $m$; the useful
separation in \eqref{tutorial:upper-reverse}--\eqref{tutorial:upper-forward}
describes the starting blocks, not a claim that the full coupled
polynomial has disjoint degree supports.

The profiles and their constants will be constructed below. Their
roles can already be stated. The reverse profile controls the lower
end of $[a,b]$. The forward profile controls its upper end. The factor
$1-J_1$ turns off the reverse contribution as $p$ increases, while the
last term connects the polynomial to the interior value
$\omega_0p^\alpha$. The coefficients $A'=Ae^{-M\gamma}$ and $B'=Be^{-M\delta}$
balance the profiles; $A,B$ are the envelope weights in
\eqref{oc:definitions}, and $M$ is a fixed damping constant.
Their detailed normalization determines the sample scale, so merely
knowing that a polynomial approximation exists would not suffice.

\paragraph{3. The algorithm uses counts, not the unknown probabilities.}
The numbers $a$ and $b$ are deterministic design parameters. They are
not thresholds at which the algorithm can inspect $p_i$.
Instead, it uses an independent pilot count to select a branch.
The following steps specify the construction for noninteger
$1<\alpha<\log k$; the integer and high-order cases were treated above.

\begin{enumerate}[leftmargin=*,label=\textbf{Step \arabic*.}]
\item \textbf{Choose deterministic parameters.}
The inputs are $k,\alpha,\varepsilon$ and an i.i.d. sample from $p$.
Use Lemma~\ref{upcal:parameters} to calculate the root $X$ and
$E_0,u,L,n,a,b$. In particular,
\[
 n=D\alpha k^{1-1/\alpha}X,\qquad
 L=\lceil\alpha(C_L+c_L\log X)\rceil,\qquad
 b=c_bL/n,\qquad
 a=\lambda\frac{bhz^{1/h}}{m^2L^2},
\]
where $G_\gamma(z)=e^{-E_0}$ and $h=\gamma\delta$.
All constants in this rule are fixed in terms of $\varepsilon$ alone.
Construct $Q_0$ above with $\omega_0=e^{-E_0/8}$, and set
\[
 s=(1+c_uu)^{-1},\qquad Q=sQ_0=\sum_{j=0}^dq_jp^j,
 \qquad \omega=s\omega_0.
\]
These operations depend on the known inputs, not on $p$.

\item \textbf{Form independent main and pilot counts.}
Use Poisson splitting to obtain independent
$N_i\sim\Poi(np_i)$ and $Y_i\sim\Poi(c_pnp_i)$.
Select the light branch using the observable indicator
\[
 I_i=\mathbf1\{Y_i\le c_pB\alpha\}.
\]
The constants $c_p$ and $B$ are part of the construction. Their
threshold is fixed by $\alpha$; it is not a test of whether the
unknown probability lies in $[a,b]$.

\item \textbf{Compute the two candidate contributions.}
For each of the $k$ symbols, calculate
\[
 \widehat Q_i=\sum_{j=0}^dq_j\frac{(N_i)_j}{n^j},\qquad
 \widehat G_i=\omega n^{-\alpha}\mathsf H_\alpha(N_i),\qquad
 \mathsf H_\alpha(j)=\frac{\Gamma(j+1)}{\Gamma(j+1-\alpha)}.
\]
Use the entire reciprocal-Gamma convention at its zeros.
The statistic $\mathsf H_\alpha$ is signed at some low counts;
it is not the nonnegative truncated Gamma statistic.
Its low-count values are included in the proof.

\item \textbf{Combine the contributions and the known affine total.}
Put $\chi=\min\{\alpha-1,1\}$, $P_*=k^{1-\alpha}$, and
$E_a=\chi d_\varepsilon C_0^m(ka)^{\alpha-1}$, and return the
power-sum statistic
\begin{equation}\label{tutorial:upper-observable}
 \widehat P=2E_aP_*+q_1+
 \sum_{i=1}^k\left[
 I_i\left(\widehat Q_i-q_1N_i/n\right)
 +(1-I_i)\left(\widehat G_i-q_1N_i/n\right)\right].
\end{equation}
For $\widehat P>0$, the entropy output is
$-\log\widehat P/(\alpha-1)-\varepsilon$; define it arbitrarily
otherwise. The proof establishes positivity on the success event.
A fixed sample pool of size $\lceil4(1+c_p)n\rceil$ implements the
Poisson splitting; the exponentially small overflow probability is
absorbed into the stated failure probability.
\end{enumerate}

\paragraph{4. Why subtract the linear term?}
Let $T_i=I_i\widehat Q_i+(1-I_i)\widehat G_i$.
After omitting the deterministic repair $2E_aP_*$,
\eqref{tutorial:upper-observable} equals
\[
 q_1+\sum_i(T_i-q_1N_i/n)
 =\sum_iT_i+q_1\left(1-\sum_iN_i/n\right).
\]
Because $\sum_i p_i=1$, the added term has mean zero. Its variance is
$q_1^2/n$, since $\sum_iN_i\sim\Poi(n)$. This is a control variate:
we replace the noisy estimate of a known linear functional by its
exact value. The proof bounds the variance of the resulting residual;
it does not assume that adding any centered term must reduce variance.

Near order one, write $\alpha=1+\beta$. For any known $t>0$,
\begin{equation}\label{tutorial:upper-nearone-identity}
 P_{1+\beta}(p)=t^\beta+
 \beta\sum_i\frac{p_i^{1+\beta}-t^\beta p_i}{\beta}.
\end{equation}
For $p_i>0$, the quotient tends to $p_i\log(p_i/t)$ as
$\beta\downarrow0$, and its value at $p_i=0$ is zero.
For the specific scale $t=b/N_{\rm ker}^2$, where $N_{\rm ker}\asymp L$
is the kernel degree parameter, the near-one construction has
$q_1=t^\beta+O_\varepsilon(\beta t^\beta)$.
Subtracting the affine part therefore exposes a residual of order
$\beta$. The detailed proof retains this factor in bias and its square
in variance, matching the interval width
$(1-e^{-2\varepsilon\beta})P=O_\varepsilon(\beta P)$.

\paragraph{5. What do the margins pay for?}
Write $P=P_\alpha(p)$ and $\kappa=e^{-2\varepsilon(\alpha-1)}$.
The analysis uses
\begin{equation}\label{tutorial:upper-margin}
 \mathcal M=E_aP_*+\kappa uP+\chi\kappa\mathcal B,
 \qquad \mathcal B=\sum_i p_i^\alpha\min(Lnp_i,1).
\end{equation}
This quantity depends on $p$ and is used only in the proof.
The baseline $E_aP_*$ absorbs the aggregate error from $p_i<a$;
the deterministic correction uses part of this allowance, while the
upper approximation leaves room for that positive correction.
The term $\kappa uP$ is an additional interior reserve supplied by the
root calibration and final rescaling. The term $\chi\kappa\mathcal B$
provides count-dependent room that controls higher-degree fluctuations
and pilot errors. The variance calculations use their sum, rather than
assign every error to exactly one term.

The proof first constructs an ideal mixed quantity in
$[\kappa P+c_\varepsilon\mathcal M,
P-c_\varepsilon\mathcal M]$. It then shows that the observed statistic
is close to this quantity: conditional polynomial fluctuations are
small relative to $\mathcal M$, genuine heavy counts concentrate,
and the total contribution of pilot mistakes is small.
These bounds imply $\kappa P\le\widehat P\le P$ with high probability.
Thus the deterministic construction, the pilot, and the noise bounds
have separate roles, and each is needed before taking the logarithm.

\subsection{Exact overlap normalization at the original accuracy}
\label{oc:section}

The purpose of the next lemma is to combine two endpoint approximations
without wasting any of the requested entropy accuracy. Think of $u$ and
$v$ temporarily as two numbers in $[0,1]$, representing the normalized
values of those approximations. Their interaction is described by the
weighted product $u^\delta v^\gamma$. We choose weights $A,B$ so that
$Au+Bv$ has maximum exactly one when that product is at most $z$.
The same weights give an explicit minimum when the product is at least
$z$. The ratio between these two extrema determines how large an
approximation interval can be tolerated. This scalar calculation alone
is not an estimator or a variance bound.

\begin{lemma}[Exact two-arm envelope]\label{oc:envelope}
Let $0<\gamma<1$, $\delta=1-\gamma$, $h=\gamma\delta$, and $0<z<1$.
Set
\begin{equation}\label{oc:definitions}
 A=\frac{1-z^{1/\gamma}}{1-z^{1/h}},\qquad
 B=\frac{1-z^{1/\delta}}{1-z^{1/h}},\qquad
 G_\gamma(z)=\frac{zA^\delta B^\gamma}
                   {\delta^\delta\gamma^\gamma}.
\end{equation}
Then
\begin{align}
 \max_{\substack{0\le u,v\le1\\u^\delta v^\gamma\le z}}
       (Au+Bv)&=1,\label{oc:maximum}\\
 \min_{\substack{0\le u,v\le1\\u^\delta v^\gamma\ge z}}
       (Au+Bv)&=G_\gamma(z).\label{oc:minimum}
\end{align}
Moreover, $G_\gamma$ is a continuous strictly increasing bijection
from $(0,1)$ onto $(0,1)$.
\end{lemma}
\begin{proof}
\noindent\textbf{Step 1: compute the maximum.}
For fixed $u$, the largest admissible $v$ in
\eqref{oc:maximum} is
$\min\{1,z^{1/\gamma}u^{-\delta/\gamma}\}$.
On $0\le u\le z^{1/\delta}$ the maximum is at the right endpoint;
on $z^{1/\delta}\le u\le1$ the resulting objective is convex in $u$.
Both boundary values equal one, since
$Az^{1/\delta}+B=A+Bz^{1/\gamma}=1$.

\medskip\noindent\textbf{Step 2: show that the desired overlap level can be chosen uniquely.}
To verify monotonicity, write $z=e^{-ht}$, $t>0$, and
$F_\gamma(t)=-\log G_\gamma(e^{-ht})$.
An exact cancellation gives
\begin{equation}\label{oc:hyperbolic}
 G_\gamma(e^{-ht})=
 \frac{[\sinh(\delta t/2)/\delta]^\delta
       [\sinh(\gamma t/2)/\gamma]^\gamma}{\sinh(t/2)}.
\end{equation}
Let $J(x)=x\coth x$ for $x>0$ and $J(0)=1$.
Its derivatives satisfy
\begin{equation}\label{oc:J}
 0<J'(x)=\coth x-\frac{x}{\sinh^2x}<1,
 \qquad
 J''(x)=\frac{2(x\coth x-1)}{\sinh^2x}>0.
\end{equation}
The first strict inequalities follow respectively from
$\sinh x\cosh x>x$ and $1-e^{-2x}<2x$.
For $q=t/2$, differentiation of \eqref{oc:hyperbolic} yields
\begin{equation}\label{oc:derivative}
 F_\gamma'(t)=
 \frac{J(q)-\delta J(\delta q)-\gamma J(\gamma q)}t>0.
\end{equation}
Also $F_\gamma(0+)=0$ and $F_\gamma(\infty)=\infty$.
This proves the asserted bijection and, in particular, $G_\gamma(z)<1$.

\medskip\noindent\textbf{Step 3: compute the minimum.}
Weighted AM--GM, applied to $Au/\delta$ and $Bv/\gamma$, gives
\[
 Au+Bv\ge
 \left(\frac{Au}{\delta}\right)^\delta
 \left(\frac{Bv}{\gamma}\right)^\gamma
 \ge\frac{zA^\delta B^\gamma}{\delta^\delta\gamma^\gamma}.
\]
This is the lower bound in \eqref{oc:minimum}.
Equality occurs at
$u_* =\delta G_\gamma(z)/A$ and
$v_* =\gamma G_\gamma(z)/B$.
Concavity of $1-e^{-x}$ gives $A\ge\delta$ and $B\ge\gamma$,
so these two coordinates lie in $(0,1)$; their weighted product is $z$.
\end{proof}

\paragraph{Why the next estimate is needed.}
If $\gamma$ is very close to zero or one, then $h$ is small. A bounded
error in a logarithm can become a very large multiplicative error after
division by $h$. The next lemma therefore tracks the correction on the
scale $h$, uniformly as the order approaches an integer.

\begin{lemma}[Uniform overlap correction]\label{oc:correction}
Write
\[
 H(\gamma)=-\gamma\log\gamma-\delta\log\delta,
 \qquad \ell(t)=-\log(1-e^{-t}),
 \qquad
 D_\gamma(t)=\delta\ell(\delta t)+\gamma\ell(\gamma t)-\ell(t).
\]
For every $t>0$,
\begin{equation}\label{oc:Fidentity}
 F_\gamma(t)=ht-H(\gamma)+D_\gamma(t),
 \qquad
 2[J(t/4)-1]\le\frac{F_\gamma(t)}h\le t.
\end{equation}
For every fixed $\eta>0$, uniformly over $t\ge\eta$ and
$0<\gamma<1$,
\begin{equation}\label{oc:Destimate}
 \frac{D_\gamma(t)}h
 =\log\left(1+\frac1{ht}\right)+O_\eta(1).
\end{equation}
\end{lemma}
\begin{proof}
The identity follows directly from \eqref{oc:definitions}.
The numerator in \eqref{oc:derivative} equals
\[
 \delta\int_{\delta q}^q J'(s)\,ds+
 \gamma\int_{\gamma q}^q J'(s)\,ds.
\]
It is at most $2hq$ by \eqref{oc:J}.
Choosing the larger of $\delta,\gamma$, its corresponding integral is
at least $hqJ'(q/2)$. Thus
$J'(t/4)/2\le F_\gamma'(t)/h\le1$;
integration proves \eqref{oc:Fidentity}.

Let $a=\min\{\gamma,\delta\}$ and $b=1-a\ge1/2$.
The exact identity
\begin{equation}\label{oc:Ddifference}
 \frac{D_\gamma(t)}h=
 \frac{\ell(bt)-\ell(t)}a+
 \frac{\ell(at)-\ell(t)}b
\end{equation}
has first term between zero and
$t/(e^{bt}-1)\le1/b\le2$.
The second term differs from $\ell(at)$ by
$(a/b)\ell(at)-(1/b)\ell(t)$.
For $t\ge\eta$, the latter difference is bounded in absolute value by
$2\ell(\eta)+\log(1+2/\eta)$: use
$\ell(x)\le\log(1+1/x)$ and the increasing function
$a\mapsto a\log(1+c/a)$.
Finally,
\[
 \left|\ell(x)-\log(1+1/x)\right|\le\log2,
 \qquad
 \left|\log(1+1/(at))-\log(1+1/(ht))\right|\le\log2.
\]
The first bound follows from
$x/(1+x)\le1-e^{-x}\le\min\{x,1\}$, and the second from $h=ab$.
These estimates prove \eqref{oc:Destimate}.
\end{proof}

\paragraph{From overlap to an allowable interval size.}
We choose $z_*$ so that the minimum of the linear combination equals
$e^{-E}$. The relevant interval length will involve $z_*^{-1/h}$.
The following calculation retains both its exponential factor and its
prefactor; dropping the prefactor would lose uniformity near order one.

\begin{theorem}[Capacity with the exact accuracy]\label{oc:capacity}
Fix $\varepsilon>0$. Let $\alpha>1$ be noninteger, and define
\[
 \beta=\alpha-1,\quad
 \gamma=\alpha-\lfloor\alpha\rfloor,\quad
 \delta=1-\gamma,\quad h=\gamma\delta,\quad
 E=2\varepsilon\beta,\quad
 \Psi=\gamma^{1+1/\delta}\delta^{1+1/\gamma}.
\]
The unique solution $z_*\in(0,1)$ of $G_\gamma(z_*)=e^{-E}$ satisfies
\begin{equation}\label{oc:capacity-evaluation}
 z_*^{-1/h}
 =\Theta_\varepsilon\!\left(
 \frac\beta\alpha
 (\gamma^\gamma\delta^\delta)^{-1/h}e^{E/h}\right)
 =\Theta_\varepsilon\!\left(
 \frac{\beta h}{\alpha\Psi}e^{E/h}\right).
\end{equation}
All constants are uniform over noninteger $\alpha>1$.
\end{theorem}
\begin{proof}
Write $z_*=e^{-ht_*}$, so $F_\gamma(t_*)=E$.
Since $\beta\ge\gamma$,
$E/h=2\varepsilon\beta/(\gamma\delta)\ge2\varepsilon$.
Equation \eqref{oc:Fidentity} therefore gives $t_*\ge2\varepsilon$.
Convexity of $J$ implies that $[J(t/4)-1]/t$ is increasing. Consequently,
with $c_\varepsilon=[J(\varepsilon/2)-1]/\varepsilon>0$,
\[
 c_\varepsilon ht_*\le E\le ht_*.
\]
Equation \eqref{oc:Destimate}, with $\eta=2\varepsilon$, now gives
\[
 D_\gamma(t_*)/h=\log(1+1/E)+O_\varepsilon(1),
 \qquad
 e^{-D_\gamma(t_*)/h}
 =\Theta_\varepsilon\!\left(\frac{E}{1+E}\right)
 =\Theta_\varepsilon\!\left(\frac\beta\alpha\right).
\]
Exponentiating
$t_*=E/h+H(\gamma)/h-D_\gamma(t_*)/h$
proves the first equality in \eqref{oc:capacity-evaluation}.
The second uses
$(\gamma^\gamma\delta^\delta)^{1/h}=\Psi/h$.
\end{proof}

\paragraph{How the scalar calculation predicts the sample scale.}
At this point we have solved only a two-number optimization problem.
The following calculation explains which cutoffs a polynomial should
achieve. It is conditional: the existence and the shape of that polynomial
are proved later in the section.

\paragraph{Conditional polynomial-capacity calibration.}
The algebraic consequence for a common polynomial construction is as
follows. Suppose its lower cutoff and endpoint capacity obey
\begin{equation}\label{oc:conditional-capacity}
 a\asymp_\varepsilon\frac\beta{\alpha k},\qquad
 \frac ba\asymp_\varepsilon
 \frac{m^2L^2}{h z_*^{1/h}},\qquad m\asymp\alpha.
\end{equation}
Here $m$ and $L$ denote the two endpoint degree scales; these relations
are assumptions about the construction, not conclusions of the scalar
lemmas. Then \eqref{oc:capacity-evaluation} gives, with both factors
$\beta/\alpha$ retained,
\begin{align}
 b&\asymp_\varepsilon
 \underbrace{\frac\beta{\alpha k}}_{\text{lower cutoff}}
 \frac{m^2L^2}h
 \underbrace{\frac{\beta h}{\alpha\Psi}e^{E/h}}_
             {\text{overlap normalization}}\notag\\
 &=\frac{m^2}{\alpha^2}\frac{\beta^2L^2}{k\Psi}e^{E/h}
 \asymp_\varepsilon\frac{\beta^2L^2}{k\Psi}e^{E/h}.
 \label{oc:b-scale}
\end{align}
If also $b\asymp L/n$, the resulting deterministic balance is
\begin{equation}\label{oc:n-scale}
 n\asymp_\varepsilon\frac{k\Psi}{\beta^2L}e^{-E/h}.
\end{equation}
Putting $K=\log k$, $n=\alpha k^{1-1/\alpha}x$, and
$L\asymp\alpha(1+\log x)$, with $x\ge1$, rewrites this balance as
\begin{equation}\label{oc:root-balance}
 x(1+\log x)\asymp_\varepsilon Q^2,
 \qquad
 Q^2=\frac{\Psi}{\alpha^2\beta^2}
       \exp\left(\frac K\alpha-\frac Eh\right).
\end{equation}
Thus the scalar calculation produces the original-accuracy exponential
factor and the two near-one normalization factors in the root.
It does not prove \eqref{oc:conditional-capacity} for actual polynomial
profiles, nor supply the stochastic regularizer by itself. Those tasks are
verified in Sections~\ref{upcal:section}--\ref{cu:section}.

\subsection{One parameter rule for the regularized upper bound}
\label{upcal:section}

The deterministic balance above does not yet account for random noise.
The root definition below includes a further term that supplies room for
that noise. It specifies one collection of parameters for every
noninteger order. The proof now checks that these numerical choices have
the properties required by the polynomial construction and its later
variance analysis. In particular, increasing the final sample multiplier
$D$ makes the lower cutoff $a$ sufficiently small without changing the
requested accuracy $\varepsilon$.

\begin{lemma}[Regularized-root parameter rule]\label{upcal:parameters}
Fix $0<\varepsilon<1/2$. For a noninteger $1<\alpha<\log k$, set
\begin{gather*}
 \beta=\alpha-1,\quad m=\lfloor\alpha\rfloor,\quad
 \gamma=\alpha-m,\quad\delta=1-\gamma,\quad h=\gamma\delta,\quad
 E=2\varepsilon\beta,\\
 \Psi=\gamma^{1+1/\delta}\delta^{1+1/\gamma},\qquad
 Q^2=\frac{\Psi}{\alpha^2\beta^2}
       \exp\{(\log k)/\alpha-E/h\},\\
 X=\inf\left\{x\ge1:\frac{\alpha h}{\beta}x^\alpha
       \left[1+\log\frac{x(1+\log x)}{Q^2}\right]_+\ge1\right\},
 \qquad v=\frac\beta\alpha X^{-\alpha}.
\end{gather*}
The constants in the following definitions are design parameters. They
will be fixed in an order stated below, using $\varepsilon$ only; they do
not depend on $k$ or $\alpha$. Fix positive constants
$C_h,c_b,\lambda,c_L,u_0$, $C_L\ge1$, and
$0<c_\omega<1$. For $D\ge1$, define
\begin{gather}
 \Delta=\min\{v+C_hh,E/2\},\qquad E_0=E-\Delta,
 \qquad u=\min\{v+h,u_0\},\qquad \omega=e^{-c_\omega E_0},
 \label{upcal:width}\\
 n=D\alpha k^{1-1/\alpha}X,\qquad
 L=\lceil\alpha(C_L+c_L\log X)\rceil,\qquad
 b=\frac{c_bL}{n},\qquad
 a=\lambda\frac{b h z^{1/h}}{m^2L^2},
 \label{upcal:cutoffs}
\end{gather}
where $z\in(0,1)$ is the unique solution of $G_\gamma(z)=e^{-E_0}$.
Put
\begin{equation}\label{upcal:slack}
 \sigma=1+\log\frac{X(1+\log X)}{Q^2}-\frac vh,
 \qquad
 w=\sigma+\frac{(v-E/2)_+}{h}.
\end{equation}
Then $\sigma\ge0$, and $\sigma=0$ if $X>1$.
There are constants $r_-,r_+>0$, independent of $k,\alpha,D$ and
depending only on the other fixed constants and $\varepsilon$, such that
\begin{equation}\label{upcal:actual-cutoff}
 \frac{r_-}{D}e^{-w}\le
 \frac{a}{\beta/(\alpha k)}
 \le\frac{r_+}{D}e^{-w}.
\end{equation}
In addition,
\begin{equation}\label{upcal:convolution}
 \frac{am^2L^2}{b}=\lambda h z^{1/h}\le\lambda/4.
\end{equation}
Thus any fixed requirement $a\le c_a\beta/(\alpha k)$ holds by
increasing $D$. Fixed choices of $C_L$ and then $D$ also ensure
$a\le b\le1$ and $nb\ge2B\alpha$ for any prescribed fixed pilot
constant $B$.
\end{lemma}
\begin{proof}
\noindent\textbf{Step 1: extract the identity supplied by the root.}
The root-defining expression is continuous and strictly increasing
once its positive part is nonzero. Hence the inequality at its infimum
is an equality unless $X=1$. This proves the assertions about $\sigma$
and the exact identity
\begin{equation}\label{upcal:root-identity}
 \frac{Q^2}{X}=(1+\log X)e^{1-v/h-\sigma}.
\end{equation}
\medskip\noindent\textbf{Step 2: evaluate the actual lower cutoff.}
Since $E/2\le E_0\le E$, Theorem~\ref{oc:capacity} applies with
accuracy in $[\varepsilon/2,\varepsilon]$, uniformly on that compact
interval. It gives
\[
 z^{1/h}\asymp_\varepsilon
       \frac{\alpha\Psi}{\beta h}e^{-E_0/h}.
\]
Substitution into \eqref{upcal:cutoffs}, followed by
\eqref{upcal:root-identity}, yields
\begin{equation}\label{upcal:cutoff-arithmetic}
 \frac{a}{\beta/(\alpha k)}
 \asymp_\varepsilon
 \frac{\lambda c_b}{D}
 \frac{\alpha^3(1+\log X)}{m^2L}
 e^{1-\sigma-(v-\Delta)/h}.
\end{equation}
The degree factor is between
$1/\max(C_L+1,c_L)$ and $4/\min(C_L,c_L)$.
Also $v-\Delta=\max\{-C_hh,v-E/2\}$, so
\[
 e^{-w}\le e^{-\sigma-(v-\Delta)/h}\le e^{C_h}e^{-w}.
\]
These bounds prove \eqref{upcal:actual-cutoff}.
\medskip\noindent\textbf{Step 3: verify the remaining scale restrictions.}
Equation \eqref{upcal:convolution} is immediate from the definition.
Choose $C_L$ so that $\lambda/(4m^2L^2)\le1$ and
$c_bC_L\ge2B$. Finally,
\[
 b\le\frac{c_b}{D}
       \frac{C_L+1+c_L\log X}{X}
 \le\frac{c_b(C_L+1+c_L)}D,
\]
which proves the remaining restrictions after increasing $D$.
\end{proof}

\paragraph{A uniform lower bound for the error budget.}
The actual value of $a$ can be smaller than its nominal scale. A later
variance calculation is easier to express using a baseline depending
only on $D$ and the order. The next lemma shows that any discrepancy
between that nominal baseline and the actual one is paid by the
regularization margin $\kappa uP$. It also checks that the reserved
logarithmic width produces a multiplicative gap of that size.

\begin{lemma}[Actual baseline and the reserved width]
\label{upcal:baseline}
Use Lemma~\ref{upcal:parameters} with $\alpha\ge3/2$.
Let $C_0\ge1$, $0<d_\varepsilon\le1$, and define
\begin{gather*}
 \kappa=e^{-E},\qquad
 u_* =\min\{\varepsilon/2,1/3,u_0\},\qquad
 \rho_-=r_-/3,\quad\rho_+=r_+,\\
 E_{\rm base}=\min\{\beta,1\}\,d_\varepsilon C_0^m
                      (\rho_-/D)^\beta,\qquad
 E_{\rm actual}=\min\{\beta,1\}\,d_\varepsilon C_0^m(ka)^\beta.
\end{gather*}
If $D\ge\rho_-C_0^2e^{2\varepsilon}$, then
\begin{equation}\label{upcal:baseline-transfer}
 E_{\rm base}\le E_{\rm actual}+\kappa u/u_*.
\end{equation}
Consequently, for $P_*=k^{1-\alpha}$, $P=\sum_i p_i^\alpha$,
and any $\mathcal B\ge0$,
\begin{align}
 E_{\rm base}P_*+\kappa uP+\kappa\mathcal B
 \le(1+1/u_*)
       \{E_{\rm actual}P_*+\kappa uP+\kappa\mathcal B\}.
 \label{upcal:margin-transfer}
\end{align}
Furthermore, for a fixed $c>0$ depending only on the stated constants,
\begin{equation}\label{upcal:reserve}
 \Delta\ge cu,\qquad e^{-E_0}-e^{-E}\ge c\kappa u.
\end{equation}
\end{lemma}
\begin{proof}
Here $\beta\ge1/2$, $\beta/\alpha\ge1/3$, and $m\le2\beta$.
Thus $ka\le\rho_+/D$, and $w=0$ implies $ka\ge\rho_-/D$.
The condition on $D$ gives $E_{\rm base}\le\kappa$.
If $w>0$, then either $\sigma>0$, which forces $X=1$ and
$v\ge1/3$, or $v>E/2=\varepsilon\beta\ge\varepsilon/2$.
Equivalently, the single bound
\[
 \mathbf1\{w>0\}\le u/u_*
\]
holds. Combining it with the cutoff lower bound when $w=0$ proves
\eqref{upcal:baseline-transfer}; using $P_*\le P$ proves
\eqref{upcal:margin-transfer}.
Finally,
\[
 \Delta\ge
 \min\{\min(1,C_h)(v+h),\varepsilon/2\}\ge cu.
\]
Since $e^{-E_0}-e^{-E}=\kappa(e^\Delta-1)$, this proves
\eqref{upcal:reserve}.
\end{proof}

\paragraph{Choice of the baseline constant and count multiplier.}
The factor $\min\{\beta,1\}$ is between $1/2$ and one in
Lemma~\ref{upcal:baseline}; it also makes the actual baseline vanish
at the order-one endpoint. The prefactor $d_\varepsilon$ can be fixed
after $C_0$, for example at a constant multiple of
$c_\varepsilon/C_0^2$, to fit the deterministic baseline bound.
This changes only fixed factors in the noise normalization.

There is an important arithmetic point if $D$ is chosen last.
For a low degree $1\le j\le m$, write $\theta=j/\alpha$ and
$\eta=\alpha-j$. A coefficient bound with $ka\le\rho_+/D$,
combined with the factorial second-moment estimate and normalization
by $(E_{\rm base}+\kappa uP/P_*)^2$, gives the exact power
\begin{equation}\label{upcal:D-power}
 D^{-2\eta-j+\beta(2-\theta)}
       =D^{-2+j/\alpha}\le D^{-1}.
\end{equation}
Here $u\ge\min\{1/3,u_0\}X^{-\alpha}$ absorbs the corresponding
$X$ powers. The remaining baseline factors are bounded by
$4d_\varepsilon^{-2}C_0^{-m}e^{C\alpha}$ after the factorial
coefficient sums are evaluated. For the high-degree block the analogous
factor is
\[
 \frac{D^{-\alpha}}{E_{\rm base}}
 \le2d_\varepsilon^{-1}D^{-1}C_0^{-m}\rho_-^{-\beta}.
\]
Thus one first fixes profile and pilot constants, chooses $C_L$ and
$c_L$ for the degree bounds, determines $\rho_-,\rho_+$, and chooses
$C_0$ to absorb the exponential coefficient costs. One then fixes
$d_\varepsilon$ and chooses $D$ last. Substituting $c_a\asymp D^{-1}$
into a cruder bound and subsequently choosing $D$ in terms of $c_a$
would lose this arithmetic. The displayed calculations are normalization identities. The coefficient
and conditional-variance estimates that use them are proved in
Sections~\ref{sec:coupled-bias-bound} and~\ref{sec:conditional-count}.

\paragraph{Why order one requires a separate calculation.}
As $\alpha\downarrow1$, the target power sum approaches the known value
$\sum_i p_i=1$, while the admissible multiplicative error is proportional
to $\beta=\alpha-1$. Bounds that merely stay finite are insufficient:
the relevant approximation errors must vanish at least linearly in
$\beta$. The next lemma records precisely that cancellation for the
same parameter rule.

\begin{lemma}[Order-one calibration of the same parameters]
\label{upcal:near-one}
Use Lemma~\ref{upcal:parameters} with $1<\alpha\le3/2$.
Let $N\asymp L$, with fixed comparison constants, and set
$t=b/N^2$. For the degree-one reverse profile, let
$\tau=c(1+\beta)a/\beta$ with fixed $c>0$.
Then
\begin{equation}\label{upcal:near-one-scales}
 \frac{k}{nL}\le C_\varepsilon/D,\qquad
 kt\le C_\varepsilon/D,\qquad
 c_\varepsilon\le\tau/t\le C_\varepsilon.
\end{equation}
The exact amplitudes $A,B$ of \eqref{oc:definitions} satisfy
$B\le C_\varepsilon\beta$ and $1-A\le C_\varepsilon\beta$.
In particular, $A-\omega=O_\varepsilon(\beta)$, and
\begin{equation}\label{upcal:near-one-margins}
 \min\{1-\omega,\omega-\kappa\}
       \ge c_{\varepsilon,c_\omega}\beta.
\end{equation}
All constants may depend on the fixed profile constants in
Lemma~\ref{upcal:parameters}.
\end{lemma}
\begin{proof}
Here $m=1$, $\gamma=\beta$, and $\delta=1-\beta\ge1/2$.
Cancellation of the near-one factors gives
\begin{equation}\label{upcal:near-one-Q}
 \frac{Q^2}{k^{1/\alpha}}
 =\frac{\beta^{\beta/\delta}\delta^{1+1/\beta}}{\alpha^2}
       e^{-2\varepsilon/\delta}\asymp_\varepsilon1.
\end{equation}
The logarithm is uniformly bounded, using
$\beta|\log\beta|\le1/e$ and
$\log(1-\beta)\ge-\beta/(1-\beta)$.
Equation~\eqref{upcal:root-identity} gives
$Q^2\le eX(1+\log X)$, including $X=1$. Therefore
\[
 k^{1/\alpha}\le C_\varepsilon X(1+\log X),
 \qquad
 \frac{k}{nL}\le
 \frac{C_\varepsilon(1+\log X)}{D\alpha L}
 \le C_\varepsilon/D.
\]
This also proves the bound for $kt$.

Write $z=e^{-ht_*}$, so $F_\gamma(t_*)=E_0$.
Since $E_0/h\in[\varepsilon,4\varepsilon]$,
\eqref{oc:Fidentity} yields
\[
 \varepsilon\le t_*\le T_\varepsilon,
 \qquad J(T_\varepsilon/4)=1+2\varepsilon.
\]
In particular $z^{1/h}$ lies in a fixed positive compact interval.
Equation \eqref{upcal:cutoffs} gives the exact relation
\[
 \frac\tau t
 =c\lambda\frac{N^2}{L^2}\alpha\delta z^{1/h}.
\]
Since $3/4\le\alpha\delta=1-\beta^2\le1$, this proves the
remaining assertion in \eqref{upcal:near-one-scales}.
Finally,
\[
 B=\frac{1-e^{-\beta t_*}}{1-e^{-t_*}}
       \le C_\varepsilon\beta,
 \qquad 1-A=e^{-\delta t_*}B\le C_\varepsilon\beta.
\]
The estimates for $\omega$ follow from
$\varepsilon\beta\le E_0\le2\varepsilon\beta$ and
$0<c_\omega<1$.
\end{proof}

\paragraph{Correctly centered localized bias.}
Suppose the common coupled polynomial satisfies
$\sup_{0\le p\le b}|Q(p)-\omega p^\alpha|
\le C_\varepsilon\beta t^\alpha$ in the interval of
Lemma~\ref{upcal:near-one}. Then its aggregate light-domain bias,
divided by $P_*=k^{1-\alpha}$, is at most
\[
 C_\varepsilon\beta(kt)^\alpha
       \le C_\varepsilon\beta/D.
\]
Increasing $D$ places this inside the original-accuracy margins in
\eqref{upcal:near-one-margins}, without changing $\varepsilon$ in the
root. This assertion uses the target $\omega P_\alpha$.

\paragraph{The remaining compact parameter range.}
The ratio $E/h$ distinguishes a long overlap interval from one with
bounded width. When it is bounded and $\alpha\ge3/2$, all order
parameters remain in a compact set. Ordinary fixed-constant bounds are
then sufficient. This observation simplifies later proofs without
introducing a different estimator.

\begin{lemma}[Bounded interpolation width]
\label{upcal:bounded-width}
Fix $0<M<\infty$. Within the same parameter rule, suppose
$\alpha\ge3/2$ and $E/h\le M$. Then, with constants depending on
$\varepsilon,M$ and the fixed profile parameters,
\begin{equation}\label{upcal:bounded-width-conclusions}
 Q^2\asymp_{\varepsilon,M}k^{1/\alpha},\qquad
 \frac{k}{nL}\le C_{\varepsilon,M}/D,
 \qquad kt\le C_{\varepsilon,M}/D,
 \qquad u\ge c_{\varepsilon,M}>0,
\end{equation}
where $t=b/N^2$ and $N\asymp L$.
\end{lemma}
\begin{proof}
The hypothesis and $h\le1/4$ imply
\[
 1/2\le\beta\le M/(8\varepsilon),\qquad
 h\ge E/M\ge\varepsilon/M.
\]
Thus $\alpha$ is bounded above and both $\gamma,\delta$ are bounded
away from zero. The positive factor
$\Psi e^{-E/h}/(\alpha^2\beta^2)$ lies in a fixed positive compact
interval, proving the first assertion.
Equation \eqref{upcal:root-identity}, including its clipped-root
inequality, gives $Q^2\le eX(1+\log X)$.
The proof of \eqref{upcal:near-one-scales} then applies verbatim to
$k/(nL)$ and $kt$.
Finally, $u=\min\{v+h,u_0\}\ge\min\{\varepsilon/M,u_0\}>0$.
This verifies additional margins of the same construction; it does not
select a different estimator.
\end{proof}

\subsection{Constructing the polynomial explicitly}
\label{sec:coupled-bias-bound}

\paragraph{What the construction must accomplish.}
We need a polynomial in $p$, even though $p^\alpha$ itself is not a
polynomial. The identity $\alpha=m+\gamma=m+1-\delta$ suggests two
representations: multiply a function of $a/p$ behaving like
$(a/p)^\gamma$ by $p^\alpha$, or multiply a function of $p/b$
behaving like $(p/b)^\delta$ by $p^\alpha$. The resulting powers of
$p$ are integers. The first term controls the lower endpoint and the
second controls the upper endpoint. Smooth polynomial cutoffs join them.

We give one polynomial construction. All constants below may depend on the
fixed accuracy $0<\varepsilon<1/2$, but not on the alphabet size or the order.
Write
\[
 \beta=\alpha-1,\quad m=\lfloor\alpha\rfloor,\quad
 \gamma=\alpha-m,\quad\delta=1-\gamma,\quad h=\gamma\delta,
 \qquad E=2\varepsilon\beta.
\]
For the deterministic bound first suppose $\alpha\ge3/2$ and
$E/2\le E_0\le E$. Put $\kappa_0=e^{-E_0}$ and $\omega=e^{-E_0/8}$.
For $0<z<1$ define
\begin{align*}
 A(z)&=\frac{1-z^{1/\gamma}}{1-z^{1/h}},&
 B(z)&=\frac{1-z^{1/\delta}}{1-z^{1/h}},&
 G_\gamma(z)&=\frac{zA(z)^\delta B(z)^\gamma}
 {\delta^\delta\gamma^\gamma}.
\end{align*}
Choose $z$ from $G_\gamma(z)=\kappa_0$. In the following proof we use
$z\le\kappa_0$, $A,B\le1$, and $A,B\ge1-z$; these also follow directly
from these formulas and the monotonicity properties of $G_\gamma$.
Fix
\[
 q=e^{-\varepsilon/16},\qquad q_L\in(q,1).
\]
In particular $q>\kappa_0$ and $q\omega>\kappa_0$ with a fixed relative
margin, because $E_0\ge\varepsilon/2$.

\paragraph{The elementary building block.}
Let $U_j$ be the Chebyshev polynomial of the second kind, defined by
$U_0(t)=1$, $U_1(t)=2t$, and
$U_{j+1}(t)=2tU_j(t)-U_{j-1}(t)$.
The kernel below is concentrated near $x=0$, on the scale $N^{-2}$.
The normalized integral $J_s$ is a polynomial times $x^s$ that rises
from zero to one. The variable $y=N^2x$ measures position relative to
the concentration scale. These facts allow us to build endpoint
approximations while controlling both their size and their coefficients.

\begin{lemma}[Elementary endpoint-kernel estimates]
\label{lem:coupled-elementary-kernel}
For $N\ge2$ and $d\in\{2,6\}$, set
\[
 K_{N,d}(x)=\left[U_{N-1}(\sqrt{1-x})/N\right]^{2d},
 \quad y=N^2x,\quad I_s=\int_0^1x^{s-1}K_{N,d}(x)\,dx.
\]
For $0<s\le1$, the following estimates are uniform in $N,s$:
\begin{gather}
 0\le K_{N,d}(x)\le\min(1,y^{-d}),\qquad
 e^{-Cy}\le K_{N,d}(y/N^2)\le e^{-cy}\quad(y\le y_0),
 \label{eq:coupled-elementary-kernel-local}\\
 1-K_{N,d}(y/N^2)\le Cy\quad(y\le1),\qquad
 \sup_{y\ge y_0}K_{N,d}(y/N^2)\le\rho<1,
 \label{eq:coupled-elementary-kernel-gap}\\
 sN^{2s}I_s=\exp(O(s)),\qquad
 J_s(x):=I_s^{-1}\int_0^x t^{s-1}K_{N,d}(t)\,dt
       =y^s\exp(O(s))\quad(y\le1),
 \label{eq:coupled-elementary-integral}\\
 1-J_s(x)\le Cs y^{s-d}\quad(y\ge1),\qquad
 J_s(x)\ge e^{-Cs}\quad(y\ge1),
 \label{eq:coupled-elementary-tail}\\
 \max_{0\le x\le1}x^sK_{N,d}(x)
       =N^{-2s}s^s\exp(O(s)).
 \label{eq:coupled-elementary-sharp}
\end{gather}
The derivative of the normalized kernel density with respect to
$s\in[1/2,1]$ is bounded in $L^1$; in particular
$0\le J_\delta-J_1\le C(1-\delta)$ for $1/2\le\delta\le1$.
The kernel and the polynomial $J_s(x)/x^s$ have coefficient sums at
most $C^{N+1}$.
\end{lemma}
\begin{proof}
Put $\theta=\arcsin\sqrt x$. The elementary identity
\[
 U_{N-1}(\cos\theta)=\frac{\sin(N\theta)}{\sin\theta}
\]
gives the tail bound because $N\sin\theta=\sqrt y$.
The modulus of the quotient divided by $N$ is at most one: it is the
modulus of an average of $N$ unit complex numbers from the geometric
sum for $\sin(N\theta)/\sin\theta$. This inequality is strict for
$0<\theta\le\pi/2$ when $N\ge2$.
For $y$ sufficiently small, Taylor's formula with its uniform remainder
gives
\[
 \frac{\sin(N\theta)}{N\sin\theta}
 =1-\frac{(N^2-1)\theta^2}{6}+O(N^4\theta^4)
 =1-\frac{1-N^{-2}}6y+O(y^2).
\]
Since $1-N^{-2}\in[3/4,1)$, this proves the two local exponential
bounds and the linear deficit. The lower exponential bound also holds
up to $y=1$ after changing its constant: on that interval
$N\theta\le N\arcsin(1/N)\le\pi/3$, so the sine quotient is bounded
away from zero.

For the gap, first use $K\le y^{-d}$ beyond a fixed large $R$.
On $[y_0,R]$, the quotient converges uniformly as $N\to\infty$ to
$\sin\sqrt y/\sqrt y$, whose modulus is strictly less than one.
For each of the finitely many remaining $N$, continuity and the strict
geometric-sum inequality give the same conclusion on its nonempty
scaled interval. Taking their maximum proves a uniform $\rho<1$.

After rescaling, let
$M_s=N^{2s}I_s=\int_0^{N^2}y^{s-1}K_{N,d}(y/N^2)\,dy$.
The upper bounds give
\[
 sM_s\le1+s\int_1^\infty y^{s-1-d}\,dy\le1+Cs.
\]
For the lower bound, integrate $e^{-Cy}$ on $[0,1]$ and use Jensen's
inequality under the probability density $sy^{s-1}$:
\[
 sM_s\ge s\int_0^1y^{s-1}e^{-Cy}\,dy
 \ge\exp[-Cs/(s+1)]\ge e^{-Cs}.
\]
This is why the normalization error is $O(s)$ rather than a fixed
constant. The same argument after rescaling $[0,y]$ proves the
truncated-integral estimate. Integrating the inverse-power bound from
$y$ to infinity proves \eqref{eq:coupled-elementary-tail}.

For the sharp maximum, evaluate at $y=c_0s$ for a fixed small $c_0$
to obtain the lower bound. On the local interval, the upper bound is
at most $\max_{y\ge0}y^se^{-cy}=(s/(ce))^s$.
Outside that interval the uniform gap and inverse-power tail imply
$\sup y^sK\le\rho'<1$ for all sufficiently small $s$.
This follows by first fixing a large tail cutoff and then decreasing
$s$ on the remaining compact interval. Since $s^s\to1$, this region
cannot violate the asserted upper bound. For $s$ bounded below, the
trivial uniform bound on $y^sK$ is absorbed into $e^{Cs}$.
The same argument applies to any kernel with a local exponential
bound, a uniform gap away from zero, and a tail $Cy^{-a}$ for fixed
$a>1$; its constants may depend on the constants in those bounds.

For the exponent derivative, the rescaled normalized density is
$y^{s-1}K(y/N^2)/M_s$. Its mean absolute logarithm is uniformly bounded
for $s\in[1/2,1]$: integrate $y^{-1/2}|\log y|$ below one and use the
inverse-power tail above one. Differentiating a normalized integral is
a covariance with $\log y$, giving the asserted uniform derivative
bound. The sign follows because the likelihood ratio for the smaller
exponent is decreasing in $y$.

Finally, $U_{n+1}(t)=2tU_n(t)-U_{n-1}(t)$ bounds coefficient sums by an
exponential in $N$. Taking a fixed even power and replacing its even
monomials in $\sqrt{1-x}$ by powers of $1-x$ preserves such a bound.
For integration, the $j$th coefficient is divided by $j+s$ and then
multiplied by $I_s^{-1}\le CsN^{2s}$; the factor $s$ cancels the only
possible singularity at $j=0$. Polynomial factors in $N$ are absorbed
into $C^{N+1}$.
\end{proof}

\paragraph{The two endpoint profiles.}
A \emph{profile} here is simply a function taking values in $[0,1]$.
The low profile must rise with the correct fractional power near zero,
remain above a fixed positive level after that transition, and retain
a definite deficit below one away from its maximum. The high profile
has a similar local rise but approaches a prescribed level $\omega$
in its tail. Both are a fractional power multiplied by an ordinary
polynomial. These properties will make their combination a polynomial
in the original probability $p$.

The low profile $u_s$ is a normalized raised-floor Chebyshev profile of
ordinary polynomial degree at most $m$ after division by $x^s$.
Let $\ell\asymp m$ denote its endpoint resolution. A fixed floor parameter
strictly between $q_L$ and one gives the following bounds, with constants
depending only on $\varepsilon$:
\begin{align}
 0\le u_s\le1,\qquad
 u_s(x)&\le e^{Cs}(\ell^2x/s)^s,\label{eq:coupled-low-upper}\\
 u_s(x)&\ge\min\{q,e^{-Cs}(\ell^2x/s)^s\},\label{eq:coupled-low-lower}\\
 1-u_s(x)&\ge c\min(\ell^2x,1)
 \quad(s\le s_0,\ \ell^2x\ge Cs).\label{eq:coupled-low-deficit}
\end{align}
Its coefficient sum is at most $C^{m+1}$. These are fixed-floor bounds:
the floor is not allowed to tend to one with $\alpha$.

For completeness the degree-one completion has no boundary ambiguity.
Choose the unique $c_s\in(s/(1+s),1)$ satisfying
\[
 (1-c_s)c_s^s(1+s)^{1+s}/s^s=q_L,
\]
and set
\[
 u_s(x)=\frac{x^s(1-c_sx)}{a_s},\qquad
 a_s=\frac{s^s}{c_s^s(1+s)^{1+s}}.
\]
Its maximum is attained at $s/[c_s(1+s)]<1$, so $a_s$ is the exact
normalizer. In particular $c_s\to1-q_L>0$ as $s\downarrow0$.
Continuity of the unique solution makes $c_s$ uniformly bounded above
and away from zero for $0<s\le1/2$, with constants depending only on
$\varepsilon$. Thus the common near-one tangent-scale comparison is
valid for this variable $c_s$ as well.
The finitely many degrees below the Chebyshev construction are covered
by this same completion; their fixed resolution ratios are absorbed in
$e^{\pm Cs}$. For larger degrees one may use
$K=[U_{\ell-1}(\sqrt{1-x})/\ell]^4$, its normalized integral $J_s$,
and
\[
 u_s=x^sK/a_s+f_LJ_s(1-K),\qquad
 a_s=\max_x\frac{x^sK}{1-f_LJ_s(1-K)},\quad q_L<f_L<1.
\]
Its degree is at most $4(\ell-1)$, so $\ell=\lfloor m/4\rfloor+1$
is admissible. We verify all three low-profile properties explicitly.
Write $d_L=1-f_L>0$ and $D_s=1-f_LJ_s(1-K)$.
Then
\[
 \frac{K}{D_s}\le\frac{K}{d_L+f_LK}.
\]
The right side has local exponential decay with a positive
$d_L$-dependent constant, a uniform gap away from zero, and a
$y^{-2}$ tail. Lemma~\ref{lem:coupled-elementary-kernel} therefore gives
$a_s\le\ell^{-2s}s^s e^{C_\varepsilon s}$.
Since $D_s\le1$, the ordinary sharp maximum gives the reverse bound.
The defining maximum proves $u_s\le1$ exactly.
For $y=\ell^2x\le s$, its first term is
$(y/s)^s e^{O_\varepsilon(s)}$ and its added term is at most $Cy^{s+1}$;
relative to $(y/s)^s$ the latter is $O(s)$. This proves the sharp
bounds there; the global upper follows from $u_s\le1$.

If $y\ge C_fs$, sharp normalization gives $x^s/a_s\ge1$ after
increasing $C_f$. Hence
\[
 u_s\ge K+f_LJ_s(1-K)
 \ge f_L-(1-J_s)(1-K).
\]
For $y\le1$, the product on the right is at most
$Cs\,y(1+|\log y|)\le Cs$; for $y\ge1$ it is again at most $Cs$.
Thus $u_s\ge q$ for all sufficiently small $s$ on this range.
For the other $s$, take $C_q$ sufficiently large and use
$J_s\to1$, $K\to0$ uniformly in the scaled tail to get $u_s\ge q$
when $y\ge C_qs$.
On $s\le y\le C_qs$, if $s\le y_0/C_q$, the local kernel lower bound
makes the first term at least $(y/s)^s e^{-C_\varepsilon s}$.
Otherwise $s$ is bounded below, and $u_s\ge c_\varepsilon>0$ on
$y\ge s$: both $J_s$ and $x^s/a_s$ are bounded below there, so their
weighted contributions through $K$ and $1-K$ give this assertion.
Increasing $C_\varepsilon$ absorbs $(y/s)^s\le C_q$ on this compact
range. This proves \eqref{eq:coupled-low-lower}.

For the deficit, the ordinary sharp normalizer $a_s^{(0)}$ is at most
$a_s$, so
\[
 u_s\le x^sK/a_s^{(0)}+f_L(1-K).
\]
On $C_fs\le y\le y_0$, the excess of the first term over $K$ is bounded
by $K\{\exp(Cs+s\log(y/s))-1\}$.
Choose $C_f$ so that $s[C+\log(y/s)]\le\eta y$ with $\eta$ small
relative to $d_L$. The excess is then at most
$d_L(1-K)/2$, giving a deficit at least $c_\varepsilon y$.
On $y\ge y_0$, that excess tends uniformly to zero as $s\downarrow0$:
use the uniform gap on compact intervals and the inverse-power tail,
or the bound $Cs(1+|\log s|)\sqrt K$.
Since $1-K$ is bounded below there, sufficiently small $s$ gives a
fixed deficit. This proves \eqref{eq:coupled-low-deficit}.
The coefficient statement follows from the final part of
Lemma~\ref{lem:coupled-elementary-kernel} and
$a_s^{-1}\le C\ell^2$.

For the degree-one completion, $c_s$ stays in a fixed positive compact
interval as $s\downarrow0$, and
$a_s=s^s\exp(O_\varepsilon(s))$.
The same local estimates follow directly from $1-c_sx$.
Past its interior maximum the profile is decreasing to $q_L>q$,
so the floor argument is immediate; away from zero its limit is
$1-(1-q_L)x$, giving the asserted small-$s$ deficit by the same
$s\log(x/s)$ comparison. The finitely many resolution factors below
$\ell=2$ only alter $C_\varepsilon$.

For the high profile choose $N\asymp L$, where $L$ is the light-polynomial
budget, and let
\[
 R(x)=\left[U_{N-1}(\sqrt{1-x})/N\right]^6,\quad K=R^2,
 \quad J_s(x)=\frac{\int_0^x t^{s-1}K(t)\,dt}
 {\int_0^1 t^{s-1}K(t)\,dt}.
\]
With a fixed $f\in(0,1)$ sufficiently close to one, put
\begin{align}
 W(x)&=f+(1-f)(1-R(x))^2,\nonumber\\
 a_{s,\omega}&=\sup_x
 \frac{x^sK(x)}{1-\omega J_s(x)(1-K(x))W(x)},\nonumber\\
 v_{s,\omega}(x)&=x^sK(x)/a_{s,\omega}
              +\omega J_s(x)(1-K(x))W(x).\label{eq:coupled-high-definition}
\end{align}
The value of a ratio at a zero of $K$ is its continuous extension.
One has, uniformly for $0<\omega\le1$,
\begin{gather}
 a_{s,\omega}=N^{-2s}s^s\exp(O_\varepsilon(s)),\qquad
 0\le v_{s,\omega}\le1,\label{eq:coupled-high-normalization}\\
 e^{-Cs}(y/s)^s\le v_{s,\omega}(x)\le e^{Cs}(y/s)^s
 \quad(y=N^2x\le s),\label{eq:coupled-high-sharp}\\
 v_{s,\omega}(x)\ge c\omega\quad(y\ge s),\qquad
 v_{s,\omega}(x)\ge q\omega\quad(y\ge C_qs),\label{eq:coupled-high-floors}\\
 |v_{s,\omega}(x)-\omega J_s(x)|
 \le C\min\{y^s,y^{-3}\}.\label{eq:coupled-high-localization}
\end{gather}
The upper sharp bound extends to all $y$, using $v\le1$.
The coefficient sum of $v_{s,\omega}/x^s$ is at most $C^{L+1}$.
We choose the fixed degree constant large enough that $N^2\ge2C_q$.

Here is a normalization argument retaining the factor $s$ in the
exponent. For $\omega=1$ the denominator is at least
\[
 K+(1-f)(1-K)(2R-R^2),
\]
so its kernel quotient is bounded by
\[
 \frac{R}{R+(1-f)(1-R^2)(2-R)}.
\]
This has local exponential decay, a uniform gap away from zero, and a
$y^{-3}$ tail. The sharp maximum is therefore
$N^{-2s}s^s\exp(O_\varepsilon(s))$. For general $\omega$, the normalizer
lies between the ordinary sharp normalizer and this one.
Exact boundedness follows from the definition. The polynomial degree
is at most $18(N-1)$ after division by $x^s$.
The difference in \eqref{eq:coupled-high-localization} equals
\[
 x^sK/a_{s,\omega}-\omega J_sK
 -\omega(1-f)J_s(1-K)(2R-R^2),
\]
which proves its two tail bounds using $K\le y^{-6}$ and $R\le y^{-3}$.
The floor $q\omega$ follows, for small $s$ and $y\ge C_fs$, from
$x^s/a_{s,\omega}\ge1$ and
\[
 v_{s,\omega}\ge\omega[K+J_s(1-K)W]\ge\omega(1-(1-f)-Cs).
\]
For the other $s$, increase $C_q$ and use the integral and kernel tails.
In addition $J_s\ge q$ for $y\ge C_qs$, after enlarging $C_q$.
Indeed $J_s\ge y^s e^{-Cs}$ for $y\le1$, while
$1-J_s\le Cs y^{s-6}$ for $y\ge1$.

The following distinction is necessary. A global bound
$v_{s,\omega}\ge\min\{q\omega,e^{-Cs}(y/s)^s\}$ is asserted only
when $\omega$ is bounded below by a fixed positive constant. In that
case it follows by the same compact-interval argument as for the low
profile. For arbitrary $\omega\downarrow0$ we use only
\eqref{eq:coupled-high-sharp} and \eqref{eq:coupled-high-floors}.
A claim of the stronger global envelope uniformly down to $\omega=0$
would be false at interior kernel zeros.

\paragraph{Combining the profiles.}
The scalar envelope determines the weights $A,B$ and the product scale
$z$. Small exponential adjustments $A',B'$ leave a usable strict gap
below one. The factor $1-J_1$ turns off the reverse profile as $p$
increases; the remaining terms approach $\omega p^\alpha$.
The proposition gives two-sided bounds for every $p\in[a,b]$ and a
separate bound below $a$. These are deterministic statements about
functions, before any data are sampled.

\begin{proposition}[One coupled deterministic bound]
\label{prop:coupled-deterministic-bound}
There are fixed $M,\lambda,C_\varepsilon,c_\varepsilon>0$ with
the following property. Set $A'=Ae^{-M\gamma}$, $B'=Be^{-M\delta}$,
and choose $0<a<b\le1$ by the exact capacity equation
\begin{equation}
 a=\lambda\frac{b h z^{1/h}}{m^2L^2}.
 \label{eq:coupled-exact-capacity}
\end{equation}
Define one polynomial by
\begin{equation}
 Q(p)=p^\alpha\left[
 A'u_\gamma(a/p)(1-J_1(p/b))
 +B'v_{\delta,\omega}(p/b)
 +\omega(1-B')J_\delta(p/b)\right].
 \label{eq:coupled-single-polynomial}
\end{equation}
The expression extends polynomially to $p=0$: the reverse degree is at
most $m$, and $\alpha+\delta=m+1$. For $a\le p\le b$, write
$Y=N^2p/b$ and $Z=\ell^2a/p$. Then
\begin{align}
 \frac{Q(p)}{p^\alpha}
 &\ge\kappa_0e^{-C_\varepsilon h}
       +c_\varepsilon\kappa_0\min(Y,1),
 \label{eq:coupled-global-lower}\\
 \frac{Q(p)}{p^\alpha}
 &\le1-c_\varepsilon\min\{h+Y+\min(Z,1),1\}.
 \label{eq:coupled-global-upper}
\end{align}
Below $a$, its absolute contribution and the target are bounded by
$C_\varepsilon^{m+1}a^\alpha$ per coordinate. The constants do not
change with the fractional part or with $m$.
\end{proposition}

\begin{proof}
\noindent\textbf{Step 1: reduce the interaction to the scalar envelope.}
Write $u=u_\gamma(a/p)$, $v=v_{\delta,\omega}(p/b)$, $H=J_1(p/b)$,
$J=J_\delta(p/b)$, and $T=Q(p)/p^\alpha$.
The degree ratios $\ell/m$ and $N/L$ are bounded above and below by
fixed positive constants. Thus
$ZY=\lambda(\ell/m)^2(N/L)^2h z^{1/h}$.
A sufficiently small fixed $\lambda=\lambda_\varepsilon$ gives
\begin{equation}
 u^\delta v^\gamma\le z,\qquad
 u\le e^{C\gamma}z^{1/\delta}(\delta/Y)^\gamma.
 \label{eq:coupled-product}
\end{equation}
In the central interval $Z\le\gamma$, $Y\le\delta$, the reverse
inequality holds up to $e^{-Ch}$ in the product.
Maximizing a linear functional under $u^\delta v^\gamma\le z$ gives
$Au+Bv\le1$: the two endpoint values are exactly one, and the
functional restricted to the product curve is convex.

\medskip\noindent\textbf{Step 2: establish a strict upper gap.}
The proof separates fractional parts close to zero, close to one, and
away from both endpoints. This is only an analysis of the same polynomial;
its definition does not change. Fix a sufficiently small fractional
threshold $s_0=s_0(\varepsilon)$ and then choose $M$ sufficiently large.
If $\gamma\le s_0$ and $Y\ge\gamma$, \eqref{eq:coupled-product} gives
$u\le e^{-\varepsilon/4}$ after decreasing $s_0$, since
$E_0\ge\varepsilon/2$. Also $J-H\le C\gamma$. Hence
\[
 T\le\max(e^{-\varepsilon/4},e^{-\varepsilon/16})
      +C\gamma+e^{-M(1-s_0)},
\]
which has a fixed deficit. If $Y<\gamma$, choose a fixed
$r\in(\kappa_{\max},1)$, where $\kappa_{\max}=e^{-\varepsilon/2}$.
When $u\le r$, the bounds $v,J\le C\gamma$ give a fixed deficit.
When $u>r$, the product constraint gives $v\le e^{-c/\gamma}$;
on this range $J\le Cv$. Therefore
\[
 T\le e^{-M\gamma}u+Ce^{-c/\gamma}
 \le1-c[(1-u)+\gamma].
\]
The low-profile deficit then supplies $\min(Z,1)$; when
$Z<C\gamma$ it is already controlled by $\gamma$.

For small $\delta$, the high profile has
\begin{equation}
 1-v\ge c_\varepsilon\min(Y,1)
 \quad(Y\ge C_\varepsilon\delta).
 \label{eq:coupled-high-deficit}
\end{equation}
To verify this, bound $v$ by the ordinary sharp bump plus
$\omega(1-K)$. On $C\delta\le Y\le y_0$, its excess over $K$ is
at most an arbitrarily small multiple of $Y$, while
$(1-\omega)(1-K)\ge c_\varepsilon Y$.
Outside this neighborhood the sharp bump converges uniformly to $K$
as $\delta\downarrow0$ and $1-K$ is bounded below. This proves
\eqref{eq:coupled-high-deficit}.
On $Y\ge C_\varepsilon\delta$, the low term is at most
$Ce^{-\varepsilon/(2\delta)}$, while
\[
 B'v+\omega(1-B')J
 \le1-B'c_\varepsilon\min(Y,1)-(1-B')(1-\omega).
\]
The exponential term is absorbed. Since $u\ge c\min(Z,1)$ for
$\gamma\ge1/2$, the same margin controls $\min(Z,1)$.
On $Y<C_\varepsilon\delta$, if $u\le c_1\delta$ then
$1-B'\ge1-e^{-M\delta}\ge\delta/2$ supplies a linear deficit.
If $u>c_1\delta$, the product constraint gives
$v\le e^{-\varepsilon/4}$, and the low amplitude $A'\le e^{-M\gamma}$
gives a fixed deficit. Finally, when both fractional parts are bounded
below, $T\le e^{-Ms_0}+e^{-Ms_0}+\omega$ gives a fixed deficit.
This proves \eqref{eq:coupled-global-upper} for $\alpha\ge3/2$.

\medskip\noindent\textbf{Step 3: control the loss introduced by the cutoff.}
We next record the cancellation needed for the lower bound. Since
$1-B=Az^{1/\delta}\le e^{-E_0/\delta}$, concavity of the logarithm gives
\[
 \log\frac{1-B'}{1-B}
 \ge(1-e^{-M\delta})E_0/\delta\ge(1-e^{-M})E_0.
\]
Take $M\ge1$, so that $\omega(1-B')\ge1-B$. Consequently
\begin{equation}
 T\ge A'u+B'v-(A'u-1+B)_+H,
 \qquad (A'u-1+B)_+H\le C_\varepsilon\kappa_0 h.
 \label{eq:coupled-cutoff-loss}
\end{equation}
For the second inequality, use \eqref{eq:coupled-product} and
$1-B=Az^{1/\delta}$. If $Y\le\delta$, put $r=Y/\delta$ and use
$H\le CY$ together with
$\max_{0\le r\le1}(r^\delta-r)=\gamma\delta^{\delta/\gamma}\le\gamma$.
For $Y\ge\delta$ the loss is at most $C\gamma(1-B)$.
When $\delta\ge1/2$ this is $O(\kappa_0h)$; otherwise use
$(1-B)/\kappa_0\le\kappa_0^{\gamma/\delta}
 \le C_\varepsilon\delta$.

\medskip\noindent\textbf{Step 4: reserve a term that supplies the local lower gap.}
Reserve a fixed small fraction of the high arm by putting
$B''=B'e^{-\eta\delta}$, with fixed $\eta>0$.
On $Y\le C_q\delta$ and bounded $E_0$, the sharp and compact-interval
profile bounds give
\begin{equation}
 (B'-B'')v\ge c_\varepsilon\kappa_0Y.
 \label{eq:coupled-high-reserve}
\end{equation}
For $Y\le\delta$ this follows from $\delta v\ge cY$;
for $\delta\le Y\le C_q\delta$ use the fixed positive floor,
with $\omega$ bounded below. On $Y\ge C_q\delta$,
$J\ge q$ and $v\ge q\omega$, so
\begin{equation}
 T\ge q\omega\ge(1+c_\varepsilon)\kappa_0.
 \label{eq:coupled-outer-lower}
\end{equation}

\medskip\noindent\textbf{Step 5: prove the main lower bound uniformly in the order.}
It remains to justify the base term without assuming false global
sharp bounds as $\omega\downarrow0$.
Choose a fixed sufficiently large $E_*=E_*(\varepsilon)$.
If $E_0\ge E_*$, a saturated low profile ($Z\ge\gamma$) supplies
$A'u(1-H)\ge c_\varepsilon$, because then $Y\le\delta$ and
$1-H$ is bounded below. A saturated high profile ($Y\ge\delta$)
supplies $B'v\ge c_\varepsilon\omega$. Choose $E_*$ so both quantities
are at least $2\kappa_0$. These cases give the base and local gain.
In the remaining central interval both sharp lower bounds apply;
weighted AM--GM, after the reserve, gives
$A'u+B''v\ge\kappa_0e^{-Ch}$.
The reserve is at least $c_\varepsilon\kappa_0Y$ there.

If $E_0\le E_*$, then $m$ is bounded and $\omega$ is bounded below.
Set
\[
 U=e^{-C\gamma}(Z/\gamma)^\gamma,\qquad
 V=e^{-C\delta}(Y/\delta)^\delta.
\]
On $Y\le C_q\delta$, the valid global comparisons are
$u\ge\min(q,U)$ and $v\ge\min(q\omega,V)$.
If a saturated low term satisfies $A'q\ge\kappa_0$, the base sum is
already at least $\kappa_0$. Otherwise $A'q<\kappa_0$ forces
$\gamma\ge c_\varepsilon$: as $\gamma\downarrow0$, $A'\to1$
uniformly for $z\le\kappa_{\max}$, whereas $q>\kappa_{\max}$.
Since $p\ge a$, $Z\le\ell^2$ and $m$ is bounded. Thus $U\le C_\varepsilon$,
and replacement of $U$ by its floor costs only
$\delta\log C_\varepsilon=O_\varepsilon(h)$ in weighted AM--GM.
The same argument for $B''q\omega<\kappa_0$ forces
$\delta\ge c_\varepsilon$, because $B''\to1$ uniformly as
$\delta\downarrow0$ and $q\omega/\kappa_0>1+c_\varepsilon$.
On the bounded $Y$ interval, $V\le C_\varepsilon$, so its floor costs
$\gamma\log C_\varepsilon=O_\varepsilon(h)$.
If its saturated term already exceeds $\kappa_0$ no comparison is
needed. In every case
\[
 A'u+B''v\ge\kappa_0e^{-C_\varepsilon h}.
\]

For sufficiently small $h$, subtracting the loss in
\eqref{eq:coupled-cutoff-loss} merely changes the constant in the
exponential; retain the reserve separately. If $h$ is bounded below,
$1-H$ is bounded below on $Y\le C_q$ by our choice $N^2\ge2C_q$.
For an explicit uniform justification, integrate the scaled kernel over
$[C_q,2C_q]$. As $N\to\infty$ it converges uniformly there to
$(\sin\sqrt y/\sqrt y)^{12}$, whose integral is positive.
The finitely many remaining admissible $N$ also have positive integrals,
and the normalized total integral is uniformly bounded. This proves a
fixed positive tail bound.
Thus the base part of $T$ is at least a fixed positive multiple of
$A'u+B''v$, again absorbed by $e^{-C_\varepsilon h}$; the reserve is
unchanged. Together with \eqref{eq:coupled-outer-lower} this proves
\eqref{eq:coupled-global-lower}. The proof used the same polynomial
and the same parameters throughout.

\medskip\noindent\textbf{Step 6: bound the portion below the lower cutoff.}
Finally, below $p=a$ the reverse polynomial has coefficients at most
$C^{m+1}a^{\alpha-j}$ for $0\le j\le m$. Since $p/b\le a/b$,
$0\le1-H\le1$, and the high profiles are bounded, its absolute value
and $p^\alpha$ are at most $C^{m+1}a^\alpha$. This proves the final
assertion.
\end{proof}

\subsection{From pointwise approximation to a power-sum estimate}
\paragraph{Why a deterministic correction is possible.}
There can be as many as $k$ probabilities below $a$, so their errors
must be added. The upper bound on the approximation has deliberately
left a positive gap on the other coordinates. The following calculation
shows that this gap is large enough to pay for a known correction that
covers all the below-cutoff errors. This is why the estimator includes
a deterministic additive term instead of requiring an accurate
approximation at every arbitrarily small probability.

\paragraph{Summed affine baseline.}
Let $P_*=k^{1-\alpha}$, $a_0=ka\le c_\varepsilon$, and
$P=\sum_i p_i^\alpha$. For $\alpha\ge2$, at least half the mass is on
$p_i\ge1/(2k)\ge a$, and hence
\[
 \sum_{p_i\ge a}p_i^\alpha\min(\ell^2a/p_i,1)
 \ge2^{1-\alpha}\min(\ell^2ka,1/2)P_*.
\]
For $3/2\le\alpha<2$, put $s=kp$ and $\theta=N^2/(bk)$.
The weighted sum of $Z+Y$ has pointwise factor
\[
 \ell^2a_0s^{\beta-1}+\theta s^{\beta+1}
 \ge(\ell^2a_0)^{(\beta+1)/2}\theta^{(1-\beta)/2}
 \ge c_\varepsilon a_0^\beta,
\]
where the last step uses capacity, $\beta=\gamma\ge1/2$, and
$z\ge\kappa_0/2\ge e^{-2\varepsilon\beta}/2$.
If $Z+Y\ge1$, truncation still gives the same conclusion on $s\ge1/2$.
Consequently the upper margin pays
\[
 E_aP_*,\qquad E_a=d_\varepsilon C_0^m(ka)^\beta,
\]
with sufficiently small fixed $d_\varepsilon>0$ and sufficiently small
$ka$, and it dominates $C^{m+1}ka^\alpha$.
This justifies a deterministic correction of size $2E_aP_*$ while
retaining fixed portions of the upper margins. For a common formula
including $\alpha\downarrow1$, replace $E_a$ by
$\min(\beta,1)d_\varepsilon C_0^m(ka)^\beta$.

\paragraph{The order-one divided difference.}
Near $\alpha=1$, the reverse profile has a particularly transparent
interpretation: it is a tangent line to $p\mapsto p^{1+\beta}$.
The exact tangent identity below exposes a factor $\beta$ in every
nonlinear error. The known linear aggregate $\sum_i p_i=1$ can
therefore be removed before estimating the remainder.

For $\alpha=1+\beta\le3/2$, use the same degree-one profile, whose
exact reverse identity is
\[
 p^\alpha u_\beta(a/p)=(1+\beta)\tau^\beta p
 -\beta\tau^{1+\beta}=:T_\tau(p),\qquad
 \tau=c_\beta(1+\beta)a/\beta.
\]
Suppose the common capacity rule gives $\tau/t$ in a fixed positive
compact interval, where $t=b/N^2$, and gives
$A'=1+O_\varepsilon(\beta)$ and $B'=O_\varepsilon(\beta)$.
Then the very same polynomial satisfies
\begin{equation}
 \sup_{0\le p\le b}|Q(p)-\omega p^\alpha|
 \le C_\varepsilon\beta t^\alpha.
 \label{eq:coupled-near-one-bias}
\end{equation}
Subtract $\omega[T_\tau(1-H)+p^\alpha J_\delta]$.
Its tangent-cutoff error is $O(\beta t^\alpha)$.
The remaining terms are $(A'-\omega)T_\tau(1-H)$ and
$B'p^\alpha(v_{\delta,\omega}-\omega J_\delta)$;
localization bounds each by $O_\varepsilon(\beta t^\alpha)$.
This includes the interval $p<a$.
If $Q=\sum_jq_jp^j$, then
\[
 q_0=-A'\beta\tau^{1+\beta},\qquad
 q_1=A'[(1+\beta)\tau^\beta+
       \beta\tau^{1+\beta}J_1'(0)/b],\qquad
 |q_1-t^\beta|\le C_\varepsilon\beta t^\beta.
\]
For $j\ge2$, coefficientwise cancellation at $\beta=0$ gives
$|q_j|\le\beta C_\varepsilon^{L+1}b^{\alpha-j}$.
The known constant and linear aggregates are therefore removed before
the common factorial-variance calculation.

\paragraph{What the deterministic proposition supplies to the statistical theorem.}
For any light/heavy assignment with light coordinates at most $b$,
replace heavy coordinates by a mean centered at $\omega p_i^\alpha$.
The proposition gives the baseline and local margins after paying the
below-cutoff error. Rescaling both parts by $(1+c_u u)^{-1}$ supplies an
additional upper margin of order $uP$. If
$E-E_0\ge C_{\rm bias}h+c_{\rm reg}u$, its lower counterpart
follows directly from \eqref{eq:coupled-global-lower}.
When this inequality is not available, one must use the same
polynomial's localization toward $\omega p^\alpha$ and the exact root
calibration to bound the aggregate error. The present proposition does
not by itself establish that final calibration, the pilot errors, or
the sampling variance; these are separate hypotheses to be verified
for the single common estimator.

\paragraph{The deterministic output used in the sampling proof.}
The next statement collects all preceding estimates into the precise
form needed later. The sets $\mathcal L$ and $\mathcal H$ form an
arbitrary partition of the coordinates into light and heavy groups.
The only restriction is $p_i\le b$ for light coordinates. Thus the
statement can later be applied after conditioning on an independent
pilot sample, provided that pilot classification has succeeded.
The quantity $\mathcal M$ is an explicit error budget: it is a sum of
a below-cutoff baseline, a regularization gap, and a local gap.

\begin{corollary}[Bias transfer for the single corrected expression]
\label{cor:coupled-common-bias-transfer}
Let $\kappa=e^{-E}$, $\chi=\min(\beta,1)$,
\[
 \mathfrak u=\min\{h+(\beta/\alpha)X^{-\alpha},u_0\},\quad
 r=(1+c_u\mathfrak u)^{-1},\quad
 E_a=\chi d_\varepsilon C_0^m(ka)^\beta.
\]
For any assignment to light and heavy coordinates, with every light
coordinate at most $b$, define the oracle mean
\[
 \mu=r\left[\sum_{i\in\mathcal L}Q(p_i)
             +\omega\sum_{i\in\mathcal H}p_i^\alpha\right]
            +2E_aP_*.
\]
Use the same $Q$, radius rule, and coefficients for all orders. Put
\[
 B_p=\sum_i p_i^\alpha\min(N^2p_i/b,1),\qquad
 \mathcal M=E_aP_*+\kappa\mathfrak uP+\chi\kappa B_p.
\]
Suppose the root calibration has established the following quantitative
facts, with constants depending only on $\varepsilon$:
\begin{enumerate}
\item $ka\le c_\varepsilon/D$, and $E/2\le E_0\le E$.
\item For $\alpha\ge3/2$ with $E/h$ above a fixed sufficiently large
threshold,
$E-E_0\ge C_{\rm bias}h+c_{\rm reg}\mathfrak u$, where
$C_{\rm bias}$ is chosen after the deterministic proposition and
$c_{\rm reg}>0$ is a fixed, possibly small, accuracy-dependent constant.
\item If $\alpha\ge3/2$ and $E/h$ is bounded by that threshold, then
$\alpha$ is bounded, $\gamma,\delta$ are bounded below, and
$kt\le C_\varepsilon/D$ for $t=b/N^2$.
\item For $1<\alpha\le3/2$, one has $kt\le C_\varepsilon/D$,
$\tau/t$ in a fixed positive compact interval,
$A'=1+O_\varepsilon(\beta)$, $B'=O_\varepsilon(\beta)$, and
$\mathfrak u=\Theta_\varepsilon(\beta)$.
\end{enumerate}
Then, choosing $c_u,d_\varepsilon>0$ sufficiently small and subsequently
$D$ sufficiently large, one has
\[
 \boxed{\kappa P+c_\varepsilon\mathcal M
        \le\mu\le P-c_\varepsilon\mathcal M.}
\]
The enumerated statements are calibration hypotheses, not consequences
of polynomial approximation alone.
\end{corollary}
\begin{proof}
The upper baseline bounds and $ka$ small imply that $2E_aP_*$ pays the
below-cutoff error and consumes only a fixed part of the upper baseline.
For $\alpha\ge3/2$ under the second calibration inequality, apply
\eqref{eq:coupled-global-lower}--\eqref{eq:coupled-global-upper}.
Heavy coordinates have mean $\omega p_i^\alpha$ with a fixed relative
margin from both endpoints; their gaps pay any omitted $B_p$ or baseline
terms. Rescaling supplies an upper gap proportional to $\mathfrak uP$.
The calibration inequality makes
$r\kappa_0e^{-C_\varepsilon h}\ge\kappa(1+c\mathfrak u)$,
so the corresponding lower gap is retained. Since $\chi\ge1/2$, the
result is the stated margin.

To verify the other calibration range with the same polynomial, observe
on $a\le p\le b$ that
\[
 |Q(p)-\omega p^\alpha|
 \le Cp^\alpha\min\{1,(t/p)^3\}.
\]
This follows directly from
$T-\omega=A'u(1-H)+\omega(J-1)+B'(v-\omega J)$.
For $\nu=\min(\alpha,2)$, H\"older's inequality bounds its aggregate,
divided by $P$, by $C(kt)^\nu(P/P_*)^{-\nu/\alpha}$.
The below-$a$ contribution is at most $C^{m+1}(ka)^\alpha P_*$.
Under the third calibration statement these errors tend uniformly to
zero as $D$ increases. In that range $\alpha$ lies in a fixed compact
interval and $\beta\ge1/2$, so the center $r\omega$ remains a fixed
positive distance from both $\kappa$ and one after $c_u$ is chosen small.
Choose $d_\varepsilon$ small and then $D$ large. Since
$\mathcal M\le C_\varepsilon P$, these fixed gaps imply the claimed
inequality.

Finally, under the fourth calibration statement,
\eqref{eq:coupled-near-one-bias} bounds the aggregate error by
$C_\varepsilon\beta(kt)^\alpha P_*$.
The center $r\omega$ is separated from both $\kappa$ and one by fixed
positive multiples of $\beta$, after choosing $c_u$ small.
The correction is at most $C d_\varepsilon\beta P$ because $m=1$.
Taking $d_\varepsilon$ small and $D$ large leaves gaps of order $\beta P$.
Here $\mathcal M\le C_\varepsilon\beta P$; the factor $\chi$ in its
last term is essential. This gives the same asserted inequality without
changing the estimator or its parameter rule.
\end{proof}

\paragraph{What has now been proved, and what remains.}
Once the listed calibration conditions are checked, the corrected
oracle mean is at distance at least $c_\varepsilon\mathcal M$ from
both endpoints of $[\kappa P,P]$. The next section replaces that
oracle mean by observable count statistics. It must show that
polynomial-count fluctuations, heavy-count errors, and pilot mistakes
jointly use less than this available distance. No sample-complexity
conclusion follows from the deterministic bounds alone.

\clearpage
\section{From the polynomial approximation to a sampling guarantee}
\label{sec:tutorial-sampling}

The preceding section constructed a polynomial whose values are close
 enough to $p^\alpha$ on the range where small probabilities matter.
We now replace its monomials by factorial statistics and prove that the
resulting random quantity stays inside the required power-sum interval.
There are two different statements to establish:
\begin{enumerate}
\item the deterministic approximation has room on both sides of the
interval $[e^{-2\varepsilon(\alpha-1)}P_\alpha(p),P_\alpha(p)]$;
\item count fluctuations and mistakes made by the pilot consume less
than that room with sufficiently high probability.
\end{enumerate}
A small approximation error alone does not establish the second statement.
Large polynomial coefficients can amplify count noise. This is why we
first bound the coefficients of the actual polynomial, rather than
replace it by an unspecified polynomial of the same degree.

\paragraph{The objects used throughout this section.}
Write $\alpha=m+\gamma$, where $m=\lfloor\alpha\rfloor$ and
$0<\gamma<1$, and put $\beta=\alpha-1$ and $\delta=1-\gamma$.
The previous section has specified a sample scale $n$, a polynomial
parameter $L$, and probability cutoffs $0<a<b$. The actual polynomial
has degree $d\le C_dL$ and is written $Q(p)=\sum_{j=0}^d q_jp^j$.
For every symbol we generate two independent counts,
\[
 N_i\sim\Poi(np_i),\qquad Y_i\sim\Poi(c_pnp_i),
\]
where $N_i$ is used to estimate the contribution and $Y_i$ is used only
to select the branch. The indicator
$I_i=\mathbf1\{Y_i\le c_pB\alpha\}$ selects the polynomial branch.
The Gamma branch is used when $I_i=0$. All counts are independent
across symbols as well. Here $B,c_p$ are fixed constants selected in
the proof, and not unknown features of the distribution.

\paragraph{Why the proof is split into several parts.}
For $\alpha\ge3/2$, the approximation has a margin that controls the
conditional variance of the polynomial, pilot errors, and the Gamma
branch separately. For $\alpha$ close to one, the desired interval
shrinks at rate $\beta=\alpha-1$; the proof must therefore retain a
factor $\beta$ in bias and a factor $\beta^2$ in variance.
The estimator is unchanged, but its analysis uses the cancellation
$\sum_i p_i=1$ explicitly. After these estimates are established,
Theorem~\ref{cu:theorem} assembles them into one uniform upper bound.

\paragraph{A useful identity to keep in mind.}
If $\widehat Q_i=\sum_jq_j(N_i)_j/n^j$, then
$\E\widehat Q_i=Q(p_i)$. Moreover, the sum of the linear terms can be
replaced by the known number $q_1$:
\[
 q_1+\sum_i\left(\widehat Q_i-q_1N_i/n\right)
 =\sum_i\widehat Q_i+q_1\left(1-\sum_iN_i/n\right).
\]
The added term has mean zero. Its role is particularly important near
order one, where the target is almost the known linear functional
$\sum_i p_i$. We call this replacement the \emph{affine anchor}.

\subsection{Conditional count noise for the common polynomial}
\label{sec:conditional-count}
\label{sec:common-conditional-count}

This section evaluates the count variance for one polynomial and one pilot
rule. The deterministic approximation margins are hypotheses of
Theorem~\ref{thm:cc-count}; they are not consequences of a variance bound.
Throughout, logarithms are natural and $(w)_j=w(w-1)\cdots(w-j+1)$.

\paragraph{Step 1: control the coefficients before controlling variance.}
The kernels used in the approximation have many coefficients, but their
coefficients decay factorially after the relevant endpoint scale is
removed. This permits summation of the variances of all monomials.
The first lemma records that decay; the next transports it to the
coupled polynomial actually used by the estimator. As usual, $[z^j]P$
denotes the coefficient of $z^j$ in a polynomial $P$.

\begin{lemma}[Endpoint coefficients]\label{lem:cc-endpoint}
For $N\ge2$, put
\[
 B_N(z)=\frac{1-T_N(1-2z)}{2N^2z},\qquad K_N(z)=B_N(z)^2.
\]
The singularity at zero is removable, $K_N(0)=1$, and
\[
 |[z^\ell]K_N(z)|\le\frac{(16N^2)^\ell}{(2\ell)!}.
 \tag{CC1}\label{eq:cc-kernel-coeff}
\]
The same type of estimate, with a different absolute constant, holds for
any fixed power of $K_N$ and for a fixed finite number of products of
its normalized integral profiles.

In particular, the low endpoint profiles used in the coupled construction
have polynomial part $P$ of degree at most $m$ satisfying
\[
 |[z^\ell]P|\le C m^{2\gamma}
                \frac{(Cm^2)^\ell}{(\ell!)^2},
 \qquad 0\le\ell\le m,\quad 0<\gamma<1.
 \tag{CC2}\label{eq:cc-low-profile-coeff}
\]
For $Q_-(p)=p^{m+\gamma}F(a/p)$, this implies
\[
 |[p^{m-\ell}]Q_-|
 \le C m^{2\gamma}a^{\gamma+\ell}
                 \frac{(Cm^2)^\ell}{(\ell!)^2}.
 \tag{CC3}\label{eq:cc-reverse-coeff}
\]
Constants may depend on the fixed profile parameters, but not on $m$, the
fractional part, $a$, or the alphabet size.
\end{lemma}
\begin{proof}
The Chebyshev recurrence, or induction in that recurrence, gives
\[
 [z^\ell]B_N(z)=(-1)^\ell
 \frac{2^{2\ell+1}}{(2\ell+2)!}
 \prod_{r=1}^{\ell}(N^2-r^2),\qquad0\le\ell<N.
\]
Thus $|[z^\ell]B_N|\le(4N^2)^\ell/(2\ell+1)!$.
Convolution and the identity
\[
 \sum_{i=0}^{\ell}\frac1{(2i+1)!(2\ell-2i+1)!}
 =\frac{2^{2\ell+1}}{(2\ell+2)!}
 \le\frac{4^\ell}{(2\ell)!}
\]
prove \eqref{eq:cc-kernel-coeff}. Each additional fixed product changes
only the constant, by the even binomial convolution.

For a normalized integral profile
\[
 J_s(z)=I_s^{-1}\int_0^z u^{s-1}K(u)\,du,
 \qquad I_s=\int_0^1u^{s-1}K(u)\,du,
\]
its polynomial part has coefficients $k_\ell/[I_s(\ell+s)]$.
The endpoint normalization $(sI_s)^{-1}=N^{2s}e^{O(s)}$ therefore
preserves \eqref{eq:cc-kernel-coeff}, up to the prefactor $CN^{2s}$.
For the actual raised-floor low profile, let $f\in(0,1)$ be its fixed
floor parameter and define
\[
 A_{s,f}=\max_{0\le z\le1}
 \frac{z^sK(z)}{1-fJ_s(z)(1-K(z))},\qquad
 \frac{F_s(z)}{z^s}=\frac{K(z)}{A_{s,f}}
             +f\frac{J_s(z)}{z^s}(1-K(z)).
\]
Its proved normalization is
$A_{s,f}=N^{-2s}[s/(1-f)]^se^{O_f(s)}$.
Because $s/(1-f)\le1/(1-f)$ and $s^{-s}\le e^{1/e}$,
one has $A_{s,f}^{-1}\le C_fN^{2s}$ and
$N^{2s}\le C_fA_{s,f}^{-1}$.
The kernel term therefore has coefficient bound
$C_fN^{2s}(CN^2)^\ell/(2\ell)!$.
For the product with $1-K$, convolve the integral-polynomial estimate
with \eqref{eq:cc-kernel-coeff}; the even binomial convolution only
changes the base constant. The added product has zero constant
coefficient, so the leading coefficient remains exactly
$A_{s,f}^{-1}$. This proves
\eqref{eq:cc-low-profile-coeff} when $N$ is a fixed fraction of $m$.
For the finitely many smaller available degrees, use the actual
fixed-floor degree-one profile. For the prescribed $q_L\in(0,1)$,
let $c_s\in(s/(1+s),1)$ be the unique solution of
\[
 (1-c_s)c_s^s(1+s)^{1+s}/s^s=q_L.
\]
The left side decreases continuously from one to zero on this
interval. Its stationary point $s/[c_s(1+s)]$ lies in $(0,1)$,
so its exact normalization is
\[
 F_s(z)=z^s(1-c_sz)/A_s,\qquad
 A_s=\frac{s^s}{c_s^s(1+s)^{1+s}}.
\]
Both polynomial coefficients are bounded by $4e^{1/e}$, since
$c_s\le1$, $s^{-s}\le e^{1/e}$, and $(1+s)^{1+s}\le4$.
This is the same prescribed-floor degree-one profile as in the
approximation construction; no different profile is substituted.
Substitution of $a/p$ proves
\eqref{eq:cc-reverse-coeff}.
\end{proof}

\begin{lemma}[Coefficients of the coupled polynomial]
\label{lem:cc-coupled-coeff}
Write $\alpha=m+\gamma$, $\delta=1-\gamma$, and consider the actual
coupled polynomial
\[
 Q(p)=A'Q_-(p)\{1-J_1(p/b)\}
       +p^\alpha\{B'v_\delta(p/b)+\omega(1-B')J_\delta(p/b)\}.
 \tag{CC4}\label{eq:cc-coupled-Q}
\]
Assume that its amplitudes are bounded, that $a\le b$, $nb=c_bL$,
that all high profiles have degree $O(L)$ and endpoint scale $L^2$,
and that
\[
 Z=am^2L^2/b\le Z_0.
 \tag{CC5}\label{eq:cc-Z}
\]
Here $c_b,Z_0$ and the profile constants are fixed. Then, writing
$Q=\sum_{j=0}^dq_jp^j$, $d\le C_dL$, one has
\begin{align}
 |q_j|&\le C m^{2\gamma}a^{\alpha-j}
    \frac{(Cm^2)^{m-j}}{((m-j)!)^2},&&1\le j\le m,
                                                        \tag{CC6}\label{eq:cc-low-coeff}\\
 |n^{-\delta}q_{m+1}|&\le C(L/\delta)^\delta,
                                                        \tag{CC7}\label{eq:cc-leading-high}\\
 |q_j|&\le e^{CL}b^{\alpha-j},&&m+1\le j\le d.
                                                        \tag{CC8}\label{eq:cc-high-coeff}
\end{align}
These bounds hold for the omega-tail profile with
$\omega=e^{-c_0E_0}$ and attenuated amplitudes
$A'=Ae^{-M\gamma}$, $B'=Be^{-M\delta}$.
\end{lemma}
\begin{proof}
The coefficient bound for the normalized cutoff is
\[
 J_1(z)=\sum_{t\ge1}v_tz^t,
 \qquad |v_t|\le(CL^2)^t/(t!)^2.
\]
For a coefficient of degree $m-\ell$ in the cross term $Q_-J_1$, use
$(\ell+t)!\ge\ell!t!$ to bound its absolute value by
\[
 C m^{2\gamma}a^{\gamma+\ell}
 \frac{(Cm^2)^\ell}{(\ell!)^2}
 \sum_{t\ge1}\frac{(CZ)^t}{(t!)^4}.
\]
The last sum is bounded under \eqref{eq:cc-Z}; the pure high arm starts
at degree $m+1$. This proves \eqref{eq:cc-low-coeff}.

The cross coefficient at degree $m+1$, expressed in count units, is at
most
\[
 C\left(\frac{am^2n}{\gamma}\right)^\gamma
 \frac{L^2}{nb}
 \sum_{\ell\ge0}\frac{(CZ)^\ell}
                         {(\ell!)^2((\ell+1)!)^2}
 \le C L^\delta(Z/\gamma)^\gamma
 \le C(L/\delta)^\delta.
\]
We used $am^2n=c_bZ/L$ and $\gamma^{-\gamma}\le e^{1/e}$.
The pure high arm has the same leading bound by its sharp
normalization. This proves \eqref{eq:cc-leading-high}.
Rescaling $p=bz$ gives \eqref{eq:cc-high-coeff}: the reverse coefficients
contribute factors $(a/b)^{\gamma+\ell}\le1$, and all polynomial
coefficient sums are at most $C^L$.

For the omega-tail normalization, the denominator in the defining
maximum lies between the ordinary bump denominator and the localized
unit-tail denominator. Its normalizer therefore lies between the
corresponding two normalizers, both equal to
$N^{-2s}s^se^{O(s)}$. The added term
$\omega J_s(1-K)W$ is a fixed product of the same kernel and integral
polynomials. It preserves all the bounds just proved. Bounded scalar
amplitudes and their attenuation also preserve these bounds.
\end{proof}

\paragraph{Step 2: use the pilot to suppress excessive polynomial noise.}
The polynomial is designed for small probabilities. Without a pilot,
its factorial moments could be large at a frequent symbol. Independence
of $Y$ and $N$ lets us multiply the factorial second moment by the small
probability that a frequent symbol is incorrectly assigned to this
branch. In the following lemma $j$ is an integer, as required by the
falling factorial.

\begin{lemma}[A pilot bound for factorial moments]\label{lem:cc-pilot}
Fix $B,c_p>0$. Let $N\sim\operatorname{Poi}(\lambda)$ and
$Y\sim\operatorname{Poi}(c_p\lambda)$ be independent, and put
$I=\mathbf1\{Y\le c_pB\alpha\}$. For $1\le j\le\alpha$,
\[
 \mathbb E[I(N)_j^2]\le(C_{B,c_p}\alpha)^j\lambda^j.
 \tag{CC9}\label{eq:cc-pilot-moment}
\]
\end{lemma}
\begin{proof}
The factorial identity gives
\[
 \mathbb E(N)_j^2
 =\sum_{r=0}^j\binom jr^2r!\lambda^{2j-r}
 \le\lambda^j(\lambda+j)^j.
\]
For $\lambda\le2B\alpha$ this proves the assertion. For larger means,
$\mathbb P(I=1)\le e^{-c_p\lambda/8}$, and
$\sup_{\lambda\ge0}e^{-c_p\lambda/8}(\lambda+j)^j
\le(C_{c_p}j)^j$. Independence finishes the proof.
\end{proof}

\paragraph{Step 3: measure noise against the available margin.}
Let $P=\sum_i p_i^\alpha$ and $P_*=k^{1-\alpha}$. The inequality
$P\ge P_*$ follows from convexity, so $A=P/P_*\ge1$ measures how far
the power sum is above its smallest possible value.
The margin below has three terms. The first is the deterministic
correction for probabilities below $a$. The second is the room
reserved by the scalar-root calibration. The third rewards coordinates
whose expected counts are large enough to provide additional room.
We need a variance bound proportional to the square of this entire
margin, uniformly over the unknown $p$.

Conditioning on the pilot makes all selectors $I_i$ fixed; it does not
remove count noise in $N_i$. The theorem bounds that conditional noise.
Pilot randomness itself is handled later, so this theorem should not
be read as a bound on the full variance of the hybrid estimator.

\begin{theorem}[Conditional count bound with the actual cutoff]
\label{thm:cc-count}
Fix $\varepsilon>0$ and $\alpha\ge3/2$ noninteger. Set
\[
 P_*=k^{1-\alpha},\quad P=\sum_i p_i^\alpha,\quad
 \kappa=e^{-2\varepsilon(\alpha-1)},\quad
 n=D\alpha k^{1-1/\alpha}X,
\]
\[
 L=\lceil\alpha(C_L+c_L\log X)\rceil,\qquad nb=c_bL,
 \qquad u\ge c_uX^{-\alpha},\quad X\ge1.
\]
Assume the coefficients satisfy
\eqref{eq:cc-low-coeff}--\eqref{eq:cc-high-coeff}, $nb\ge2B\alpha$,
and that fixed positive $\rho_-,\rho_+$ obey $ka\le\rho_+/D$.
Let
\[
 E_b=\chi_\alpha d_\varepsilon C_0^m(\rho_-/D)^{\alpha-1},\qquad
 E_a=\chi_\alpha d_\varepsilon C_0^m(ka)^{\alpha-1},\quad
 \chi_\alpha=\min\{\alpha-1,1\},
\]
where $1/2\le\chi_\alpha\le1$, and assume the calibrated baseline comparison
\[
 E_b\le E_a+C_\varepsilon\kappa u.
 \tag{CC10}\label{eq:cc-baseline-comparison}
\]
Define the actual margin
\[
 \mathcal M=E_aP_*+\kappa uP+\kappa\mathcal B,
 \qquad
 \mathcal B=\sum_i p_i^\alpha\min(Lnp_i,1).
 \tag{CC11}\label{eq:cc-margin}
\]
For independent counts $N_i\sim\operatorname{Poi}(np_i)$ and
$Y_i\sim\operatorname{Poi}(c_pnp_i)$, let
\[
 I_i=\mathbf1\{Y_i\le c_pB\alpha\},\qquad
 \widehat Q=\sum_iI_i\sum_{j=0}^dq_j(N_i)_j/n^j.
\]
For any fixed $\eta>0$, the constants can be chosen so that
\[
 \mathbb E\operatorname{Var}(\widehat Q\mid Y)
       \le\eta\mathcal M^2.
 \tag{CC12}\label{eq:cc-main-variance}
\]
A valid order of choices is: fix the profiles, pilot constants, $c_b$
and $C_L$; choose $c_L$ sufficiently small; determine the calibration
constants $\rho_-,\rho_+$; choose $C_0$ sufficiently large; choose the
fixed positive baseline prefactor $d_\varepsilon$; and choose $D$ last.
In particular, the proof permits the actual cutoff to be proportional
to $D^{-1}$ or arbitrarily smaller.
\end{theorem}
\begin{proof}
\noindent\textbf{3a. Introduce a comparison margin.}
The actual cutoff $a$ can be arbitrarily small near an integer, which
makes direct normalization inconvenient. The baseline comparison
replaces it by a controlled reference quantity without changing the
margin by more than a constant factor.
Write $A=P/P_*\ge1$ and
$\mathcal M_b=P_*(E_b+\kappa uA)+\kappa\mathcal B$.
Equation \eqref{eq:cc-baseline-comparison} gives
$\mathcal M_b\le(1+C_\varepsilon)\mathcal M$.

\medskip\noindent\textbf{3b. Bound degrees below the real order.}
These degrees satisfy $j<\alpha$. Their sums $P_j=\sum_i p_i^j$
can be bounded by $P$ through the power-mean inequality, and their
coefficients contain positive powers of the small cutoff $a$.
For $j\le m$, put $\ell=m-j$, $\xi=\alpha-j$, and
$\theta=j/\alpha\in(0,1)$. Power means give
$P_j\le k^{1-j}A^\theta$. Lemma~\ref{lem:cc-pilot} and
\eqref{eq:cc-low-coeff} bound the degree-$j$ conditional variance,
divided by $P_*^2$, by
\[
 C m^{4\gamma}\frac{(Cm^2)^{2\ell}}{(\ell!)^4}
 (\rho_+/D)^{2\xi}
 \left(\frac C{DX}\right)^j
 k^{j/\alpha-1}A^\theta.
 \tag{CC13}\label{eq:cc-degree-var}
\]
For $0<\theta<1$,
\[
 \sup_{A>0}\frac{A^\theta}{(E_b+\kappa uA)^2}
 \le E_b^{\theta-2}(\kappa u)^{-\theta}.
\]
Also $X^{-j}u^{-\theta}\le c_u^{-1}$ and
$k^{j/\alpha-1}\le1$. Retaining the sample multiplier exactly yields
\[
 D^{-2\xi-j+(\alpha-1)(2-\theta)}
       =D^{-2+j/\alpha}\le D^{-1}.
 \tag{CC14}\label{eq:cc-D-power}
\]
All powers of the fixed $\rho_-,\rho_+$ and of $\kappa^{-1}$ are
$e^{O_\varepsilon(\alpha)}$; the baseline contributes
$\chi_\alpha^{\theta-2}d_\varepsilon^{\theta-2}\le4d_\varepsilon^{-2}$ and
$C_0^{m(\theta-2)}\le C_0^{-m}$.
Finally,
\[
 \frac{m^{2\ell}}{(\ell!)^2}
 \le e^{2\ell}\binom m\ell^2.
\]
Indeed $(m)_\ell\ge(m/e)^\ell$, by the geometric mean of the largest
$\ell$ numbers in $\{1,\ldots,m\}$. Taking square roots in
\eqref{eq:cc-degree-var}, summing by the $L^2$ triangle inequality,
and using the binomial theorem costs only $e^{C\alpha}$.
Since $m\ge\alpha/2$, choosing $C_0$ sufficiently large gives
\[
 \mathbb E\operatorname{Var}(\widehat Q_{\le m}\mid Y)
       /\mathcal M_b^2
 \le C_{d_\varepsilon}D^{-1}e^{-c\alpha}.
 \tag{CC15}\label{eq:cc-low-var-final}
\]
The degree zero is conditionally deterministic.

\medskip\noindent\textbf{3c. Bound degrees above the real order.}
For $j>\alpha$, the previous power-mean comparison is no longer the
useful direction. We instead split expected counts into very small
and larger means. The leading degree receives a sharp bound; the
remaining degrees are controlled using the pilot and a geometric
series. This split explains why the coefficient of degree $m+1$ was
bounded separately in the preceding lemma.
For the high block put $r=m+1$, $H=nb$, and
$\tau=e^{-C_\tau L}$, with $C_\tau$ fixed sufficiently large. Define
$A_>=P_*^{-1}\sum_{np_i>\tau}p_i^\alpha$.
We claim
\[
 \frac{\mathbb E\operatorname{Var}(\widehat Q_{\ge r}\mid Y)}{P_*^2}
 \le\frac{e^{C\alpha}}{(DX)^\alpha}A
  +e^{CL}\left(\frac{c_bL}{D\alpha X}\right)^\alpha A_>.
 \tag{CC16}\label{eq:cc-high-var}
\]
For $\lambda=np\le\tau$, use
$\mathbb E(N)_j^2\le j!\lambda^je^{2\sqrt{j\lambda}}$.
The leading coefficient \eqref{eq:cc-leading-high} is harmless because
$(L/\delta)^{2\delta}\tau^\delta\le C$.
For $j\ge r+1$, use $j!\le(r+1)!d^{j-r-1}$ and
\eqref{eq:cc-high-coeff}. The $L^2$ sum is a geometric series with
ratio at most $\sqrt{d\lambda}/H$; increasing $C_L$ makes it bounded.
Its square is at most $e^{CL}r!\lambda^{r+1}\le r!\lambda^\alpha$
by the choice of $C_\tau$. Thus the small-mean part is at most
$e^{C\alpha}r!\lambda^\alpha$ in count units. Summing and using
$r!\le e^{C\alpha}\alpha^\alpha$ proves the first term of
\eqref{eq:cc-high-var}.

For $\tau<\lambda\le H$ and $j\ge r>\alpha$,
\[
 H^{2\alpha-2j}\mathbb E(N)_j^2
 \le H^\alpha\lambda^\alpha(1+d/H)^j.
\]
All coefficient and summation costs are $e^{CL}$.
For $\lambda>H\ge2B\alpha$, the pilot probability is at most
$e^{-c_p\lambda/8}$ and
$\mathbb E(N)_j^2\le\lambda^{2j}e^{j^2/\lambda}$.
After division by $H^\alpha\lambda^\alpha$, the extra factor is
\[
 e^{-c_p\lambda/8+j^2/\lambda}(\lambda/H)^{2j-\alpha}
 \le e^{CL},
\]
because $j\le C_dL$ and the supremum of
$C\log z-c'z$ on $z\ge1$ is finite. This proves
\eqref{eq:cc-high-var}.

\medskip\noindent\textbf{3d. Compare both bounds with the margin.}
A variance estimate is useful only after its ratio to $\mathcal M^2$
is small. The following calculation makes the dependence on the final
sample multiplier $D$ explicit. The constants in the construction are
chosen before $D$; a negative power of $D$ therefore genuinely permits
us to make the noise small.
To normalize its first term, use
$\sup_{A>0}A/(E_b+\kappa uA)^2\le(4E_b\kappa u)^{-1}$.
Since $X^\alpha u\ge c_u$, its bound is a fixed multiple of
\[
 d_\varepsilon^{-1}D^{-1}C_0^{-m}\rho_-^{-(\alpha-1)}e^{C\alpha}.
\]
For the second term, $\mathcal B/P_*\ge\tau A_>$; hence its bound is
\[
 d_\varepsilon^{-1}D^{-1}C_0^{-m}\rho_-^{-(\alpha-1)}
 e^{C\alpha+C'L}\left(\frac{c_bL}{\alpha X}\right)^\alpha.
\]
Writing $s=\log X$, one has
\[
 \log(L/\alpha)\le\log(C_L+1)+\frac{c_L}{C_L+1}s.
\]
Choose $c_L[C'+1/(C_L+1)]\le1/2$. The last bound becomes
$d_\varepsilon^{-1}D^{-1}C_0^{-m}\rho_-^{-(\alpha-1)}
e^{C''\alpha}X^{-\alpha/2}$.
The stated choice of $C_0$ absorbs every fixed exponential in
$\alpha$, and $D$ then makes the remaining fixed prefactors as small
as required. Combining the two blocks and the comparison of
$\mathcal M_b$ with $\mathcal M$ proves
\eqref{eq:cc-main-variance}.
\end{proof}

\paragraph{Step 4: justify replacing the random linear aggregate.}
The total main-sample count has mean $n$ and variance $n$. Thus the
anchor introduced at the start of this section is centered and its
variance can be computed exactly. This calculation prevents the
anchor from becoming an unaccounted source of noise.

\begin{lemma}[The common affine anchor]\label{lem:cc-anchor}
Under the low coefficient bound,
\[
 |q_1|\le C^m(ka)^{\alpha-1}P_*.
\]
Consequently, replacing the random global affine aggregate by its known
value changes the unanchored estimator by
$q_1(1-\sum_iN_i/n)$, whose expectation is zero and whose variance is
$q_1^2/n$. Relative to the actual margin \eqref{eq:cc-margin}, this is
at most $C_{d_\varepsilon}\mathcal M^2/n$ after fixing $C_0$ sufficiently
large. This estimate is independent of how small $a$ becomes.
\end{lemma}
\begin{proof}
Set $j=1$ in \eqref{eq:cc-low-coeff};
$m^{2\gamma}(Cm^2)^{m-1}/((m-1)!)^2\le C'^m$.
Since $a^{\alpha-1}=(ka)^{\alpha-1}P_*$, the first assertion follows.
The total count has law $\operatorname{Poi}(n)$, proving the variance
identity. Finally $\mathcal M\ge\tfrac12d_\varepsilon C_0^m
(ka)^{\alpha-1}P_*$.
\end{proof}

\paragraph{What changes when $\alpha$ approaches one?}
The logarithmic target divides by $\beta=\alpha-1$. A variance estimate
that is merely $O(P^2)$ would therefore be insufficient. Subtracting
the nearly linear part exposes coefficients of order $\beta$ in all
remaining terms. The next lemma states the resulting variance bound;
Section~\ref{nuc:section} verifies its coefficient hypotheses for the
actual polynomial and controls the entire hybrid estimator.

\begin{lemma}[The same polynomial near order one]
\label{lem:cc-near-one}
Suppose $\alpha=1+\beta$, $0<\beta\le1/2$, and set $t=b/N_0^2$,
where $N_0$ is comparable to $L$. Assume the actual coupled polynomial
has the coefficient cancellation
\[
 kt\le C_\varepsilon,\qquad
 |q_1-t^\beta|\le C_\varepsilon\beta t^\beta,\qquad
 |q_j|\le\beta e^{CL}b^{\alpha-j}\quad(j\ge2).
 \tag{CC17}\label{eq:cc-near-one-coeff}
\]
With the same pilot and $L,n,b$ as above, the conditional variance of
the light residual after subtracting $t^\beta N_i/n$ is at most
$\eta\beta^2P^2$, for fixed sufficiently small $c_L$, sufficiently large
$D$, and sufficiently large $k$. The same conclusion holds when the
single anchor used at every order is $q_1$.
\end{lemma}
\begin{proof}
The linear residual contributes at most
$C_\varepsilon\beta^2t^{2\beta}/n
\le C_\varepsilon\beta^2P^2/n$.
Every remaining degree satisfies $j\ge2>\alpha$, so the medium- and
large-mean argument in \eqref{eq:cc-high-var} applies at every mean,
without a small-mean split. It bounds the variance, divided by
$\beta^2P^2$, by
\[
 e^{CL}\left(\frac{c_bL}{D\alpha X}\right)^\alpha.
\]
The same logarithmic estimate used in the proof of
Theorem~\ref{thm:cc-count} makes this uniformly small.
The degree zero is conditionally deterministic.
Changing the global anchor from $t^\beta$ to $q_1$ adds the centered
quantity $(q_1-t^\beta)(1-\sum_iN_i/n)$, whose variance is at most
$C_\varepsilon\beta^2P^2/n$ by
\eqref{eq:cc-near-one-coeff}.
\end{proof}

\paragraph{Scope of the bound.}
The coefficient assumptions above are preserved by the omega-tail
modification and the common global anchor. The calibrated identity
$am^2L^2/b=\lambda h z^{1/h}\le\lambda/4$ verifies
\eqref{eq:cc-Z}, and the actual-cutoff comparison verifies
\eqref{eq:cc-baseline-comparison}. The section therefore supplies the
conditional count-variance component of the common proof. The
remaining pilot and heavy-count estimates are proved next; the mean
margin, including both saturation intervals, is combined with these
estimates in Theorem~\ref{cu:theorem}.

\subsection{Pilot errors and the analytic heavy count}
\label{sec:pilot-count}
\label{sec:cc-pilot-heavy}

There are two possible pilot errors. A \emph{false-heavy} assignment
uses the Gamma statistic at a symbol with a small true mean. A
\emph{false-light} assignment uses the polynomial outside its intended
interval. Neither error can simply be ruled out symbol by symbol:
a union bound would introduce an unnecessary dependence on $k$.
Instead we bound their summed contribution in expectation, then use
one probability bound for that aggregate.

\paragraph{False-heavy assignments.}
The analytic statistic can be signed at low counts. Consequently we
bound its absolute first moment and its second moment, rather than
assume positivity. The pilot upper tail provides enough powers of
$\lambda$ to offset even the nonzero value at count zero.

\begin{lemma}[False-heavy moments for the analytic count]
\label{lem:cc-false-heavy}
For real $\alpha\ge3/2$, define
\[
 H_\alpha(j)=\frac{\Gamma(j+1)}{\Gamma(j+1-\alpha)},
 \qquad j=0,1,\ldots,
\]
using the entire reciprocal-Gamma value at integer orders.
Given any fixed $c_*>0$, fixed constants $B,c_p$ can be chosen so
that, for independent $N\sim\operatorname{Poi}(\lambda)$ and
$Y\sim\operatorname{Poi}(c_p\lambda)$, with
$J=\mathbf1\{Y>c_pB\alpha\}$, and every
$0\le\lambda\le B\alpha/2$,
\begin{align}
 \mathbb E[J|H_\alpha(N)|]
   &\le e^{-c_*\alpha}\lambda^\alpha\min(\lambda,1),
                                                       \tag{CC18}\label{eq:cc-false-first}\\
 \mathbb E[JH_\alpha(N)^2]
   &\le e^{-c_*\alpha}\alpha^\alpha\lambda^\alpha
                         \min(\lambda,1).
                                                       \tag{CC19}\label{eq:cc-false-second}
\end{align}
\end{lemma}
\begin{proof}
Write $\alpha=m+\gamma$. For $j\le m$, reciprocal-Gamma recurrence
and $\Gamma(1-\gamma)\ge1$ give
\[
 |H_\alpha(j)|\le j!(m-j)!/\Gamma(1-\gamma)
                 \le m!\le\alpha^\alpha.
\]
At integer orders the same bound follows from the factorial statistic.
For $j\ge m+1$, $0\le H_\alpha(j)=G_\alpha(j)\le j^\alpha$.
For any $q\ge1$, the bound
$x^q\le(q/(et))^qe^{tx}$, with $t=q/(\lambda+q)\le1$,
and $e^t-1\le t+t^2$, gives
$\mathbb EN^q\le e^{Cq}(\lambda+q)^q$.
Therefore
\[
 \mathbb E|H_\alpha(N)|\le e^{C\alpha}(\lambda+\alpha)^\alpha,
 \qquad
 \mathbb EH_\alpha(N)^2\le e^{C\alpha}(\lambda+\alpha)^{2\alpha}.
 \tag{CC20}\label{eq:cc-H-moments}
\]
For $\alpha\le\lambda\le B\alpha/2$, the pilot upper tail is at most
$\exp\{-c_pB\alpha(\log2-1/2)\}$.
After division by the power scales in
\eqref{eq:cc-false-first}--\eqref{eq:cc-false-second}, the other
factors in \eqref{eq:cc-H-moments} are bounded by
$e^{C\alpha}(CB)^\alpha$. Increasing $B,c_p$ proves the assertions
on this range.
For $\lambda<\alpha$, set $r=\lambda/\alpha<1$. The pilot probability
is at most $(er/B)^{c_pB\alpha}$, and both normalized moments are
bounded by
\[
 e^{C\alpha}(e/B)^{c_pB\alpha}r^{(c_pB-1)\alpha}.
\]
Choose $c_pB\ge3$. The last power is at most $\min(\lambda,1)$,
and the preceding factor is at most $e^{-c_*\alpha}$ after the
fixed constants are increased. At $\lambda=0$, the pilot never
selects the heavy branch.
\end{proof}

\begin{corollary}[Aggregated false-heavy error]\label{cor:cc-false-heavy}
For independent coordinates, restrict the wrongly selected heavy
statistic to $np_i\le B\alpha/2$. Its absolute expected value is at most
$e^{-c_*\alpha}\mathcal B$, and its unconditional variance, including
pilot randomness, is at most
\[
 e^{-c_*\alpha}(\alpha/n)^\alpha P.
\]
Under the calibration in Theorem~\ref{thm:cc-count}, this variance is
an arbitrarily small fixed multiple of $\mathcal M^2$ after the stated
choices of constants.
\end{corollary}
\begin{proof}
Sum \eqref{eq:cc-false-first} after multiplying by $n^{-\alpha}$;
$\min(np_i,1)\le\min(Lnp_i,1)$. Independence and second moments give
the variance bound from \eqref{eq:cc-false-second}. Its normalization
is exactly the first high-degree term in
\eqref{eq:cc-high-var}, and uses the same baseline comparison and
exact $D^{-1}$ calculation.
\end{proof}

\paragraph{False-light assignments.}
A degree-$d$ polynomial may grow rapidly for $p>b$. However, the
probability that the pilot calls such a symbol light decays
exponentially in $np$. Exponential decay dominates polynomial growth
once the fixed pilot multiplier is large enough.

\begin{lemma}[False-light mean beyond the interval]
\label{lem:cc-false-light}
Suppose the complete polynomial has degree $d\le C_dL$ and rescaled
coefficient sum bounded by $e^{CL}b^\alpha$. Let $nb=c_bL\ge2B\alpha$.
For any prescribed fixed $c_*>0$, choosing the fixed pilot multiplier
$c_p$ sufficiently large gives, for all $p\ge b$,
\[
 \mathbb P\{\operatorname{Poi}(c_pnp)\le c_pB\alpha\}|Q(p)|
      \le e^{-c_*L}p^\alpha.
 \tag{CC21}\label{eq:cc-false-light}
\]
\end{lemma}
\begin{proof}
Put $\lambda=np$ and $z=\lambda/(nb)\ge1$. The polynomial bound and
pilot lower tail give a ratio to $p^\alpha$ at most
\[
 e^{CL-c_p\lambda/8}
           \{z^{-\alpha}+z^{d-\alpha}\}.
\]
Its logarithm is bounded by
$C'L+C_dL\log z-c_pc_bLz/8$.
The supremum over $z\ge1$ tends to minus infinity, after division by
$L$, as $c_p$ increases. This proves the claim.
\end{proof}

\paragraph{The genuine heavy contribution.}
For a frequent symbol, the Gamma statistic behaves approximately like
the $\alpha$th power of its count. It is more efficient to control the
$\ell_\alpha$ norm of the entire count vector than to require an
accurate estimate of every coordinate. The next proof uses the
triangle inequality for that norm. This is also why it has no
alphabet-sized union bound.

\begin{lemma}[One high-mean concentration bound]\label{lem:cc-heavy-norm}
For every fixed $0<t<1/8$ and $0<\rho<1$, there is $B=B(t,\rho)$
such that, for every real $\alpha\ge3/2$ and independent Poisson
counts of means $\lambda_i\ge B\alpha$,
\[
 \mathbb P\left\{
 (1-2t)^\alpha\sum_i\lambda_i^\alpha
 \le\sum_iH_\alpha(N_i)
 \le(1+2t)^\alpha\sum_i\lambda_i^\alpha
 \right\}\ge1-\rho.
 \tag{CC22}\label{eq:cc-heavy-norm}
\]
The bound also holds conditionally for any subset selected by independent
pilot counts, provided all means in that subset satisfy the displayed
lower bound.
\end{lemma}
\begin{proof}
Poisson Chernoff bounds imply
\[
 \mathbb P(|N-\lambda|\ge z)
 \le2e^{-z^2/[2(\lambda+z)]}
 \le2e^{-z^2/(4\lambda)}+2e^{-z/4}.
\]
Integration of $q\int_0^\infty z^{q-1}\mathbb P(|N-\lambda|\ge z)\,dz$
gives
\[
 (\mathbb E|N-\lambda|^q)^{1/q}
       \le C(\sqrt{q\lambda}+q),\qquad q\ge2.
\]
For $3/2\le q<2$, Jensen's inequality and the Poisson variance give
the same estimate. Let $S=\sum_i\lambda_i^\alpha$ and
$R=S^{1/\alpha}$. Since $\lambda_i\ge B\alpha$,
\[
 \mathbb E\sum_i|N_i-\lambda_i|^\alpha
            \le(C/\sqrt B)^\alpha S.
\]
Markov's inequality makes the $\ell_\alpha$ norm of the count error
at most $tR/2$, except with probability at most $\rho/3$, uniformly
for $\alpha\ge3/2$, once $B$ is large.

Write $g_\alpha(j)=G_\alpha(j)^{1/\alpha}$. The pointwise inequalities
$(j-\alpha)_+\le g_\alpha(j)\le j$ imply
\[
 \|g_\alpha(N)-\lambda\|_\alpha
 \le\|N-\lambda\|_\alpha+\alpha s^{1/\alpha}
 \le(t/2+1/B)R\le tR,
\]
where $s$ is the number of coordinates. Thus the truncated Gamma sum
lies in $[(1-t)^\alpha S,(1+t)^\alpha S]$ on this event.
The analytic and truncated statistics differ only when $N_i\le m$;
their absolute difference is at most $\alpha^\alpha$ there. For
$B\ge4$, Poisson lower tails give
\[
 \mathbb E\sum_i|H_\alpha(N_i)-G_\alpha(N_i)|
            \le B^{-\alpha}e^{-cB\alpha}S.
\]
The gap between $(1-t)^\alpha S$ and $(1-2t)^\alpha S$ is at least
$t(1-2t)^\alpha S$, and the upper gap is larger. Increase $B$ so that
Markov's inequality puts this signed correction below the gap with
probability at least $1-\rho/3$. This proves
\eqref{eq:cc-heavy-norm}. Conditioning on an independent pilot
only fixes the subset, so the same proof applies.
\end{proof}

The multiplicative factor $\omega$ deliberately places the heavy
estimate inside the permissible interval. This bias is useful:
it reserves room on both sides for random fluctuations. The
following corollary checks that the room is uniform for
$\alpha\ge3/2$.

\begin{corollary}[The actual omega heavy arm]\label{cor:cc-omega-heavy}
Let $E=2\varepsilon(\alpha-1)$, $E/2\le E_0\le E$,
$\omega=e^{-c_0E_0}$ for a fixed $0<c_0<1$, and $\kappa=e^{-E}$.
The constants in Lemma~\ref{lem:cc-heavy-norm} can be chosen so that
its omega-scaled statistic lies, with probability at least $1-\rho$,
in
\[
 [(\kappa+c_\varepsilon\kappa)S,
                        (1-c_\varepsilon\kappa)S].
\]
\end{corollary}
\begin{proof}
Put $d_0=\min(c_0/2,1-c_0)>0$. The logarithm $-c_0E_0$ has
distance at least $d_0E$ from both endpoints $-E$ and $0$.
Since $\alpha/(\alpha-1)\le3$, choose $t=t(\varepsilon,c_0)>0$
so that the logarithmic error in \eqref{eq:cc-heavy-norm} is at most
$d_0E/2$. The resulting bounds are
$\kappa e^{d_0E/2}$ and $e^{-d_0E/2}$. Because
$\alpha-1\ge1/2$, both have the asserted margins after choosing a
fixed positive $c_\varepsilon$.
\end{proof}

\paragraph{Connection to the deterministic bound.}
These lemmas control signed low-count leakage, count variance, and
true-heavy concentration for the specified statistic. The complementary
polynomial and aggregate bounds are established in
Section~\ref{sec:coupled-bias-bound}. Theorem~\ref{cu:theorem}
combines those deterministic bounds with the stochastic estimates here.

\paragraph{The remaining deterministic boundary: $p<a$.}
The approximation guarantee on $[a,b]$ says nothing about symbols
below $a$. There may be $k$ such symbols. We therefore bound their
total contribution by $k$ times a uniform bound, then compare that
quantity with the deterministic baseline already included in the
estimator.

\begin{lemma}[The complete polynomial below its low cutoff]
\label{lem:cc-below-cutoff}
For the coupled polynomial \eqref{eq:cc-coupled-Q}, suppose
$a\le b$, $0\le J_1,J_\delta,v_\delta\le1$ on $[0,1]$, and its
amplitudes are bounded. Then
\[
 |Q(p)|\le C^m a^\alpha\qquad(0\le p\le a).
\]
Consequently, the total absolute contribution of these coordinates
is at most $C^mka^\alpha$. For $\alpha\ge3/2$, relative to the actual
baseline $E_aP_*$ in Theorem~\ref{thm:cc-count}, its ratio is at most
\[
 \frac{2}{d_\varepsilon}(C/C_0)^m\,ka.
\]
This tends uniformly to zero as the final multiplier $D$ is increased,
since $ka\le\rho_+/D$.
\end{lemma}
\begin{proof}
Substitute $p=ay$, $0\le y\le1$, in the reverse polynomial.
Its rescaled coefficient sum is bounded by $C^m$, so
$|Q_-(p)|\le C^ma^\alpha$. The cutoff multiplier $1-J_1$ lies in
$[0,1]$. The complete high term is bounded in absolute value by
$Cp^\alpha\le Ca^\alpha$, since the high profiles and amplitudes
are bounded. Summation and
$ka^\alpha=(ka)^\alpha P_*$ prove the assertions.
\end{proof}
\subsection{One calibrated estimator and its matching upper bound}
\label{cu:section}

The preceding lemmas now have one concrete use: they prove that the
following hybrid power-sum estimate lies inside its prescribed interval.
The statement includes the full formula so that the estimator can be
read without reconstructing it from the technical lemmas. Its four
components are a deterministic correction, the known affine total,
the polynomial branch for a light pilot, and the Gamma branch for a
heavy pilot. In this subsection we treat noninteger orders
at least $3/2$; the near-one range is treated separately below.

\begin{theorem}[The common noninteger upper bound]
\label{cu:theorem}
Fix $0<\varepsilon<1/2$. For $k\ge k_\varepsilon$, uniformly over noninteger
$3/2\le\alpha<\log k$, the root $X$ in Lemma~\ref{upcal:parameters}
satisfies
\[
 S_\alpha(k,\varepsilon,1/3)
 \le C_\varepsilon\alpha k^{1-1/\alpha}X.
\]
Use the following construction in this range. Use the parameters
of Lemma~\ref{upcal:parameters}, with $c_\omega=1/8$, and the polynomial
$Q_0$ of \eqref{eq:coupled-single-polynomial}. Put
\[
 \chi=\min\{\alpha-1,1\},\quad
 s=(1+c_u u)^{-1},\quad Q=sQ_0=\sum_{j=0}^{d}q_jp^j,
 \quad \omega=s e^{-E_0/8},
\]
\[
 P_*=k^{1-\alpha},\qquad
 E_a=\chi d_\varepsilon C_0^m(ka)^{\alpha-1}.
\]
For independent $N_i\sim\operatorname{Poi}(np_i)$ and
$Y_i\sim\operatorname{Poi}(c_pnp_i)$ define
\[
 I_i=\mathbf1\{Y_i\le c_pB\alpha\},\quad
 \widehat Q_i=\sum_{j=0}^{d}q_j(N_i)_j/n^j,
 \quad \widehat G_i=\omega n^{-\alpha}\mathsf H_\alpha(N_i),
 \quad \mathsf H_\alpha(j)=\frac{\Gamma(j+1)}{\Gamma(j+1-\alpha)}.
\]
Use the entire reciprocal-Gamma value at its zeros, and set
\begin{equation}\label{cu:estimator}
 \widehat P=2E_aP_*+q_1+
 \sum_{i=1}^{k}\left[
 I_i(\widehat Q_i-q_1N_i/n)
 +(1-I_i)(\widehat G_i-q_1N_i/n)\right].
\end{equation}
All constants can be fixed in terms of $\varepsilon$ alone so that
\begin{equation}\label{cu:interval}
 \Pr_p\{e^{-2\varepsilon(\alpha-1)}P_\alpha(p)
               \le\widehat P\le P_\alpha(p)\}\ge 3/4.
\end{equation}
On a positive output return
$-\log\widehat P/(\alpha-1)-\varepsilon$, with an arbitrary output
otherwise. A fixed sample pool of size $C_\varepsilon n$ implements
the estimator with failure at most $1/3$.
\end{theorem}

\begin{proof}
Write $P=P_\alpha(p)$, $\kappa=e^{-E}$, and introduce the single margin
\begin{equation}\label{cu:margin}
 \mathcal M=E_aP_*+\kappa uP+\chi\kappa\mathcal B,
 \qquad \mathcal B=\sum_i p_i^\alpha\min(Lnp_i,1).
\end{equation}
The factor $\chi$ is necessary: both the logarithmic interval and the
affine-anchored fluctuation vanish at order one.

\medskip\noindent\textbf{Final assembly, Step 1: establish a deterministic interior.}
We first verify the deterministic margin and its parameter budget.
Let $C$ be the lower-exponent constant in
Proposition~\ref{prop:coupled-deterministic-bound} and take
$C_h\ge2C+2$. Since $u\le v+h\le2\beta$, choose
$c_u>0$ sufficiently small in terms of $\varepsilon$.
If $\Delta=v+C_hh$, then $\Delta-Ch\ge u$.
If $\Delta=E/2$ and $E\ge4Ch$, then
$\Delta-Ch\ge E/4=\varepsilon\beta/2\ge\varepsilon u/4$.
Thus, wherever $E\ge4Ch$, the same reserve pays the lower $uP$
margin and the common rescaling. The pointwise bound, the summed
baseline, and the bound below $a$ consequently give strict margins
$c_\varepsilon\mathcal M$ for any admissible light/heavy assignment.
The baseline is chosen to leave at least $8E_aP_*$ before adding
$2E_aP_*$. Such a choice is possible by the explicit summed bounds:
first fix $C_0$, then $d_\varepsilon$ sufficiently small, and finally
make $ka\le C_\varepsilon/D$ small by increasing $D$.

For clarity, the allowed heavy aggregate in this assertion may replace
$\omega P_{\mathcal H}$ by any value with lower margin
$c_\varepsilon\kappa(u+\chi)P_{\mathcal H}$ above
$\kappa P_{\mathcal H}$ and upper margin
$c_\varepsilon(u+\chi)P_{\mathcal H}$ below $P_{\mathcal H}$.
Here $u\le u_0$ and $\chi\le1$, so a fixed heavy interior margin
dominates the stated multiple of $u+\chi$.
Its fixed upper deficit also pays its share of the baseline, since
$\min\{h+Y+\min(Z,1),1\}\le1$. The light coordinates have
$p_i\le b$; coordinates below $a$ contribute only the explicitly
bounded baseline error. This is exactly the summation used in
Corollary~\ref{cor:coupled-common-bias-transfer}.

The complementary inequality $E<4Ch$ is handled by the other bound for
the same polynomial, namely its localization toward $\omega p^\alpha$.
For $\alpha\ge3/2$, this inequality confines $\alpha$ and $h$ to
fixed compact ranges, and Lemma~\ref{upcal:bounded-width} gives
$kt\le C_\varepsilon/D$ and $u\ge c_\varepsilon$, with
$t=b/N^2$. The localized error, with $\nu=\min\{\alpha,2\}$, sums to
at most $C_\varepsilon(kt)^\nu P$; below $a$ its bound is
$C^{m+1}(ka)^\alpha P_*$. The deterministic correction is at most
$C_\varepsilon D^{-\beta}P$. All tend to zero uniformly as $D$
increases. Both gaps $1-\omega$ and $\omega-\kappa$ have fixed
positive lower bounds here, while
$c_\varepsilon P\le\mathcal M\le C_\varepsilon P$.
An aggregate heavy value within $\eta_\varepsilon P_{\mathcal H}$ of
$\omega P_{\mathcal H}$ is therefore admissible, for a fixed
sufficiently small $\eta_\varepsilon>0$.
These two bounds verify one deterministic bound; the polynomial,
pilot, root, rescaling, and correction in \eqref{cu:estimator} do not
change between them.

\medskip\noindent\textbf{Final assembly, Step 2: condition on the pilot and build an ideal comparison.}
The comparison statistic uses polynomial expectations on its light
coordinates and a concentrated heavy aggregate on the others. It is
chosen only for analysis; the estimator does not know $p_i$ or any
of the sets defined below.
Choose the true-heavy threshold $B\alpha/2$ with $B$ large.
The count budget has $nb=c_bL\ge2B\alpha$ and $na<B\alpha/2$;
the latter follows directly from the capacity equation after fixing
the degree constant. For a realization of the pilot define
\[
 \mathcal L=\{i:np_i<B\alpha/2\}
  \cup\{i:B\alpha/2\le np_i\le nb,\ I_i=1\},
 \qquad \mathcal H=\mathcal L^c,
\]
\[
 \mathcal G=\{i:I_i=0,\ np_i\ge B\alpha/2\},
 \qquad \mathcal F=\{i:p_i>b,\ I_i=1\}.
\]
Thus every light coordinate is at most $b$, and
$\mathcal H=\mathcal G\mathbin{\dot\cup}\mathcal F$.
Conditional on the pilot, Lemma~\ref{lem:cc-heavy-norm} applies to
$\mathcal G$. By choosing its relative-root tolerance sufficiently
small, the aggregate
\[
 V_{\mathcal H}=\sum_{i\in\mathcal G}\widehat G_i
                         +\omega P_{\mathcal F}
\]
is admissible for the preceding deterministic bound except with
a prescribed small fixed probability. Indeed
$-\log\omega=E_0/8+\log(1+c_uu)$ lies strictly between zero and $E$,
with gaps proportional to $\beta$. The logarithmic fluctuation of
the heavy norm is at most a fixed small multiple of $\alpha$;
$\alpha/\beta\le3$ converts it into a small part of those gaps.
When $E<4Ch$, $\alpha$ is bounded, so the same fixed choice of tolerance
also supplies the stronger absolute error
$\eta_\varepsilon P_{\mathcal H}$. Increasing $B$ enforces both
requirements simultaneously, including for the empty set.

\medskip\noindent\textbf{Final assembly, Step 3: bound the difference from the ideal comparison.}
On this event, the ideal mixed statistic
\[
 Z_0=2E_aP_*+\sum_{i\in\mathcal L}Q(p_i)+V_{\mathcal H}
\]
lies in $[\kappa P+c\mathcal M,P-c\mathcal M]$, with a fixed
$c=c_\varepsilon>0$. We compare the actual statistic to $Z_0$.
Its centered conditional polynomial fluctuation has unconditional variance
$\mathbb E\operatorname{Var}(\widehat Q\mid Y)$ at most
$\eta\mathcal M^2$ by Theorem~\ref{thm:cc-count} and
Lemma~\ref{upcal:baseline}. The factor $\chi\in[1/2,1]$ only changes
fixed constants in that theorem. The affine replacement contributes
$q_1(1-\sum_iN_i/n)$, with variance
$q_1^2/n\le C_\varepsilon\mathcal M^2/n$ by
Lemma~\ref{lem:cc-anchor}. These two terms are smaller than a fixed
fraction of $c\mathcal M$ with arbitrarily small fixed failure
probability, by Chebyshev's inequality and the stated choices of
$C_0,D,k_\varepsilon$.

The remaining difference consists precisely of
\[
 \sum_{i\in\mathcal F}\{Q(p_i)-\omega p_i^\alpha\}
 -\sum_{\substack{np_i<B\alpha/2\\I_i=0}}Q(p_i)
 +\sum_{\substack{np_i<B\alpha/2\\I_i=0}}\widehat G_i.
\]
Lemma~\ref{lem:cc-false-light} bounds the expectation of the absolute
polynomial part of the first term by $Ce^{-cL}\mathcal B$;
the same pilot lower tail bounds $\omega P_{\mathcal F}$, since
$Lnp_i\ge1$ on $p_i>b$. For the second term with
$p_i\ge a$, use $|Q(p_i)|\le Cp_i^\alpha$ and the pilot upper tail;
its bound is $Ce^{-c\alpha}\mathcal B$. The portion below $a$ is
at most $C^{m+1}ka^\alpha$, which is an arbitrarily small part of
$E_aP_*$ by Lemma~\ref{lem:cc-below-cutoff} and the choice of $D$.
The last, signed Gamma term has the same exponentially small
expectation by Lemma~\ref{lem:cc-false-heavy}.
Choose the pilot constants so that the exponents dominate
$\kappa^{-1}=e^{2\varepsilon\beta}$. Markov's inequality then makes
their total smaller than the remaining fixed fraction of
$c\mathcal M$, with a prescribed small fixed failure probability.
No union bound over the alphabet is used.

\medskip\noindent\textbf{Final assembly, Step 4: combine the events and return to fixed samples.}
A union bound over these fixed probability budgets proves
\eqref{cu:interval}, with slack; for example each of the heavy event,
the combined centered fluctuation, and the boundary errors can be
allocated failure $1/32$. 

The hierarchy is consistent: fix the profiles and damping, their
bias constants and $C_h$, then the pilot and degree constants; choose
$c_L$ sufficiently small; determine the cutoff comparison constants;
choose $C_0$, then $d_\varepsilon$ sufficiently small, and finally
$D$ and $k_\varepsilon$. The exact powers of $D$ in
\eqref{eq:cc-D-power} ensure that no earlier constant is required to
depend on $D$. Every choice depends only on the original
$\varepsilon$, and $E=2\varepsilon\beta$ has never been replaced in
the root.

Finally the conversion of \eqref{cu:interval} into entropy accuracy is
immediate. Independent Poisson main and pilot samples are generated
from a fixed pool of $\lceil4(1+c_p)n\rceil$ iid observations by
Poisson thinning; the pool overflow probability is exponentially small
in $n$. The strict probability slack absorbs it, giving the claimed
fixed-sample theorem.
\end{proof}

\clearpage
\section{The near-one upper bound}
\label{sec:near-one-proof}
\label{nuc:section}

This section proves the upper bound for $1<\alpha\le3/2$, including
the random pilot and signed Gamma arm. We use the regularized root and
a degree-one low profile. Auxiliary constants are chosen for this range;
they need not coincide with those used for orders at least $3/2$.

\paragraph{Why a separate calculation is necessary.}
When $\beta\downarrow0$, the admissible interval has width
\[
 (1-e^{-2\varepsilon\beta})P
   =2\varepsilon\beta P+O(\beta^2P).
\]
Thus the required absolute error in the power sum vanishes with
$\beta$. The calculation below shows that the estimator's error
vanishes at the same rate. The central cancellation is elementary:
at $\beta=0$, the target $p^{1+\beta}$ equals $p$, whose total mass
is known exactly. The technical work is to keep this cancellation
through the polynomial, the Gamma statistic, and the random pilot.

Fix $0<\varepsilon<1/2$. For $1<\alpha\le3/2$, write
$\beta=\alpha-1$, $\delta=1-\beta$, $h=\beta\delta$, and
$E=2\varepsilon\beta$. Define
\[
 \Psi=\beta^{1+1/\delta}\delta^{1+1/\beta},\qquad
 Q_*^2=\frac{\Psi}{\alpha^2\beta^2}
                  e^{(\log k)/\alpha-E/h},
\]
\[
 X=\inf\left\{x\ge1:\frac{\alpha h}{\beta}x^\alpha
 \left[1+\log\frac{x(1+\log x)}{Q_*^2}\right]_+\ge1\right\}.
\]
Use the following parameterization, with fixed constants chosen below,
\begin{align*}
 v_0&=\frac\beta\alpha X^{-\alpha},&
 E_0&=E-\min\{v_0+C_hh,E/2\},\\
 n&=D\alpha k^{1-1/\alpha}X,&
 L&=\lceil\alpha(C_L+c_L\log X)\rceil,&
 b&=c_bL/n.
\end{align*}
Let $z\in(0,1)$ solve $G_\beta(z)=e^{-E_0}$, where
$G_\beta$ is the exact envelope function in Lemma~\ref{oc:envelope},
and set
\[
 a=\lambda\frac{b h z^{1/h}}{L^2},\qquad
 A=\frac{1-z^{1/\beta}}{1-z^{1/h}},\qquad
 B=\frac{1-z^{1/\delta}}{1-z^{1/h}},
\]
\[
 A'=Ae^{-M\beta},\qquad B'=Be^{-M\delta},\qquad
 \omega_0=e^{-c_0E_0},\quad 0<c_0<1.
\]
Also use the common final rescaling
\[
 u_0^*=\min\{v_0+h,u_0\},\qquad
 s_0=(1+c_u u_0^*)^{-1},\qquad \omega=s_0\omega_0.
\]
Here $u_0>0$ is fixed and $c_u>0$ is chosen sufficiently small in terms
of $\varepsilon,c_0$. The numbers $C_h,c_b,\lambda,M,c_0$ are fixed parameters
chosen for this construction. Let the integer kernel scale $N$ satisfy $N\asymp L$,
with fixed comparison constants, and put $t=b/N^2$.

\paragraph{The common degree-one profile.}
Fix the same floor $q_L\in(0,1)$ as in the common low construction.
For $0<\beta\le1/2$, choose the unique
$c_\beta\in(\beta/(1+\beta),1)$ satisfying
\[
 (1-c_\beta)c_\beta^\beta(1+\beta)^{1+\beta}/\beta^\beta=q_L.
\]
The function on the left decreases strictly from one to zero on this
interval. Define
\[
 A_\beta^{\rm low}
 =\frac{\beta^\beta}{c_\beta^\beta(1+\beta)^{1+\beta}},\qquad
 u_\beta(x)=\frac{x^\beta(1-c_\beta x)}{A_\beta^{\rm low}},\qquad
 \tau=\frac{c_\beta(1+\beta)a}{\beta}.
\]
The maximum of $x^\beta(1-c_\beta x)$ occurs at
$\beta/[c_\beta(1+\beta)]<1$, so the normalization is exact and
$u_\beta(1)=q_L$. Since $1/4\le A_\beta^{\rm low}\le1$,
\[
 1-q_L\le c_\beta\le1-q_L/4.
\]
Thus, as a polynomial identity, including its extension to $p=0$,
\begin{equation}\label{nuc:eq:exact-low}
 p^\alpha u_\beta(a/p)=T_\tau(p),\qquad
 T_\tau(p)=(1+\beta)\tau^\beta p-\beta\tau^{1+\beta}.
\end{equation}
No differentiation of $c_\beta$ is needed below.

For clarity, the high profiles used here are the same localized profiles
as in the common construction. Write
\[
 R_N(x)=\left\{U_{N-1}(\sqrt{1-x})/N\right\}^{6},\quad
 K_N=R_N^2,\quad W=f+(1-f)(1-R_N)^2,\qquad 0<f<1,
\]
\[
 I_s=\int_0^1x^{s-1}K_N(x)\,dx,\qquad
 J_s(x)=I_s^{-1}\int_0^x u^{s-1}K_N(u)\,du,
\]
\[
 A_{s,\omega}=\sup_{0\le x<1}
       \frac{x^sK_N(x)}{1-\omega J_s(x)(1-K_N(x))W(x)},
\quad
 V_{s,\omega}(x)=\frac{x^sK_N(x)}{A_{s,\omega}}
                +\omega J_s(x)(1-K_N(x))W(x).
\]
Use the continuous extension of the quotient at a zero of $K_N$.
The coupled polynomial is
\begin{equation}\label{nuc:eq:polynomial}
 Q(p)=s_0\Bigl[A' T_\tau(p)\{1-J_1(p/b)\}
       +\omega_0 p^\alpha J_\delta(p/b)
       +B'p^\alpha\{V_{\delta,\omega_0}(p/b)
                       -\omega_0 J_\delta(p/b)\}\Bigr].
\end{equation}
It is a polynomial because $\alpha+\delta=2$.
Write $Q(p)=\sum_{j=0}^{d}q_jp^j$, with $d\le C_dL$.

The parameter $t=b/N^2$ is the effective endpoint scale of the
polynomial kernel. In this range, the root simplifies enough to show
that $n$ has order $k/\log k$ and that $kt$ can be made small by the
sample multiplier $D$. These two facts will control aggregate bias.

\begin{lemma}[Near-one parameter estimates]
\label{nuc:lem:parameters}
Uniformly over $1<\alpha\le3/2$, for sufficiently large $k$,
\[
 X\asymp_\varepsilon k^{1/\alpha}/\log k,\qquad
 n\asymp_\varepsilon Dk/\log k,\qquad
 L\asymp c_L\log k,\qquad kt\le C_\varepsilon/D.
\]
The ratio $\tau/t$ lies in a fixed positive compact interval independent
of $D$, and
\[
 A'=1+O_\varepsilon(\beta),\quad B'=O_\varepsilon(\beta),\quad
 \omega=1+O_\varepsilon(\beta).
\]
Constants may depend on the fixed profile and degree constants.
\end{lemma}
\begin{proof}
Direct cancellation gives
\[
 Q_*^2/k^{1/\alpha}
 =\frac{\beta^{\beta/\delta}\delta^{1+1/\beta}}{\alpha^2}
                                     e^{-2\varepsilon/\delta}
 \asymp_\varepsilon1.
\]
The logarithm is uniformly bounded, using
$\beta|\log\beta|\le1/e$ and
$\log(1-\beta)\ge-\beta/(1-\beta)$.
Let $y\ge1$ solve $y(1+\log y)=Q_*^2/e$. Uniformly,
$y\asymp_\varepsilon k^{1/\alpha}/\log k$.
The root expression vanishes at $y$, and at $2y$ is at least
$(3/4)(2y)^\alpha\log2>1$ for sufficiently large $k$.
Thus $y\le X\le2y$, proving the sample and degree comparisons.
The bound on $kt$ follows from $t\asymp1/(nL)$.

Write $z=e^{-ht_*}$. Since $E_0\in[\varepsilon\beta,2\varepsilon\beta]$,
\[
 \varepsilon\le E_0/h\le4\varepsilon.
\]
The overlap inequalities, with $J(x)=x\coth x$,
$2[J(t_*/4)-1]\le F_\beta(t_*)/h\le t_*$ therefore put $t_*$
in a fixed positive compact interval. Consequently
\[
 \frac\tau t
   =c_\beta\lambda\frac{N^2}{L^2}\alpha\delta z^{1/h}
\]
is bounded above and below by positive constants, independently of $D$.
Also
$B=(1-e^{-\beta t_*})/(1-e^{-t_*})=O_\varepsilon(\beta)$ and
$1-A=e^{-\delta t_*}B=O_\varepsilon(\beta)$.
The attenuations preserve the asserted bounds. Since
$u_0^*\le v_0+h\le2\beta$, both $s_0$ and $\omega$ equal
$1+O_\varepsilon(\beta)$.
\end{proof}

\paragraph{The decisive near-one approximation estimates.}
The first inequality below controls bias after summing over symbols.
The next two identify the constant and nearly linear coefficients.
The last shows that every coefficient of degree at least two contains
a factor $\beta$. This factor is what ultimately yields variance of
order $\beta^2P^2$.

\begin{lemma}[Localized bias and affine cancellation]
\label{nuc:lem:bound}
For the polynomial \eqref{nuc:eq:polynomial}, uniformly in the near-one
range,
\begin{align}
 \sup_{0\le p\le b}|Q(p)-\omega p^\alpha|
       &\le C_\varepsilon\beta t^\alpha,
                                      \label{nuc:eq:bias}\\
 |q_0|&\le C_\varepsilon\beta t^\alpha,
 & |q_1-t^\beta|&\le C_\varepsilon\beta t^\beta,
                                      \label{nuc:eq:affine}\\
 |q_j|&\le C_\varepsilon\beta e^{CL}b^{\alpha-j},
 &&2\le j\le d.                    \label{nuc:eq:coefficients}
\end{align}
Here $C$ is independent of $D$.
\end{lemma}
\begin{proof}
The sine identity for $U_{N-1}$ gives
$0\le K_N(x)\le\min\{1,(N^2x)^{-6}\}$ and a positive local lower
bound. In the scaled variable $Y=N^2x$, for $1/2\le s\le1$ the
normalizing integral and its first absolute logarithmic moment are
bounded above and below by constants. Thus, with $F_s(Y)=J_s(Y/N^2)$,
\[
 1-F_1(Y)\le C\min(1,Y^{-5}),\qquad
 |F_{1-\beta}(Y)-F_1(Y)|
 \le C\beta\begin{cases}1,&Y\le1,\\
                         Y^{-5}(1+\log Y),&Y\ge1.
          \end{cases}
\]
The second estimate follows by differentiating the normalized density
in $s$; for $Y\ge1$, differentiate its tail integral.
The ordinary tangent inequality and the mean value theorem now give
\[
 \sup_{p\le b}
 |T_t(p)(1-J_1(p/b))+p^\alpha J_\delta(p/b)-p^\alpha|
 \le C\beta t^\alpha.
\]
Indeed, after scaling by $t$, write the error as
\[
 \{(1+\beta)Y-\beta-Y^{1+\beta}\}(1-F_1(Y))
       +Y^{1+\beta}(F_{1-\beta}(Y)-F_1(Y)).
\]
Each term is uniformly $O(\beta)$ by the preceding tail bounds.
Replacing $t$ by $\tau$ in the tangent changes this error by
$O_\varepsilon(\beta t^\alpha)$ because $\tau/t$ is in a fixed
positive compact interval and $p(1-J_1(p/b))\le Ct$.

The high normalization lies between the ordinary sharp normalization
and its tail-one counterpart. For $1/2\le\delta\le1$ it follows that
$A_{\delta,\omega_0}^{-1}\le CN^{2\delta}$. Its exact defining expression
then gives
\[
 |V_{\delta,\omega_0}(x)-\omega_0 J_\delta(x)|
       \le C\min\{Y^\delta,Y^{-3}\},\qquad Y=N^2x.
\]
To check the large-$Y$ bound, expand the difference into
$x^\delta K_N/A_{\delta,\omega_0}-\omega_0 J_\delta K_N
 -(1-f)\omega_0 J_\delta(1-K_N)(2R_N-R_N^2)$.
Consequently its product with $p^\alpha$ is bounded by $Ct^\alpha$.
Now subtract $\omega$ times the tangent-cutoff expression from
\eqref{nuc:eq:polynomial}. The remaining coefficients
$s_0A'-\omega$ and $s_0B'$ are $O_\varepsilon(\beta)$.
This proves \eqref{nuc:eq:bias}.

The constant and linear coefficients are exactly
\[
 q_0=-s_0A'\beta\tau^\alpha,\qquad
 q_1=s_0A'\left\{(1+\beta)\tau^\beta
                    +\beta\tau^\alpha\frac{J_1'(0)}b\right\}.
\]
Since $J_1'(0)/b=1/(bI_1)\asymp1/t$, this proves
\eqref{nuc:eq:affine}.
For \eqref{nuc:eq:coefficients}, rescale by $p=bx$.
The kernel coefficients have sum at most $C^N$ by the Chebyshev
recurrence. The integrated coefficients are
$[x^j]K_N/((j+s)I_s)$; for $1/2\le s\le1$, their derivatives in $s$
have sum at most $C^N(1+\log N)$.
At $\beta=0$, the two undamped expressions cancel coefficientwise
as $x(1-J_1)+xJ_1=x$. The factors
$s_0A'-1$, $\omega-1$, and $s_0B'$ are $O_\varepsilon(\beta)$;
also $\log(\tau/b)=O_\varepsilon(1+\log N)$.
These observations give the stated factor $\beta$ for every degree
at least two. The high term needs no derivative of its maximum
normalization because its prefactor $B'$ already contains that factor.
All powers and logarithms of $N$ are absorbed into $e^{CL}$.
\end{proof}

\paragraph{Assemble the near-one estimator.}
The statistic is the same polynomial/Gamma hybrid as before, with the
same global anchor. Its deterministic correction is denoted $r_k$
here to make clear exactly how much correction the proof can absorb.
The estimate need not be positive on every outcome; it is positive
on the success event, which is all that is needed for the guarantee.

\begin{theorem}[Near-one estimation theorem]
\label{nuc:thm:upper}
There are fixed choices of the pilot constants $B_0,c_p$, the degree
constants $C_L,c_L$, and the multiplier $D$, and a $k_\varepsilon$,
such that the following holds for every $k\ge k_\varepsilon$ and every
$1<\alpha\le3/2$.
Let independently
\[
 N_i\sim\operatorname{Poi}(np_i),\qquad
 Y_i\sim\operatorname{Poi}(c_pnp_i),\qquad
 I_i=\mathbf1\{Y_i\le c_pB_0\alpha\}.
\]
Put
\[
 \widehat Q_i=\sum_{j=0}^dq_j(N_i)_j/n^j,\qquad
 \widehat G_i=\omega n^{-\alpha}\mathsf H_\alpha(N_i),
\]
and use the single global affine convention
\begin{equation}\label{nuc:eq:estimator}
 \widehat P=q_1+r_k+
  \sum_i\left\{I_i(\widehat Q_i-q_1N_i/n)
       +(1-I_i)(\widehat G_i-q_1N_i/n)\right\}.
\end{equation}
Here the same deterministic repair used by the common construction is
allowed whenever $|r_k|\le\eta_0\beta k^{-\beta}$, for a sufficiently
small fixed $\eta_0>0$. In particular this includes
\[
 r_k=2\beta d_\varepsilon C_0(ka)^\beta k^{-\beta},
\]
after choosing $d_\varepsilon$ sufficiently small, since $ka\le1$
for sufficiently large $D$. A further term bounded by $Cka^\alpha$
may also be included.
Then, uniformly over $p\in\Delta_k$,
\[
 \operatorname{Prob}\{e^{-2\varepsilon\beta}P_\alpha(p)
                       \le\widehat P\le P_\alpha(p)\}\ge3/4.
\]
Thus $-\beta^{-1}\log\widehat P-\varepsilon$, with any definition
when $\widehat P\le0$, estimates $H_\alpha$ within $\varepsilon$.
The sample complexity is $O_\varepsilon(k/\log k)$, equivalently
$O_\varepsilon(\alpha k^{1-1/\alpha}X)$ uniformly on this range.
\end{theorem}
\begin{proof}
Write $P=\sum_i p_i^\alpha\ge k^{-\beta}$ and $K=\log k$.
All constants below are uniform over $0<\beta\le1/2$.
The bounds $E_0\in[\varepsilon\beta,2\varepsilon\beta]$ and
$u_0^*\le2\beta$ imply, for example when
$c_u\le(1-c_0)\varepsilon/4$,
\begin{equation}\label{nuc:eq:room}
 \min\{1-\omega,\omega-e^{-2\varepsilon\beta}\}
                                      \ge c_\varepsilon\beta.
\end{equation}

\medskip\noindent\textbf{Near-one proof, Step 1: control the mean.}
We first center the analysis at $\omega P$. The previous display says
this center is separated from both endpoints by a constant multiple
of $\beta P$.
First use the anchor $t^\beta$ in place of $q_1$ and omit $r_k$;
call this statistic $\widehat P_t$. The replacement will be controlled
at the end. Let $s(\lambda)$ denote the probability of the light pilot.
The analytic Gamma mean identity gives
\[
 |\mathbb E\mathsf H_\alpha(\operatorname{Poi}(\lambda))
                                      -\lambda^\alpha|
                              \le C\beta e^{-\lambda}
\]
including $\lambda=0$, where the value is $1/\Gamma(-\beta)$.
For every prescribed $\eta>0$, fixed sufficiently large $B_0,c_p$
ensure
\[
 (1-s(\lambda))e^{-\lambda}\le\eta\lambda^\alpha
                                      \quad(\lambda\ge0).
\]
For $\lambda\ge B_0/2$ use exponential decay; for
$1\le\lambda<B_0/2$ use the pilot upper tail; for $0<\lambda<1$ use
$(1-s(\lambda))\le(e\lambda/(B_0\alpha))^{c_pB_0\alpha}$
with $c_pB_0\ge3$. Thus the aggregate Gamma mean error is at most
$C\eta\beta P$.

On $p_i\le b$, Lemma~\ref{nuc:lem:bound} bounds the polynomial
mean error by $C_\varepsilon\beta k t^\alpha$.
For $p>b$, the affine coefficient bounds imply
\[
 |Q(p)-t^\beta p|\le C_\varepsilon\beta e^{CL}b^\alpha
                                      \{1+(p/b)^d\}.
\]
Choose $C_L$ so $nb=c_bL\ge2B_0\alpha$.
For $\lambda=np\ge nb$, the pilot lower tail gives
$s(\lambda)\le e^{-c_p\lambda/8}$.
Writing $w=p/b\ge1$, the exponent controlling the polynomial tail is
$CL+C_dL\log w-c_pc_bLw/8$. Increasing $c_p$ makes its supremum
arbitrarily negative, uniformly in $L\ge1$.
The remaining affine difference satisfies
\[
 |t^\beta p-\omega p^\alpha|/p^\alpha
                 \le C_\varepsilon\beta+\beta\log(p/t),
\]
and the same pilot exponential absorbs
$\log(p/t)=\log w+2\log N$.
Therefore the false-light mean beyond $b$ is at most $C\eta\beta P$.
We have proved
\begin{equation}\label{nuc:eq:mean-final}
 |\mathbb E\widehat P_t-\omega P|
 \le\beta P\{C_\varepsilon(kt)^\alpha+C\eta\}.
\end{equation}

\medskip\noindent\textbf{Near-one proof, Step 2: control all randomness.}
We next control the entire variance, including the random pilot.
The discrete Poisson Poincar\'e inequality used below is
$\Var f(Z)\le\lambda\E(f(Z+1)-f(Z))^2$ for
$Z\sim\Poi(\lambda)$. Thus a factor $\beta$ in each discrete
increment becomes a factor $\beta^2$ in variance.
Put $\xi=nt$. The analytic recurrence and the order-one increment
bound give, for all integers $j\ge0$,
\[
 |\Delta\{\omega\mathsf H_\alpha(j)-\xi^\beta j\}|
 \le C_\varepsilon\beta\{(j+1)^\beta+\xi^\beta\}
                         \{1+|\log((j+1)/\xi)|\}.
\]
For $Z\sim\operatorname{Poi}(\lambda)$ and
$V=(Z+1)/(\lambda+1)$,
$\mathbb EV^q+\mathbb E[V^q(\log V)^2]\le C$ uniformly in
$\lambda\ge0$ and $0\le q\le1$.
For $\lambda\ge1$ use the usual Poisson lower tail and fixed moments;
for $\lambda\le1$, $V\ge1/2$ and its fixed moments are bounded.
Poisson Poincar\'e consequently gives
\[
 \operatorname{Var}\{\omega\mathsf H_\alpha(Z)-\xi^\beta Z\}
 \le C_\varepsilon\beta^2\lambda
  \{(\lambda+1)^{2\beta}+\xi^{2\beta}\}
                    \{1+\log^2((\lambda+1)/\xi)\}.
\]
Since $\xi\asymp1/L$, $\lambda\le n$, and $2\beta\le1$, summing
in probability units yields
\begin{equation}\label{nuc:eq:heavy-variance}
 \sum_i\operatorname{Var}(\widehat G_i-t^\beta N_i/n)
 \le\frac{C_\varepsilon\beta^2K^2}{n}
      \{P_{2\alpha-1}+n^{-2\beta}+t^{2\beta}\}.
\end{equation}
The inequality $P_{2\alpha-1}\le P^{2-1/\alpha}$ implies that the
right side divided by $\beta^2P^2$ is at most
\[
 \frac{C_\varepsilon K^2}{n}
 \{k^{\beta/\alpha}+(k/n)^{2\beta}+(kt)^{2\beta}\}=o(1),
\]
uniformly, because $\beta/\alpha\le1/3$ and $n\asymp k/K$.

\medskip\noindent\textbf{Near-one proof, Step 3: account for switching branches.}
The preceding estimate concerns the Gamma residual. The actual
statistic replaces it by the polynomial residual whenever the pilot
selects light. The following two second-moment bounds control that
replacement directly, including the selector's randomness.
For the pilot correction, the same affine coefficient bounds and
factorial second-moment identity prove
\begin{equation}\label{nuc:eq:light-second}
 \mathbb E[I_i(\widehat Q_i-t^\beta N_i/n)^2]
                 \le C_\varepsilon\beta^2e^{CL}b^{2\alpha}.
\end{equation}
Indeed, if $\lambda\le H=nb$, then
$\mathbb E(Z)_j^2\le\lambda^j(\lambda+j)^j
 \le H^{2j}(1+d/H)^j$.
For $\lambda>H$, multiply
$\mathbb E(Z)_j^2\le\lambda^{2j}e^{j^2/\lambda}$
by $s(\lambda)\le e^{-c_p\lambda/8}$.
With $\lambda=Hw$ and $j\le C_dL$, the remaining factor has logarithm
at most $L\{C+C\log w-cw\}$ and hence is $e^{O(L)}$.
The $L^2$ triangle inequality across the degrees proves
\eqref{nuc:eq:light-second}.

A similar, simpler bound holds for the signed Gamma residual under the
light pilot. Reciprocal-Gamma differentiation on $j\ge1$ and the value
at $j=0$ give
\[
 |\omega\mathsf H_\alpha(j)-\xi^\beta j|
 \le C_\varepsilon\beta(j+1)^\alpha
                      \{1+|\log\xi|+\log(j+1)\}.
\]
Fixed Poisson moments and the exponential tail of $s(\lambda)$ then
imply
\begin{equation}\label{nuc:eq:gamma-second}
 \mathbb E[I_i(\widehat G_i-t^\beta N_i/n)^2]
       \le C_\varepsilon\beta^2n^{-2\alpha}(1+\log L)^2.
\end{equation}
This estimate includes all signed low-count values.

Now write the statistic exactly as
\[
 \widehat P_t=t^\beta+
  \sum_i(\widehat G_i-t^\beta N_i/n)
       +\sum_i I_i(\widehat Q_i-\widehat G_i).
\]
Independence across coordinates, the variance triangle inequality,
and \eqref{nuc:eq:heavy-variance}--\eqref{nuc:eq:gamma-second} control
its full variance, without treating the pilot as deterministic.
The additional terms divided by $\beta^2P^2$ are bounded by
\[
 C_\varepsilon e^{CL}K^6/k
       +C_\varepsilon K^3(1+\log K)^2/k=o(1)
\]
after choosing $c_L$ sufficiently small. Thus
$\operatorname{Var}\widehat P_t=o(\beta^2P^2)$ uniformly.

\medskip\noindent\textbf{Near-one proof, Step 4: restore the prescribed estimator.}
Finally, the global $q_1$ convention changes $\widehat P_t$ by
\[
 (q_1-t^\beta)\left(1-\sum_iN_i/n\right).
\]
Its mean is zero and its variance is at most
$C_\varepsilon\beta^2t^{2\beta}/n=o(\beta^2P^2)$.
The deterministic repair contributes at most $\eta_0\beta P$.
Any additional $O(ka^\alpha)$ term has ratio to $P$ at most
$C_\varepsilon\beta(kt)^\alpha$, because $a/t=O_\varepsilon(\beta)$.
Choose $\eta,\eta_0$ small, then $D$ sufficiently large in
\eqref{nuc:eq:mean-final}, and finally $k_\varepsilon$ sufficiently
large for the uniform variance bound.
Chebyshev's inequality and \eqref{nuc:eq:room} prove the asserted
probability interval with strict slack in the failure budget.
For a fixed-sample implementation, draw independently
$M_0\sim\operatorname{Poi}((1+c_p)n)$ and allocate each of the first
$M_0$ observations to the main or pilot sample with probabilities
$1/(1+c_p)$ and $c_p/(1+c_p)$. Use a pool of
$\lceil4(1+c_p)n\rceil$ iid observations; if $M_0$ exceeds the pool,
return an arbitrary answer. The ideal counts have exactly the independent
laws used above, and the overflow probability is exponentially small
in $n$. The initial strict slack absorbs this additional failure.
\end{proof}

\clearpage
\section{Lower bounds: two random distributions that the data cannot distinguish}

An upper bound supplies an estimator. A lower bound must rule out
\emph{every} estimator. We do this by constructing two random choices of
an unknown distribution. Their entropies are separated by more than
$2\varepsilon$, but their sample laws are close in total variation.
An estimator accurate to $\varepsilon$ would then distinguish two
statistically indistinguishable hypotheses, a contradiction.

The proof has four stages. First, construct two nonnegative random
weights $U_0,U_1$ with mean one. Their low integer moments agree, whereas
their $\alpha$th moments differ. Second, draw many such weights and
normalize them to obtain actual probability vectors. Third, show that
normalization preserves the entropy separation. Finally, bound the
statistical distance between the resulting count data. The parameter
calculation at the end chooses the weights so that this argument works
at the same sample scale as the upper bound.

\paragraph{Notation for this section.}
The subscript $i\in\{0,1\}$ denotes a hypothesis, not a symbol in the
alphabet. The random variable $U_i$ is one unnormalized weight;
$A_i=\mathbb E U_i^\alpha$ is its power moment. The letter $B$ bounds
these weights. Later, $b=k/2$ is the expected number of weights and
$\theta$ is the sampling intensity per unit weight. Thus $\theta b$,
not $\theta$, is the expected total count before normalization.
A signed measure $\lambda=\lambda^+-\lambda^-$ is written as the
difference of its positive and negative parts; their total masses need
not be equal. The notation $\delta_a$ means unit point mass at $a$.

\subsection{Why normalization must be analyzed explicitly}
\label{sec:common-noninteger-priors}

The weights do not initially sum to one. A tempting argument would
replace their sum by its expectation and ignore the error. That loses
uniformity when $\alpha$ approaches one. The following identity keeps
exactly the cancellation that is needed in that limit.

\begin{lemma}[Normalization with the affine term removed]
\label{unif:lem:normalization}
Let $\alpha>1$, $\beta=\alpha-1$, and let $U\ge0$ satisfy
$\mathbb EU=1$ and $A=\mathbb EU^\alpha<\infty$.
Consider either of the following experiments:
\begin{enumerate}
\item $b=N$ is a positive integer and $U_1,\ldots,U_N$ are independent
copies of $U$;
\item $b>0$, $Q\sim\operatorname{Poi}(b)$, and $U_1,\ldots,U_Q$ are
independent copies of $U$, independent of $Q$.
\end{enumerate}
In the first experiment sums below have $N$ terms, and in the second
they have $Q$ terms. Define
\[
 S=\frac1b\sum_jU_j,\qquad
 V=\frac1b\sum_jU_j^\alpha,\qquad W=V-AS.
\]
On $S>0$, put $p_j=U_j/(bS)$. Then
\begin{equation}\label{unif:eq:normalization-identity}
 \beta\left\{H_\alpha(p)-\left(\log b-\frac{\log A}{\beta}\right)\right\}
 =\beta\log S-\log\left(1+\frac{W}{AS}\right).
\end{equation}
For $0<m\le1$, set
\[
 t=\min\left\{\frac14,\frac{m}{16\beta}\right\},
 \qquad u=\frac m{16}.
\]
On $|S-1|\le t$ and $|W|\le Au$, the absolute entropy error in
\eqref{unif:eq:normalization-identity} is less than $m/(4\beta)$.
If the second moments below are finite, the failure probability is at most
\begin{equation}\label{unif:eq:normalization-failure}
 \frac{v_1}{bt^2}+\frac{v_\alpha}{bA^2u^2},
\end{equation}
where, in the independent experiment,
\[
 v_1=\operatorname{Var}(U),\qquad
 v_\alpha=\mathbb E(U^\alpha-AU)^2,
\]
and, in the compound-Poisson experiment,
\[
 v_1=\mathbb EU^2,\qquad
 v_\alpha=\mathbb E(U^\alpha-AU)^2.
\]
\end{lemma}
\begin{proof}
\textbf{Step 1: isolate the normalization error.}
The actual probability is $p_j=U_j/(bS)$, so
\[
 \sum_jp_j^\alpha
 =\frac{\sum_jU_j^\alpha}{b^\alpha S^\alpha}
 =b^{1-\alpha}\frac{V}{S^\alpha}.
\]
If $S=1$ and $V=A$, its entropy would be
$\log b-(\log A)/\beta$. The difference from this deterministic center
is the quantity we must control.
Substitute $V=AS+W$ and take logarithms to obtain
\eqref{unif:eq:normalization-identity}.
On the stated event $S\ge3/4$. Hence
\[
 |\beta\log S|\le\frac43\beta t\le\frac m{12},
 \qquad
 \left|\frac W{AS}\right|\le\frac m{12}.
\]
The inequality $|\log(1+x)|\le|x|/(1-|x|)$ for $|x|<1$ gives
\[
 \left|\log\left(1+\frac W{AS}\right)\right|
 \le\frac{m/12}{1-m/12}\le\frac m{11}.
\]
Since $1/12+1/11<1/4$, the asserted entropy error follows.

\textbf{Step 2: control the two random errors separately.}
In both experiments $\mathbb ES=1$ and $\mathbb EW=0$, because
$\mathbb E(U^\alpha-AU)=0$. Independence gives
$\operatorname{Var}(S)=\operatorname{Var}(U)/b$ in the first experiment.
Conditioning on $Q$ gives
$\operatorname{Var}(S)=\mathbb EU^2/b$ in the second. Similarly, in
either experiment,
$\operatorname{Var}(W)=\mathbb E(U^\alpha-AU)^2/b$.
Chebyshev's inequality and a union bound prove
\eqref{unif:eq:normalization-failure}.
\end{proof}

\begin{remark}[The order-one limit is already included]
Define $g_\beta(u)=(u^{1+\beta}-u)/\beta$, with $g_\beta(0)=0$.
There is the exact cancellation
\begin{equation}\label{unif:eq:affine-cancellation}
 U^\alpha-AU
 =\beta\{g_\beta(U)-(\mathbb Eg_\beta(U))U\}.
\end{equation}
When the separation margin $m$ is proportional to $\beta$, the factor
$\beta^2$ in $v_\alpha$ cancels the factor $m^2$ in
\eqref{unif:eq:normalization-failure}. Thus the lemma does not lose a
factor $\beta^{-2}$ as $\alpha\downarrow1$. For example, if
$U\in\{0\}\cup[1,B]$, $0<\beta\le1/2$, and $A$ is bounded, then
$|g_\beta(U)|\le U^{1+\beta}\log B$ and
$\mathbb Eg_\beta(U)\le A\log B$. Consequently
\[
 \mathbb E\{g_\beta(U)-(\mathbb Eg_\beta(U))U\}^2
 \le C_A B^3(1+\log B)^2.
\]
This is a direct substitution into one variance formula, not a second
normalization argument.
\end{remark}

\subsection{Constructing moment-matched weights by interpolation}

A signed measure that integrates every low-degree monomial to zero
encodes moment matching: its positive and negative parts have identical
moments. Lagrange interpolation constructs such a measure explicitly.
The small perturbation $\rho$ below is deliberate. It permits one
integer moment to differ by a controlled amount, which is needed close
to an integer order.

For distinct nodes $u_0,\ldots,u_n$, the Lagrange basis is
\[
 \ell_j(u)=\prod_{i\ne j}\frac{u-u_i}{u_j-u_i}.
\]
Every polynomial $P$ of degree at most $n$ satisfies
$P(a)=\sum_j\ell_j(a)P(u_j)$. We evaluate this identity at a point
$a$ below every node, rather than inside the node interval.

\begin{lemma}[Perturbed interpolation factor]
\label{unif:lem:interpolation-factor}
Let $u_0,\ldots,u_n>0$ be distinct, let $0<a<\min_j u_j$, and let
$\ell_j$ be their Lagrange basis polynomials. For $\rho\ge0$, define
\[
 \nu_\rho=(1+\rho)\delta_a
 -\sum_{j=0}^n\ell_j(a)(1+\rho a/u_j)\delta_{u_j},
 \qquad C_a=\prod_{j=0}^n(1-a/u_j).
\]
Then
\begin{gather}
 \int u^j\,d\nu_\rho=0 \qquad(1\le j\le n),\\
 \int d\nu_\rho=\rho C_a,
 \qquad\int u^{-1}\,d\nu_0=a^{-1}C_a.
\end{gather}
In particular, $\nu_0$ annihilates every polynomial of degree at most
$n$. If $e\in\mathbb R$ and $P$ is the positive Jordan mass of
$u^{-e}\nu_\rho$, the measure $\lambda=u^{-e}\nu_\rho/P$ has positive
mass one and satisfies
\[
 \int u^{e+j}\,d\lambda=0\quad(1\le j\le n),
 \qquad \int u^e\,d\lambda=\rho C_a/P.
\]
\end{lemma}
\begin{proof}
For $1\le j\le n$, interpolation of $u^j$ and $u^{j-1}$ gives
\[
 (1+\rho)a^j-\sum_i\ell_i(a)u_i^j
             -\rho a\sum_i\ell_i(a)u_i^{j-1}=0.
\]
If $R$ interpolates $u\mapsto1/u$ at the nodes, the polynomial
$1-uR(u)$ has the same zeros and constant term as
$\prod_i(1-u/u_i)$ and has degree at most $n+1$. They therefore agree.
Evaluating at $a$ proves both remaining identities. Weighting by the
strictly positive function $u^{-e}$ and normalizing proves the last
assertions. The argument includes $n=0$, with the first cancellation
list empty.
\end{proof}

\paragraph{Why use two interpolation factors?}
Write a noninteger order as an integer plus a fractional part. One
factor will cancel integer moments below that order; the other will
cancel integer moments above it. A tensor product combines the two
cancellation lists. We separate its positive and negative terms into
two \emph{positive} measures, so the final construction gives probability
laws rather than formal signed objects.

\begin{lemma}[Tensor interpolation and mean-one priors]
\label{unif:lem:tensor-priors}
Fix $\alpha>1$ and $\beta=\alpha-1$. Let $\lambda_-$ and $\lambda_+$
be finite signed measures on $[a_-,1]$ and $[a_+,1]$, respectively,
where $0<a_-,a_+<1$. Their positive Jordan masses are one and their
negative Jordan masses are $R_-,R_+\in[0,1)$, with $R_-+R_+>0$.
Write $B=(a_-a_+)^{-1}$ and assume the single mean-matching identity
\begin{equation}\label{unif:eq:mean-annihilation}
 \left(\int u^\beta\,d\lambda_-(u)\right)
 \left(\int w^{-\beta}\,d\lambda_+(w)\right)=0.
\end{equation}
Push the following two positive tensor measures forward by
$z=a_-w/u$:
\[
 \mu_1=\lambda_-^+\otimes\lambda_+^+
       +\lambda_-^-\otimes\lambda_+^-,\qquad
 \mu_0=\lambda_-^+\otimes\lambda_+^-
       +\lambda_-^-\otimes\lambda_+^+.
\]
Then $Z=\int z^{-\beta}d\mu_i$ has the same positive value for
$i=0,1$. Place mass $(BZ)^{-1}z^{-\alpha}d\mu_i$ at $U_i=Bz$ and
complete its total mass to one by an atom at zero. These are probability
laws on $\{0\}\cup[1,B]$, with $\mathbb EU_i=1$. For every real $j>0$,
\begin{equation}\label{unif:eq:tensor-moments}
 \mathbb EU_1^j-\mathbb EU_0^j
 =\frac{B^{j-1}}Z
 \left(\int(u/a_-)^{\alpha-j}\,d\lambda_-(u)\right)
 \left(\int w^{j-\alpha}\,d\lambda_+(w)\right).
\end{equation}
Writing $A_i=\mathbb EU_i^\alpha$, one has exactly
\begin{equation}\label{unif:eq:tensor-gap}
 A_1=\frac{B^\beta}Z(1+R_-R_+),\qquad
 A_0=\frac{B^\beta}Z(R_-+R_+),\qquad
 \log\frac{A_1}{A_0}
 =\log\frac{1+R_-R_+}{R_-+R_+}.
\end{equation}
In particular, an exact accuracy $\varepsilon$ and a positive logarithmic
margin $m$ are preserved whenever
$1+R_-R_+\ge e^{2\varepsilon\beta+m}(R_-+R_+)$.
\end{lemma}
\begin{proof}
\textbf{Step 1: keep the two hypotheses positive.}
The notation ``pushforward'' means that an atom at $(u,w)$ is moved
to the scalar value $z=a_-w/u$, without changing its mass.
Both pushforwards are supported on $[a_-a_+,1]=[1/B,1]$.
Their difference is the pushforward of $\lambda_-\otimes\lambda_+$.
Its $z^{-\beta}$ integral vanishes by
\eqref{unif:eq:mean-annihilation}. Their common integral $Z$ is positive
because both positive measures have positive total mass.
\textbf{Step 2: reweight to enforce mean one.}
At the point $U=Bz$, the assigned mass is proportional to
$z^{-\alpha}$. Multiplication by $U$ therefore leaves the power
$z^{1-\alpha}=z^{-\beta}$, precisely the quantity whose integrals
agree. The constructed nonzero measure has first moment
$\int z^{-\beta}d\mu_i/Z=1$. Its support lies in $[1,B]$, so its
total mass is at most one. Completion at zero is therefore legitimate.
\textbf{Step 3: read every moment from the tensor product.}
Direct integration gives
\[
 \mathbb E U_i^j
 =\int (Bz)^j\frac{z^{-\alpha}}{BZ}\,d\mu_i(z)
 =\frac{B^{j-1}}Z\int z^{j-\alpha}\,d\mu_i(z).
\]
Subtract and factor the tensor integral to obtain
\eqref{unif:eq:tensor-moments}. At $j=\alpha$, the integral is just
the total mass of $\mu_i$, proving \eqref{unif:eq:tensor-gap}.
No cancellation of overlapping atoms after pushforward is needed.
\end{proof}

\paragraph{Which moments actually cancel?}
The next statement is the algebraic reason for choosing the two
fractional exponents $\gamma$ and $1-\gamma$. The factor below
$\alpha$ handles degrees $1,\ldots,\lfloor\alpha\rfloor$; the other
factor handles the larger degrees. Only the degrees adjacent to
$\alpha$ can fail to match, and turning off the appropriate perturbation
eliminates either exception.

\begin{corollary}[All noninteger moment cancellations]
\label{unif:cor:all-moment-cancellations}
Let $\alpha=m_0+\gamma>1$ be noninteger, where $m_0\ge1$ is an
integer and $0<\gamma<1$. Choose an integer $L\ge m_0$.
Use Lemma~\ref{unif:lem:interpolation-factor} with exponent $\gamma$
and degree $m_0-1$ for $\lambda_-$, and with exponent $1-\gamma$
and degree $L-m_0$ for $\lambda_+$. Let their perturbations be
$\rho_-,\rho_+$. Suppose the mass assumptions in
Lemma~\ref{unif:lem:tensor-priors} hold. If $m_0=1$, require
$\rho_-=0$. Then the resulting priors have equal means. Their integer
moments through degree $L+1$ all agree except possibly moments $m_0$
and $m_0+1$; moment $m_0$ also agrees if $\rho_-=0$, and moment
$m_0+1$ also agrees if $\rho_+=0$.
\end{corollary}
\begin{proof}
The two factors in \eqref{unif:eq:tensor-moments}, apart from positive
normalizers and the power of $a_-$, are
\[
 \int u^{m_0-j}\,d\nu_{\rho_-}(u),\qquad
 \int w^{j-m_0-1}\,d\nu_{\rho_+}(w).
\]
The first vanishes for $1\le j<m_0$ and also for $j=m_0$ if
$\rho_-=0$. The second vanishes for $m_0+1<j\le L+1$ and also for
$j=m_0+1$ if $\rho_+=0$. In particular $j=1$ always vanishes under
the stated restriction. This proves mean matching before the
normalization in Lemma~\ref{unif:lem:tensor-priors} is invoked.
\end{proof}

\paragraph{Substitution at the order-one endpoint.}
For $0<\beta\le1/2$, use $\alpha=1+\beta$ and
\[
 a_-=e^{-8},\qquad
 \lambda_-=\delta_{a_-}-e^{-8\beta}\delta_1.
\]
This is precisely the unperturbed degree-zero factor with exponent
$\beta$. Choose the unperturbed complementary factor of exponent $1-\beta$ and
degree $L-1$, with negative mass $R_+\le e^{-16}$. Then all moments
through $L+1$ agree, and
\[
 \frac{A_1}{A_0}
 =\frac{1+e^{-8\beta}R_+}{e^{-8\beta}+R_+}\ge e^{7\beta}.
\]
Indeed this ratio decreases with $R_+$, and, for $r=e^{-16}$,
\[
 \frac d{dx}\left[-\log\frac{e^{-x}+r}{1+re^{-x}}\right]
 =\frac{1-r^2}{(1+re^x)(1+re^{-x})}\ge\frac78
 \quad(0\le x\le8).
\]
Integrate from $0$ to $8\beta$. This is the same tensor construction;
the description using convolution in logarithmic coordinates is a
change of variables.

\paragraph{Substitution near or between higher integers.}
Write $\alpha=r+\sigma s$ with $r\ge2$, $0<s\le1/2$, and
$\sigma\in\{-1,1\}$. Perturb the low factor when $\sigma=1$, and
the high factor when $\sigma=-1$. The only potentially unmatched
integer moment is $r$. If their negative masses satisfy
\[
 R_s\le\kappa,
 \qquad R_c\le\eta\kappa,
 \qquad \kappa=e^{-2\varepsilon\beta-C_0\eta},
\]
then the same exact mass formula gives
\[
 \log(A_1/A_0)
 \ge2\varepsilon\beta+C_0\eta-\log(1+\eta)
 \ge2\varepsilon\beta+(C_0-1)\eta.
\]
Thus no replacement of $\varepsilon$ by a smaller fixed accuracy is
involved.

\subsection{From the weight construction to an impossibility theorem}

We next connect the moment calculation to the observations. A
\emph{fingerprint} records how many symbols appeared once, twice, and
so on; it discards their names. Randomly relabeling the constructed
distributions makes these names uninformative. Poissonizing the number
of weights then makes the fingerprint coordinates independent, which
allows an exact statistical-distance calculation.

For probability laws $P,Q$, write
$\operatorname{TV}(P,Q)=\sup_A|P(A)-Q(A)|$. Under equal prior
probabilities, every test of $P$ versus $Q$ has average error at least
$(1-\operatorname{TV}(P,Q))/2$. Thus total variation at most $1/4$
forces testing error at least $3/8$. We will contradict this using a
hypothetical entropy estimator with error probability at most $1/3$.

\begin{lemma}[Fingerprint separation with exact accuracy]
\label{unif:lem:fingerprint-transfer}
Let $U_0,U_1\in[0,B]$ have mean one and let $A_i=\mathbb EU_i^\alpha$.
Fix $k\ge2$, $b=k/2$, $\theta>0$, and $0<m\le1$. Suppose
\[
 \log(A_1/A_0)\ge2\varepsilon(\alpha-1)+m.
\]
For $i=0,1$, let $\zeta_i$ be the right side of
\eqref{unif:eq:normalization-failure} in the compound-Poisson variant.
Put
\[
 q_{ij}=\mathbb E\frac{e^{-\theta U_i}(\theta U_i)^j}{j!},
 \qquad
 D=\sum_{j\ge1}\frac{(q_{1j}-q_{0j})^2}{q_{1j}+q_{0j}},
\]
with $0/0=0$. If
\[
 bD\le\frac1{16},\qquad
 \max_i\zeta_i\le\frac1{64},\qquad
 (e/4)^{k/2}\le\frac1{128},\qquad
 e^{-\theta b/16}\le\frac1{128},
\]
no estimator based on $n=\lfloor\theta b/4\rfloor\ge1$ iid samples
can have additive entropy error at most $\varepsilon$ with probability
at least $2/3$, uniformly on the $k$-symbol simplex.
\end{lemma}
\begin{proof}
\textbf{Step 1: compute the fingerprint law.}
Under hypothesis $i$, draw $Q\sim\operatorname{Poi}(b)$ weights and
conditionally independent counts with means $\theta U_i$. Independent
Poisson marking gives independent fingerprint entries
$F_j\sim\operatorname{Poi}(bq_{ij})$, $j\ge1$. Their affinity is
\[
 \exp\left\{-\frac b2\sum_{j\ge1}
                    (\sqrt{q_{1j}}-\sqrt{q_{0j}})^2\right\}.
\]
Since $(\sqrt a-\sqrt c)^2\le(a-c)^2/(a+c)$, the total variation of
the two fingerprint laws is at most $\sqrt{bD}\le1/4$.
Every equal-prior test consequently has average error at least $3/8$.

\textbf{Step 2: turn the fingerprint into ordinary labeled data.}
Randomly assign distinct observed symbols to distinct labels in $[k]$
and randomly order their counts. If more than $k$ distinct symbols
are observed, output an arbitrary fixed record. This is one common
randomized kernel; the latter event is contained in $\{Q>k\}$.
On $Q\le k$, it agrees in distribution with first assigning a uniform
injection of all $Q$ coordinates into $[k]$ and then retaining only
the observed labels. The exceptional event has probability at most
$(e/4)^{k/2}$. Conditional on the weights, the total count is
$\operatorname{Poi}(\theta bS)$; its uniformly ordered labels are iid
from the normalized weights, conditional on its length. On the
normalization event, $S\ge3/4$, and the probability of fewer than
$n$ labels is at most $e^{-\theta b/16}$ by the Poisson lower tail.

\textbf{Step 3: an accurate estimate would produce an accurate test.}
Lemma~\ref{unif:lem:normalization} places the two entropies within
$m/[4(\alpha-1)]$ of centers separated by at least
$2\varepsilon+m/(\alpha-1)$. Hence an accurate entropy estimate
followed by the nearer-center decision is a correct test whenever the
normalization, support, and availability events hold. A hypothetical
estimator as in the statement would give testing error at most
\[
 \frac13+\frac1{64}+\frac1{128}+\frac1{128}
 =\frac{35}{96}<\frac38,
\]
a contradiction.
\end{proof}

\paragraph{Quantitative evaluation.}
The lemmas above give one positive-prior construction, one normalization
identity, and one testing transfer for every noninteger order.
Section~\ref{sec:calibrated-common-lower} supplies the interpolation nodes,
perturbations, moment bounds, and likelihood estimates at the root in Lemma~\ref{upcal:parameters}, including the degree-zero factor at the order-one endpoint.
Integer and bounded-root collision obstructions use the hidden-coordinate
bound of Section~\ref{sec:iid-common-lower}.
These lower-bound arguments do not provide the matching noninteger
upper approximation and noise bound.

\subsection{Choosing the parameters at the target sample scale}

The previous lemmas reduce the lower bound to three checks:
\begin{enumerate}
\item the two entropy centers differ by more than $2\varepsilon$,
      with enough margin to absorb normalization;
\item the fingerprint distance is small;
\item normalization, the alphabet-size constraint, and sample availability
      fail with sufficiently small probabilities.
\end{enumerate}
This subsection performs these checks with one parameter rule. The
calculation is longer than the preceding constructions because its
constants must remain uniform when $\alpha$ approaches one, approaches
an integer, or grows with $k$. For fixed $\alpha$, many of these
uniformity estimates would be unnecessary.
\label{sec:calibrated-common-lower}

Write $S_\alpha(k,\varepsilon,\zeta)$ for the least fixed number of
iid observations that estimates $H_\alpha$ within $\varepsilon$ with
failure probability at most $\zeta$, uniformly on $\Delta_k$. We reserve
$\delta$ below for the distance to the next integer. The notation
$K=\log k$ is only an abbreviation within this lower-bound calculation;
it does not denote another alphabet size. Likewise $Q^2$ below is a
deterministic scalar in the root equation, unrelated to the random
number of weights in the normalization lemma.
Fix $0<\varepsilon<1/2$. For $K=\log k$, noninteger
$1<\alpha<K$, and $\beta=\alpha-1$, write
$\gamma=\alpha-\lfloor\alpha\rfloor$, $\delta=1-\gamma$, and
\[
 h=\gamma\delta,\qquad E=2\varepsilon\beta,\qquad
 \Psi=\gamma^{1+1/\delta}\delta^{1+1/\gamma},\qquad
 Q^2=\frac{\Psi}{\alpha^2\beta^2}e^{K/\alpha-E/h}.
\]
Define the regularized root and sample scale by
\[
 X=\inf\left\{y\ge1:\frac{\alpha h}{\beta}y^\alpha
 \left[1+\log\frac{y(1+\log y)}{Q^2}\right]_+\ge1\right\},
 \qquad R=\alpha k^{1-1/\alpha}X.
\]

\begin{theorem}[Calibrated common lower bound]
\label{cal:thm:lower}
There are $c_\varepsilon>0$ and $k_\varepsilon$ such that, for all
$k\ge k_\varepsilon$ and all noninteger $1<\alpha<\log k$,
\[
 S_\alpha(k,\varepsilon,1/3)\ge c_\varepsilon R.
\]
The construction uses the tensor priors of
Lemma~\ref{unif:lem:tensor-priors} and the normalization and testing
lemmas of Section~\ref{sec:common-noninteger-priors}. The bounded-root
case uses the common hidden-coordinate likelihood calculation.
The interpolation parameters follow one rule for every noninteger order.
The accuracy in the displayed formula is exactly $\varepsilon$.
\end{theorem}

We first give the count estimate and the interpolation bounds needed to
verify this theorem. These arguments are included to make the calibration
self-contained at the level of the interpolation and count estimates.

\paragraph{Why matching moments makes counts hard to distinguish.}
The probability of seeing a weight exactly $j$ times is
$\mathbb E[e^{-\theta U}(\theta U)^j/j!]$. Expanding the exponential
expresses this probability through integer moments of $U$. Matching
these moments removes the low-degree terms. The next lemma quantifies
what remains if one moment is allowed to differ.

\begin{lemma}[One unmatched moment]
\label{cal:lem:one-defect}
Let mean-one laws $U_i\in[0,B]$ have matching integer moments through
$L+1$, except possibly moment $r\le L+1$. Suppose $\alpha\le L+2$,
$A_i=\mathbb EU_i^\alpha$, and, for some $D_0>0$, $\rho\ge0$, $c_0>0$,
\[
 |\mathbb EU_1^r-\mathbb EU_0^r|\le\rho D_0,
 \qquad q_{ij}\ge c_0\theta^rD_0/j!\quad(1\le j\le r),
\]
where $q_{ij}=\mathbb E[e^{-\theta U_i}(\theta U_i)^j/j!]$.
Then
\[
 \sum_{j\ge1}\frac{(q_{1j}-q_{0j})^2}{q_{1j}+q_{0j}}
 \le\frac{C2^r}{c_0r!}\theta^r\rho^2D_0
 +C(A_0+A_1)B^{-\alpha}
 \left(\frac{2e\theta B}{L+2}\right)^{L+2}.
\]
If all the moments match, the first term is omitted and the denominator
hypothesis is unnecessary.
\end{lemma}
\begin{proof}
For $j\le L+1$, write
\[
 q_{ij}
 =\sum_{\ell=0}^{L+1-j}
   \frac{(-1)^\ell\theta^{j+\ell}}{j!\,\ell!}
       \mathbb E U_i^{j+\ell}
   +\text{Taylor remainder}.
\]
Subtracting the two expressions cancels all matched moments. Equivalently,
expand $e^{-z}$ through degree $L+1-j$. The only
uncancelled polynomial term is
\[
 a_j=\frac{(-1)^{r-j}\theta^r\Delta_r}{j!(r-j)!}
 \quad(1\le j\le r),\qquad a_j=0\quad(j>r).
\]
The integral Taylor remainder gives
\[
 |R_j|\le\frac{\theta^{L+2}
 (\mathbb EU_0^{L+2}+\mathbb EU_1^{L+2})}{j!(L+2-j)!}.
\]
The tail above $L+1$ has the same factorial-moment bound.
Since $\mathbb EU_i^{L+2}\le B^{L+2-\alpha}A_i$, summing yields
\[
 \sum_{j\ge1}|R_j|\le C(A_0+A_1)B^{-\alpha}
 (2e\theta B/(L+2))^{L+2},
\]
where for tail indices $R_j=q_{1j}-q_{0j}$. If $q=q_{1j}+q_{0j}$,
then $|a_j+R_j|\le q$ implies
$(a_j+R_j)^2/q\le2a_j^2/q+2|R_j|$: for $|R_j|\le q$ use
$2R_j^2/q\le2|R_j|$, and otherwise use $(a_j+R_j)^2/q\le q$.
Finally
\[
 \sum_{j=1}^r\frac1{j!(r-j)!^2}\le\frac{2^r}{r!}.
\]
These estimates prove the assertion, including the case of exact matching.
\end{proof}

\paragraph{Why Chebyshev nodes are used.}
Extrapolating an interpolating polynomial outside its node interval can
make its coefficients very large. The nodes below keep that increase
controlled while leaving one small principal point $a$ outside the
interval. The endpoint coefficient also gives a nonzero atom in each
hypothesis; this will provide the denominator in the count-distance
bound. Here $T_n$ and $U_{n-1}$ denote the Chebyshev polynomials of the
first and second kinds, characterized by
$T_n(\cos t)=\cos(nt)$ and
$U_{n-1}(\cos t)=\sin(nt)/\sin t$, with their polynomial extensions
outside $[-1,1]$.

\begin{lemma}[Interpolation coefficients used in the calibration]
\label{cal:lem:nodes}
For degree $n\ge1$, use the Chebyshev extrema on $[b,1]$, and let
$\ell_j$ be their Lagrange basis. If
$b=c_b\eta/(n+1)^2$, $0<\eta\le3/5$, and $0<a\le b/2$, then,
for an absolute sufficiently small $c_b>0$,
\[
 \sum_j|\ell_j(a)|\le1+C\eta,\quad \ell_b(a)\ge1,
 \quad\sum_{u_j\ne b}|\ell_j(a)|\le C\eta,
 \quad |\ell_1(a)|\ge b/2.
\]
For $1/2\le\xi<1$, instead set $b=\xi/(n+1)^2$. If
$a=\exp(-5)b z^{1/\xi}$ with $0<z\le1$, the signed measure
$u^{-\xi}(\delta_a-\sum_j\ell_j(a)\delta_{u_j})$, divided by its
positive Jordan mass, has negative mass at most $\exp(-\xi)z$ and
principal positive mass at least $1/2$.
The latter conclusion also holds for $n=0$, using the single node $1$
and $b=\xi$.
\end{lemma}
\begin{proof}
Outside the node interval, the signs of the Lagrange coefficients
alternate, so their absolute sum is
\[
 |T_n(t)|,\qquad t=(2a-1-b)/(1-b)<-1.
\]
Writing $t=-\cosh u$ gives
$u\le2\operatorname{arctanh}\sqrt b$. Thus
$|T_n(t)|=\cosh(nu)$.
For the first choice of $b$, $n^2u^2\le C\eta$, whence
$\cosh(nu)\le1+C\eta$.
Every factor in the product formula for $\ell_b(a)$ is at least one;
subtracting this coefficient proves the third inequality.
The exact endpoint formula is
\[
 |\ell_1(a)|=\frac{b-a}{1-b}\frac{|U_{n-1}(t)|}{n}
 \ge b/2,
\]
since $|U_{n-1}(t)|\ge n$.
For the second choice, $b\le1/4$ and
$n\operatorname{arctanh}(\sqrt\xi/(n+1))\le\sqrt\xi$.
(The function $v\mapsto\operatorname{arctanh}v/v$ is increasing,
and $n\operatorname{arctanh}(1/(n+1))\le1$.) Therefore
$\cosh(nu)\le\cosh(2\sqrt\xi)\le\exp(2\xi)\le\exp(4\xi)$,
using $\cosh t\le e^{t^2/2}$.
The total weighted nodal mass divided by the principal mass $a^{-\xi}$
is at most $\exp(4\xi)(a/b)^\xi=\exp(-\xi)z$. Positive normalization only
reduces the negative mass, and its principal mass is at least
$1/(1+\exp(-\xi)z)\ge1/2$. The degree-zero assertion follows directly.
\end{proof}

\begin{proof}[Proof of Theorem~\ref{cal:thm:lower}]
\textbf{Step 1: choose the sample intensity and interpolation degrees.}
The target lower-bound scale is
$R=\alpha k^{1-1/\alpha}X$. We therefore choose
$\theta$ proportional to $R/k$; the eventual fixed sample size will be
$\lfloor\theta k/8\rfloor$. The interpolation degree $L$ controls the
Taylor remainder, while the support bound $B$ is chosen so that
$\theta B$ is a small multiple of $L$.
Put $x=X$. If $x\le X_0$, for a constant $X_0$ chosen below, the
common hidden-coordinate lower bound proves the claim. This case cannot
occur for $1<\alpha\le3/2$ once $k$ is large: here
$Q^2\asymp_\varepsilon k^{1/\alpha}$, so the root's bracket is negative
at every fixed $X_0$. The root comparison used in this observation is
proved at the end of the argument.

Choose a nearest integer $r\ge1$ and write
$\alpha=r+\sigma s$ with $0<s\le1/2$, $\sigma\in\{-1,1\}$;
at ties take the lower integer, so $r=1$ throughout
$1<\alpha\le3/2$. Set
\[
 T=\alpha\log x,\quad v=\frac{\beta}{\alpha}e^{-T},
 \quad w=h+v,\quad
 \zeta=\frac{c_*\ w}{1-e^{-E}},
 \quad \kappa=e^{-E-C_0w},
\]
\[
 L=\lceil C_L(T+\alpha)\rceil,\qquad
 \theta=c_n\alpha k^{-1/\alpha}x,\qquad B=c_mL/\theta.
\]
Here $c_*=c_*(\varepsilon)>0$ is chosen so that $\zeta\le1/8$ for
all orders: since $h\le\min\{\beta,1/4\}$ and
$v\le\min\{\beta,1\}$, it suffices to take a sufficiently small
multiple of $\varepsilon$. Increasing $X_0$ ensures $v\le E/4$ and
$0<w<1/2$. The active factor has exponent $s$ and degree $n_s$;
the complementary factor has exponent $\xi=1-s$ and degree $n_c$:
\[
 (n_s,n_c)=
 \begin{cases}
 (r-1,L-r),&\sigma=1,\\
 (L-r+1,r-2),&\sigma=-1.
 \end{cases}
\]
Define, for every order, the same parameters
\[
 b_s=\frac{c_b\zeta}{(n_s+1)^2},\qquad
 b_c=\frac{\xi}{(n_c+1)^2},\qquad
 a_c=\exp(-5)b_c(\zeta\kappa)^{1/\xi},\qquad
 a_s=(Ba_c)^{-1},
\]
\[
 t=(a_s/b_s)^s/\kappa,\qquad
 \rho=\left[t\exp\left\{C_1\zeta\frac{n_s}{n_s+1}\right\}-1\right]_+.
\]
Use $u^{-s}\nu_\rho$ in the active position and $u^{-\xi}\nu_0$
in the complementary position, followed by positive normalization.
When a degree is zero its only node is $1$; the associated $b$ is
only a scale parameter. The low and high positions are exactly those
of Corollary~\ref{unif:cor:all-moment-cancellations}.

\textbf{Step 2: use the root equation to control the perturbation.}
For $x>X_0>1$, the defining inequality of the root is an equality.
Because $x^\alpha=e^T$ and $v=(\beta/\alpha)e^{-T}$, it gives
\[
 1+\log\frac{x(1+\log x)}{Q^2}=\frac vh.
\]
Taking logarithms in the definition of $Q^2$ and rearranging yields
the identity below. This is the point where the lower-bound construction
uses exactly the same scalar root as the upper-bound proof.
The regularized root equation is
\[
 \frac K\alpha=\frac{E-v}{s(1-s)}+\log x+\log(1+\log x)
                  +1-\log\Psi+2\log\alpha+2\log\beta.
\]
Writing $d_s=n_s+1$, $d_c=n_c+1$, substitution gives the exact identity
\begin{align*}
 \log t={}&\frac{v+C_0w}{\xi}
  +s(1+1/\xi)\log(s/\zeta)+\log\xi
  +s\log c_n+2s\log(\alpha/\beta)\\
 &+s\left[4-\log(c_mc_b)
    +\log\frac{(d_sd_c)^2}{L\alpha^3(1+\log x)}\right].
\end{align*}
For $C_L\ge10$, $d_sd_c\asymp\alpha L$ and
$L\asymp C_L\alpha(1+\log x)$ for every order, including $r=1$.
Thus the second line is $O_{C_L,c_m,c_b}(s)$.

If $n_s\ge1$, then $\beta\ge1/2$ and
$w\asymp s+v$, $\zeta\asymp_\varepsilon s+v$.
In particular the positive part of
$s\log(s/\zeta)$ is $O_\varepsilon(s)$.
Choosing $c_n$ sufficiently small gives
\[
 \log t+C_1\zeta\frac{n_s}{n_s+1}\le C_\varepsilon v-2s.
\]
If $n_s=0$, then $r=1$, $s=\beta$, $w\asymp s$, and
$\zeta$ lies between two positive constants depending only on
$\varepsilon$. The potentially unbounded logarithms cancel exactly:
\[
 s(1+1/\xi)\log s+2s\log(\alpha/\beta)
       =\frac{s^2}{\xi}\log s+2s\log\alpha.
\]
All remaining terms in $\log t-s\log c_n$ are $O_\varepsilon(s)$.
Thus the same choice of $c_n$ makes $\log t\le-2s$ and, automatically,
$\rho=0$. In particular the mean-matching restriction at degree zero
is satisfied by the common parameter rule, without modifying the prior.
For all orders, after increasing $X_0$,
\[
 0\le\rho\le C_\varepsilon v\le1,\qquad
 \rho>0\Longrightarrow (r\ge2\ \hbox{and}\ v\ge c_\varepsilon s).
\]
The identity also gives $|\log t-s\log c_n|\le C_\varepsilon$.
To bound the principal point uniformly when $s$ is very small, retain
its individual terms before dividing by $s$:
\begin{align*}
 \log(a_s/b_s)={}&\log c_n-E/s+v/(s\xi)+C_0w/\xi\\
 &+(1+1/\xi)\log(s/\zeta)+(\log\xi)/s
       +2\log(\alpha/\beta)+O(1).
\end{align*}
Since $v\le E/4$ and $\xi\ge1/2$, the two terms containing $1/s$
and $E,v$ sum to at most $-E/(2s)$. At positive active degree,
$\zeta\ge c_\varepsilon s$ and $\alpha/\beta\le3$ bound the
remaining terms above uniformly. At degree zero, $s=\beta$ and
$\zeta\asymp_\varepsilon1$; the only unbounded logarithms cancel as
$(1+1/\xi)\log s-2\log s=(s/\xi)\log s\le0$.
It follows that
$a_s/b_s\le C_\varepsilon c_n\exp\{-E/(2s)\}\le1/2$.
For $n_s=0$ this inequality simply verifies that the principal point
lies below its only node. Fix $C_1$ before $c_n$. Since
$\rho\le C_\varepsilon v$ and $v/\zeta\le C_\varepsilon$,
a further decrease of $c_n$ ensures
$\rho a_s/b_s\le\zeta$ with coefficient one; thus absorbing the
perturbation into the nodal estimate does not require a circular
choice of $C_1$.

\textbf{Step 3: bound the positive and negative interpolation masses.}
Here ``raw'' means before dividing the signed measure by its positive
mass. Let $R_s,R_c$ be the negative masses after this division, and let
$P_s,P_c$ be the raw positive masses. For $n_s\ge1$ the total active nodal mass is at most
$b_s^{-s}(1+C\zeta)(1+\rho a_s/b_s)$, while its principal mass is
$(1+\rho)a_s^{-s}$. Since $n_s/(n_s+1)\ge1/2$, the definition of
$\rho$ gives $R_s\le\kappa$ when $C_1$ is sufficiently large.
The node at $b_s$ is negative and has coefficient at least one, whereas
the other nodes have total coefficient $O(\zeta)$. Hence
\[
 P_sa_s^s=1+\rho+O(\zeta\kappa),\qquad
 c_{\varepsilon,c_n}\kappa\le R_s\le\kappa.
\]
For $n_s=0$, exactly $P_sa_s^s=1$ and
$R_s=a_s^s=\kappa t b_s^s$. The same bounds hold because
$b_s$ is bounded above and below by positive
$\varepsilon$-dependent constants in this case.
Complementary interpolation gives $R_c\le\zeta\kappa$ and
$1\le P_ca_c^\xi\le2$.

\textbf{Step 4: verify the entropy separation at the original accuracy.}
The exact ratio $(1+R_sR_c)/(R_s+R_c)$ decreases in each negative
mass when both masses lie in $[0,1)$. Its logarithm therefore satisfies
\begin{align*}
 \log(A_1/A_0)
 &\ge E+C_0w+\log(1+\zeta\kappa^2)-\log(1+\zeta)\\
 &\ge E+C_0w-\zeta(1-\kappa^2)\\
 &\ge E+(C_0-2c_*-2C_0\zeta)w\ge E+w,
\end{align*}
after choosing, for example, $C_0\ge4$ and $c_*$ small.
The penultimate inequality uses
$1-\kappa^2\le2(1-e^{-E})+2C_0w$.
This is essential at order one, where both negative masses may not be
small relative to the desired logarithmic margin.
Only integer moment $r$ can differ through degree $L+1$; when $r=1$
even that moment agrees because $\rho=0$.

\textbf{Step 5: bound the unmatched moment and keep the count probabilities positive.}
The moment defect must be small, but smallness alone is insufficient:
the fingerprint distance divides its square by a count probability.
We therefore identify actual atoms that provide lower bounds on those
probabilities.
First suppose $r\ge2$. Let $P_-,P_+$ denote the raw positive masses in the low and high
positions, and let $Z$ be the tensor normalizer. Define
\[
 D_0=\frac{B^{r-1}}{ZP_-P_+}
 \begin{cases}a_-^{-s}a_+^{-1},&\sigma=1,\\
                a_-^{s-1},&\sigma=-1.
 \end{cases}
\]
The interpolation identities give the exact defect
$\Delta_r=\rho C_-C_+D_0$, where
$C_\pm=\prod_j(1-a_\pm/u_{\pm,j})\in(0,1]$.
For $\sigma=1$, the pairs $(a_-,a_+)$ and $(b_s,a_+)$ belong to
opposite parity measures and each contributes at least $D_0$ to its
$r$th moment. For $\sigma=-1$, use $(a_-,a_+)$ and $(a_-,b_s)$.
Their count means are at most $C\zeta/L<1$: for the first sign use
$\theta Ba_c=c_mLa_c\le C\zeta/L$, and for the second use
$\theta Bb_s=c_mLb_s\le C\zeta/L$.
Thus $(\theta U)^j\ge(\theta U)^r$ at these atoms for $j\le r$,
and $e^{-\theta U}$ is bounded below. Consequently
\[
 |\Delta_r|\le\rho D_0,\qquad
 q_{ij}\ge c\theta^rD_0/j!\quad(1\le j\le r).
\]

\textbf{Step 6: bound the power moments needed for normalization.}
The same atoms also control $A_i$ and the second moment of
$U_i^\alpha$. For all $r\ge1$, put
$J=\int u^\beta d\lambda_-^+=\int u^\beta d\lambda_-^-$.
The equality follows from low-factor interpolation, including its
possible perturbation. The exact factorization is
\[
 Z/B^\beta=J\int(a_+/w)^\beta d|\lambda_+|(w),
 \qquad c\le\int(a_+/w)^\beta d|\lambda_+|(w)\le2.
\]
The bounds on the last integral follow from the principal atom and
$a_+/w\le1$.
If $\sigma=1$ and $r\ge2$, its low exponent is $s$ and $\beta-s=r-1\ge1$.
The near node contributes at most $b_s^{r-1}\le C\zeta$, and all
other nodes have total coefficient $O(\zeta)$. The endpoint $1$ has
coefficient at least $b_s/2$; either Jordan sign can be used because
their $\beta$ moments both equal $J$. Dividing by $P_s\asymp a_s^{-s}$ gives
\[
 c a_s^s\zeta/r^2\le J\le C a_s^s\zeta,
 \qquad D_0\le Cr^2/\zeta.
\]
For $r=1$, the low factor has degree zero and exactly
$J=a_s^s$. Since $\zeta\asymp_\varepsilon1$ here, the same two-sided
bound on $J$ holds with $\varepsilon$-dependent constants. No defect
denominator is needed because all moments match.
If $\sigma=-1$, the low factor is complementary and
$\beta-\xi=r-2\ge0$. Its absolute nodal coefficient sum is bounded,
and the endpoint coefficient is at least $b_c/2$ (or one for degree
zero). Since $a_c^\xi=\exp(-5\xi)b_c^\xi\zeta\kappa$,
\[
 e^{-C_\varepsilon r}\zeta\kappa\le J\le C\zeta\kappa.
\]
The endpoint pair $(1,a_+)$ also gives
\[
 D_0\le\frac{C}{a_-|\ell_1(a_-)|}
       \le\frac{e^{C_\varepsilon r}}\zeta.
\]
Here $a_-^{-1}=\exp(5)b_c^{-1}(\zeta\kappa)^{-1/\xi}$ and
$\zeta^{-s/\xi}\le C_\varepsilon s^{-s/\xi}\le3C_\varepsilon$.

It follows in every degree that
\[
 D_0\le e^{C_\varepsilon r}/\zeta\quad(r\ge2),
 \qquad A_i\le c_n^{-s}e^{C_\varepsilon r}/\zeta.
\]
For the second inequality when $\sigma=1$ and $r\ge2$, use
$a_s^s=\kappa t b_s^s$, the lower bound on $t$, and
$b_s^{-s}\le e^{Cr}$; the other sign follows directly from $J$. At degree zero use $J=a_s^s$ and the established bound on $R_s$.

All degrees satisfy $Z/B^\beta\le C_\varepsilon\zeta\kappa$ and
\[
 \int z^\alpha d(\mu_0+\mu_1)(z)\le C_\varepsilon\zeta\kappa.
\]
Indeed $z=a_-w/u\le w$. For a complementary high factor, its total
nodal mass is $O(\zeta\kappa)$ and its principal atom satisfies
$a_c^\alpha\le C\zeta\kappa$. For an active high factor, the coefficient
bounds give $\int w^\alpha d|\lambda_+|\le C a_s^s\zeta\le C\zeta\kappa$;
the low factor's total variation is at most two. Using
$M_1=1+R_sR_c$, $M_0=R_s+R_c\ge c_{\varepsilon,c_n}\kappa$, the exact identity
\[
 \frac{\mathbb EU_i^{2\alpha}}{A_i^2}
 =B^\alpha\frac{Z}{B^\beta}
                 \frac{\int z^\alpha d\mu_i}{M_i^2}
\]
proves
\[
 \frac{\mathbb EU_i^{2\alpha}}{A_i^2}
 \le C_{\varepsilon,c_n}\zeta^2B^\alpha.
\]

\textbf{Step 7: show that the count experiments are close.}
There are two contributions. The first is the controlled defect in
moment $r$. The second is the remainder after matching all other
moments through degree $L+1$. We bound them separately and then use
the fingerprint-transfer lemma.
Lemma~\ref{cal:lem:one-defect} bounds $kD$ by a constant times
\[
 \frac{e^{C_\varepsilon r}k\theta^r\rho^2}{\zeta r!}
 +k(A_0+A_1)B^{-\alpha}
   (2\mathrm e\theta B/(L+2))^{L+2},
\]
with the first term omitted at $r=1$.
If $\rho>0$, then $r\ge2$, $\beta/\alpha\in[1/3,1)$ and
$v\ge c_\varepsilon s$. Thus
$T\le\log(1/s)+C_\varepsilon$, and the root identity implies
$sK/\alpha\le C_\varepsilon r$. Since $\rho\le C_\varepsilon v$,
\[
 \frac{e^{C_\varepsilon r}k\theta^r\rho^2}{\zeta r!}
 \le(C_\varepsilon c_n)^r\frac v\zeta
       \left(k e^{-T}\right)^{\sigma s/\alpha}
 \le(C_\varepsilon'c_n)^r.
\]
Here $v/\zeta$ is uniformly bounded; the absolute logarithm of the
last power is at most $(s/\alpha)(K+T)\le C_\varepsilon r$.
For the remainder, use
\[
 B^\alpha/k=e^{-T}(c_mL/(c_n\alpha))^\alpha,\qquad
 \theta B=c_mL.
\]
Uniformly in all orders $\zeta\ge c_\varepsilon e^{-T}$:
indeed $1-e^{-E}\le\min\{1,2\varepsilon\beta\}$ and
$v=(\beta/\alpha)e^{-T}$. Keeping the favorable factor
$c_n^{\alpha-s}\le1$, the prefactor has logarithm at most
$2T+C_\varepsilon\alpha$. Thus, by taking $c_m$ small and $C_L$
large before decreasing $c_n$, the remainder is at most $e^{-cL}$.
One then fixes $c_n$ and increases $X_0$ to make $kD$ as small as
required by the common testing lemma.

\textbf{Step 8: show that normalizing the weights does not destroy the gap.}
We now return to the residual $U_i^\alpha-A_iU_i$ from the first
lemma. Controlling this residual, rather than $U_i^\alpha$ alone,
preserves the factor $\alpha-1$ needed near order one.
The normalization margin is $m=w<1/2$. We use the following elementary
bound, which automatically removes the order-one singularity.
Since $U_i\in\{0\}\cup[1,B]$ and $\mathbb EU_i=1$, Jensen and
$U_i^\alpha\le B^\beta U_i$ give $1\le A_i\le B^\beta$.
For $u\in[1,B]$, therefore,
\[
 |u^\beta-A_i|
 \le\max\{u^\beta,A_i\}\min\{1,\beta\log B\}.
\]
The zero atom contributes no residual. Consequently, with
$d=\min\{1,\beta\log B\}$,
\[
 \frac{\mathbb E(U_i^\alpha-A_iU_i)^2}{A_i^2}
 \le d^2\left(\frac{\mathbb EU_i^{2\alpha}}{A_i^2}
                      +\mathbb EU_i^2\right)
 \le C_{\varepsilon,c_n}d^2(\zeta^2B^\alpha+B).
\]
The common normalization failure probability is therefore bounded by
\begin{equation}\label{cal:eq:uniform-failure}
 C_{\varepsilon,c_n}\left[
 \frac Bk\left(1+\frac{\beta^2}{w^2}\right)
 +\frac{d^2}{kw^2}(\zeta^2B^\alpha+B)\right].
\end{equation}
This is the same variance bound for every noninteger order.

Here are its uniform numerical bounds. The root equation and
$\Psi/(\alpha^2\beta^2)\le1$ give
\[
 K/\alpha\ge(E-v)/h\ge\varepsilon\beta/h,\qquad
 \beta/w\le K/(\varepsilon\alpha).
\]
Also, after fixing $c_n$ and increasing $X_0$,
\[
 \log(B^\alpha/k)
 =-T+\alpha\log(c_mL/(c_n\alpha))\le-T/2.
\]
At positive active degree, $\alpha\ge3/2$, $\beta\ge1/2$ and
$\zeta/w\asymp_\varepsilon1$.
Thus \eqref{cal:eq:uniform-failure} is at most
\[
 C_{\varepsilon,c_n}\bigl(K^2e^{-K/3}+e^{-T/2}\bigr),
\]
which is as small as needed by the final choices of $X_0$ and $k$.
At degree zero, the same root satisfies
$x\asymp_\varepsilon k^{1/\alpha}/K$, so $T\asymp K$ and
$B\asymp_{\varepsilon,c_n}K^2$.
Also $w\asymp\beta$, $\zeta\asymp_\varepsilon1$, and
$d\le\beta\log B$. Inserting these inequalities into the same
formula \eqref{cal:eq:uniform-failure} bounds it by
\[
 C_{\varepsilon,c_n}(B^\alpha+B)(1+\log B)^2/k=o(1),
\]
uniformly, since $\alpha\le3/2$ at degree zero.

\textbf{Step 9: complete the reduction to a fixed sample size.}
Support overflow and sample shortage vanish uniformly.
For positive degree, $\theta k\ge c_n\alpha k^{1-1/\alpha}\to\infty$;
at degree zero $\theta k\asymp_{\varepsilon,c_n} k/K$.
The exact parity gap is at least $2\varepsilon\beta+w$.
Thus the single testing transfer
Lemma~\ref{unif:lem:fingerprint-transfer} applies at
$n=\lfloor\theta k/8\rfloor\asymp_\varepsilon R$.

\textbf{Step 10: verify the root estimate used at the order-one endpoint.}
This final elementary estimate justifies the degree-zero assertions
used earlier in the proof. Note that
\[
 \log(\Psi/h^2)=(\gamma/\delta)\log\gamma
                  +(\delta/\gamma)\log\delta\in[-2,0].
\]
At degree zero $\alpha=1+\beta\le3/2$, whence
$Q^2\asymp_\varepsilon k^{1/\alpha}$ and
$\alpha h/\beta=1-\beta^2\in[3/4,1]$.
Let $y(1+\log y)=Q^2/\mathrm e$. Then
$y\asymp_\varepsilon k^{1/\alpha}/K$, and the root criterion gives
$y\le x\le2y$ for large $k$. These statements also show that the
bounded-root case occurs only at positive active degree. This completes
the proof with constants depending only on the fixed accuracy.
\end{proof}

\paragraph{What has been proved.}
The lower bound is quantitatively evaluated by one regularized-root
parameter rule, one tensor construction, and one statistical transfer.
The degree-zero interpolation factor receives its exact coefficient
bound within the same argument. This section
concerns the lower bound. The matching upper approximation and
noise theorem is proved in Section~\ref{cu:section}.

\subsection{A complete fixed-real-order specialization}
\label{sec:fixed-real-closure}

This specialization allows arbitrary fixed accuracy and a fixed
confidence level. Its upper bound is a consequence of the common
estimator, while its lower proof records the larger fixed separation
needed for the additive-function transfer.

\begin{theorem}[Matching bounds at every fixed noninteger order]
\label{fixed:thm:renyi}
Fix a noninteger $\alpha>1$, $\varepsilon>0$, and
$0<\delta<1/2$. Then, as $k\to\infty$,
\[
 S_\alpha(k,\varepsilon,\delta)
       =\Theta_{\alpha,\varepsilon,\delta}
                    \left(\frac{k}{\log k}\right).
\]
The upper bound follows from Theorem~\ref{main:theorem}, including
its separate near-one construction. The lower bound is a specialization of the common
tensor prior and Poisson likelihood certificate, with no restriction
that the fixed accuracy be smaller than $1/2$.
\end{theorem}
\begin{proof}
The upper bound follows from Theorem~\ref{main:theorem}.
For a fixed noninteger order, $h,\Psi,\alpha,\beta$ are fixed,
so the root satisfies
$X\asymp_{\alpha,\varepsilon} k^{1/\alpha}/\log k$.
For arbitrary fixed accuracy use
$\varepsilon'=\min\{\varepsilon,1/4\}$; this is only a
fixed-order specialization, and no such accuracy substitution is made
in the uniform theorem. Independent repetitions and a median give any
fixed failure probability. Thus the same estimator supplies the upper
bound stated here.

For the lower bound first take failure probability $1/3$, and allow
any fixed $\varepsilon>0$. Write
$\alpha=m+\gamma$, $m\ge1$, $0<\gamma<1$, and put
$\xi=1-\gamma$ and $E=2\varepsilon(\alpha-1)$. Let
\[
 \kappa=\tfrac12e^{-E-2},\qquad L=\lceil C_2\log k\rceil,
 \qquad n_-=m-1,\quad n_+=L-m.
\]
For an interpolation degree $d\ge1$, use the Chebyshev extrema on
$[b_d,1]$, where $b_d=c_b/(d+1)^2$ and $c_b>0$ is a sufficiently
small absolute constant. The node lemma gives an absolute
$M\ge1$ such that $\sum_j|\ell_j(a)|\le M$ whenever
$0<a\le b_d/2$. At degree zero use the single node $1$ and the
scale $b_0=c_b$; the same bound holds.

For exponent $e\in\{\gamma,\xi\}$ and its corresponding
degree $d\in\{n_-,n_+\}$, set
\[
 a_e=b_d\left(\frac{\kappa}{2M}\right)^{1/e}.
\]
Form the unperturbed weighted interpolation factor
$u^{-e}(\delta_{a_e}-\sum_j\ell_j(a_e)\delta_{u_j})$ and
normalize its positive Jordan mass to one. Denote the resulting
low and high factors by $\lambda_-,\lambda_+$, with principal
points $a_-,a_+$ and negative masses $R_-,R_+$.
The total weighted nodal mass divided by the principal mass is at
most
\[
 M(a_e/b_d)^e=\kappa/2.
\]
This bound also holds at degree zero, since its only node is $1\ge b_0$.
In particular $0<R_-,R_+\le\kappa$, and each normalized factor
retains at least $1/(1+\kappa/2)$ of its positive mass at the
principal point.

Use the common tensor-prior construction with these two factors.
Its integer moments through $L+1$ all match, because both
perturbation parameters are zero. Moreover
\[
 \log(A_1/A_0)
 =\log\frac{1+R_-R_+}{R_-+R_+}
 \ge-\log(2\kappa)=E+2.
\]
Here $B=(a_-a_+)^{-1}=\Theta_{\alpha,\varepsilon}(L^2)$.
We record why $A_0,A_1$ are bounded above independently of $k$,
which is needed in the count estimate. The low factor has fixed
degree, fixed nodes, and fixed principal point. Since it annihilates
$u^{\alpha-1}$,
\[
 J:=\int u^{\alpha-1}\,d\lambda_-^+(u)
   =\int u^{\alpha-1}\,d\lambda_-^-(u)>0
\]
is a fixed constant. Factoring the tensor normalization gives exactly
\[
 \frac Z{B^{\alpha-1}}
 =J\int(a_+/w)^{\alpha-1}\,d|\lambda_+|(w).
\]
The last integral is between $1/(1+\kappa/2)$ and $1+\kappa$:
use the principal atom for the lower bound and $a_+/w\le1$ for
the upper bound. The exact tensor moment formulas now give
$1\le A_i\le C_{\alpha,\varepsilon}$, where the lower bound
also follows from the mean-one condition and Jensen's inequality.

Choose a small absolute $c_3>0$ and put
$\theta=c_3L/B=\Theta_{\alpha,\varepsilon}(1/\log k)$.
The common moment-matching count bound, with its defect term absent,
gives
\[
 kD\le C_{\alpha,\varepsilon}kB^{-\alpha}
        \left(\frac{2e\theta B}{L+2}\right)^{L+2}
 \le C_{\alpha,\varepsilon}kB^{-\alpha}(2ec_3)^{L+2}=o(1),
\]
by first choosing $c_3$ small and then $C_2$ large. Normalization
uses the same residual bound as before, with a fixed logarithmic
margin of one. Since $U_i\in[0,B]$, $\mathbb EU_i=1$, and
$A_i\ge1$,
\[
 \mathbb EU_i^2\le B,\qquad
 \frac{\mathbb E(U_i^\alpha-A_iU_i)^2}{A_i^2}
 \le2\left(\frac{\mathbb EU_i^{2\alpha}}{A_i^2}
                 +\mathbb EU_i^2\right)
 \le2(B^\alpha+B).
\]
Thus its failure probability is
$O_{\alpha,\varepsilon}((B^\alpha+B)/k)=o(1)$.
Support overflow and sample shortage also vanish because
$\theta k\asymp k/\log k\to\infty$. The common fingerprint
transfer therefore rules out an estimator at
$\lfloor\theta k/8\rfloor=\Omega_{\alpha,\varepsilon}(k/\log k)$.
Nothing in these choices restricts the fixed $E$, and hence the
lower bound holds for every fixed positive accuracy.

For a general fixed $\delta<1/2$, an estimator with success
$1-\delta$ can be repeated a constant number of times on independent
blocks and its median has failure at most $1/3$. The binomial
Chernoff bound proves this for a number of blocks depending only
on $\delta$. The preceding $1/3$ lower bound consequently gives
the asserted lower bound for every such $\delta$.
\end{proof}

\clearpage
\section{Logarithmic additive functionals and exact experiment transfer}
\label{app:additive}
\subsection{A real-order universality theorem for additive logarithmic targets}
\label{sec:additive-real-universality}

For a functional $T$ on $\Delta_k$, write $S_T(k,e,\delta)$ for the
least fixed number of iid observations permitting error at most $e$
with failure probability at most $\delta$, uniformly over $\Delta_k$.
Write $S_\alpha$ for the R\'enyi target
$H_\alpha(p)=-(\alpha-1)^{-1}\log\sum_i p_i^\alpha$.
The following reduction is independent of any evaluation of $S_\alpha$.
Its proof uses only the local leading-power assumption and exact sampling reductions.

\begin{lemma}[Dimension-free additive estimation]
\label{lem:dimension-free-additive}
Suppose $f:[0,1]\to\mathbb R$ is continuous, $f(0)=0$, and
$g(x)=f(x)/x$ extends continuously to zero. For every $d,\delta>0$
there is $N=N(f,d,\delta)<\infty$, independent of $k$, and a statistic
$\widetilde F$ based on $N$ iid observations such that
\[
 \sup_{p\in\Delta_k}\Pr_p\left\{
  |\widetilde F-\sum_i f(p_i)|>d\right\}\le\delta.
\]
\end{lemma}
\begin{proof}
Choose a polynomial $B$ satisfying $\|B-g\|_\infty\le d/2$ and put
$Q(x)=xB(x)=\sum_{j=1}^L c_jx^j$. Then
\[
 \left|\sum_i Q(p_i)-\sum_i f(p_i)\right|
 \le\frac d2\sum_i p_i=\frac d2.
\]
For a block $X_1,\ldots,X_L$ of iid observations let
\[
 W=c_1+\sum_{j=2}^L c_j\mathbf1\{X_1=\cdots=X_j\}.
\]
Its mean is $\sum_iQ(p_i)$ and $|W|\le C_Q:=\sum_j|c_j|$.
The average of $b$ independent blocks has variance at most $C_Q^2/b$.
Taking an integer $b\ge\max\{1,4C_Q^2/(\delta d^2)\}$ and applying Chebyshev proves the
claim with $N=bL$; if $Q$ is linear its aggregate is already known.
Rao--Blackwellization with respect to the counts replaces each collision
indicator by $\sum_i(N_i)_j/(N)_j$, the same factorial-count statistic
used elsewhere in the framework. Polynomial approximation therefore
proves the claim without a dimension-dependent distribution estimate.
\end{proof}

\begin{theorem}[Transfer for every fixed real leading order]
\label{thm:real-leading-order-transfer}
Let $\alpha>1$, $\beta>0$, and $f\in C([0,1])$ satisfy
\[
 f(0)=0,\qquad f(x)>0\quad(0<x\le1),\qquad
 \lim_{x\downarrow0}\frac{f(x)}{a x^\alpha}=1
 \quad\text{for some }a>0.
\]
Define $H_{f,\beta}(p)=-\beta^{-1}\log F_f(p)$,
$F_f(p)=\sum_i f(p_i)$. For every $\varepsilon>0$, every
$0<\zeta<\beta\varepsilon$, and every $0<\delta<1/2$, there are
integers $m,N<\infty$, independent of $k$, such that
\begin{align}
 S_\alpha\left(\lfloor k/m\rfloor,
       \frac{\beta\varepsilon+\zeta}{\alpha-1},\delta\right)
 &\le S_{H_{f,\beta}}(k,\varepsilon,\delta),
 \label{eq:real-additive-lower-transfer}\\
 S_{H_{f,\beta}}(k,\varepsilon,\delta)
 &\le N+S_\alpha\left(k,
       \frac{\beta\varepsilon-\zeta}{\alpha-1},\delta/2\right).
 \label{eq:real-additive-upper-transfer}
\end{align}
The lower inequality is understood for $k\ge m$.
The constants $m,N$ may depend on $f,\alpha,\beta,\varepsilon,
\zeta,\delta$. No uniformity in a changing order or changing function
is asserted.
\end{theorem}
\begin{proof}
Choose $\tau\in(0,1)$ sufficiently small that
\begin{equation}\label{eq:local-power-band}
 e^{-\zeta}a x^\alpha\le f(x)\le e^{\zeta}a x^\alpha
 \qquad(0\le x\le\tau).
\end{equation}

For the lower bound choose an integer $m\ge1/\tau$ and put
$d=\lfloor k/m\rfloor$. Given $p\in\Delta_d$, form its replicated
distribution on $md$ symbols,
\[
 q_{ij}=p_i/m,\qquad 1\le i\le d,\quad1\le j\le m,
\]
and pad with zeros to reach $k$ symbols. Sampling from $q$ requires
one sample $i$ from $p$ and an independent uniform label $j$; the latter
contains no information about $p$. Summing \eqref{eq:local-power-band}
gives
\[
 \left|\log F_f(q)-\log\left(a m^{1-\alpha}
                       \sum_i p_i^\alpha\right)\right|\le\zeta.
\]
Thus any estimator $\widehat H_f$ with error $\varepsilon$ yields
\[
 \widehat H_\alpha=
 \frac{\beta}{\alpha-1}\widehat H_f+
 \frac{\log(a m^{1-\alpha})}{\alpha-1}
\]
with error at most $(\beta\varepsilon+\zeta)/(\alpha-1)$ and
the same failure probability. This proves
\eqref{eq:real-additive-lower-transfer}. The construction consists
of exactly normalized probability vectors; no random-weight
normalization or concentration of a prior functional is required.
Moreover, the same equivalence holds exactly in the Poisson experiment:
independent counts of means $np_i/m$ aggregate to independent counts of
means $np_i$, and multinomial splitting reconstructs the finer counts.
For a positive reference $p^*$, replication preserves the likelihood
quadratic form itself:
\[
 \sum_{i,j}\frac{(q_{ij}-q_{ij}^*)(\widetilde q_{ij}-q_{ij}^*)}
                     {q_{ij}^*}
 =\sum_i\frac{(p_i-p_i^*)(\widetilde p_i-p_i^*)}{p_i^*}.
\]
Thus every Poisson-kernel testing certificate for the base power target
is transported without statistical loss; only the deterministic
functional approximation costs the displayed $\zeta$.

For the upper bound let $b=\min_{x\in[\tau,1]}f(x)>0$ and set
\[
 d_0=\frac b4(1-e^{-\beta\varepsilon}).
\]
Since $f(x)/x\to0$, Lemma~\ref{lem:dimension-free-additive} gives
an $N$-sample statistic $\widetilde F$ with
$|\widetilde F-F_f|\le d_0$ except on an event of probability
$\delta/2$, uniformly in $k$. Independently run a R\'enyi estimator
with error $(\beta\varepsilon-\zeta)/(\alpha-1)$ and failure
$\delta/2$. Output
\[
 \widehat H_f=
 \begin{cases}
  -\beta^{-1}\log\widetilde F,&\widetilde F\ge b/2,\\[1mm]
  ((\alpha-1)\widehat H_\alpha-\log a)/\beta,
       &\widetilde F<b/2.
 \end{cases}
\]
On the simultaneous success event, in the first branch
$F_f\ge b/2-d_0>b/4$, so
\[
 \left|\frac{\widetilde F}{F_f}-1\right|
 \le4d_0/b=1-e^{-\beta\varepsilon}.
\]
The resulting ratio lies between $e^{-\beta\varepsilon}$ and
$2-e^{-\beta\varepsilon}\le e^{\beta\varepsilon}$, proving
the required logarithmic accuracy. In the second branch
$F_f\le b/2+d_0<3b/4<b$. No coordinate can then be at least
$\tau$, because such a coordinate alone contributes at least $b$.
Consequently \eqref{eq:local-power-band} applies to every coordinate
and
\[
 |\log F_f-\log(aP_\alpha)|\le\zeta.
\]
The error of the displayed output is at most
$(\beta\varepsilon-\zeta+\zeta)/\beta=\varepsilon$.
A union bound proves \eqref{eq:real-additive-upper-transfer}.
\end{proof}

\begin{corollary}[Consequence of any fixed-order minimax evaluation]
Suppose a rate $R_\alpha(k)$ satisfies
$S_\alpha(k,e,\delta)=\Theta_{\alpha,e,\delta}(R_\alpha(k))$
for every fixed $e>0$, $0<\delta<1/2$, and
$R_\alpha(\lfloor k/m\rfloor)=\Theta_{\alpha,m}(R_\alpha(k))$
for every fixed positive integer $m$. If $R_\alpha(k)\to\infty$,
then every $f$ in Theorem~\ref{thm:real-leading-order-transfer}
satisfies
\[
 S_{H_{f,\beta}}(k,\varepsilon,\delta)
 =\Theta_{f,\alpha,\beta,\varepsilon,\delta}(R_\alpha(k)).
\]
In particular, the integer evaluation proved above gives
$\Theta(k^{1-1/r})$ for any $C^r$ function positive on $(0,1]$ satisfying
$f(0)=\cdots=f^{(r-1)}(0)=0<f^{(r)}(0)$, with $r\ge2$.
For the integer rate, arbitrary fixed accuracy follows from a
smaller-accuracy upper bound and the collision lower bound at that accuracy.
Constant-size median amplification compares any two fixed failure
probabilities in $(0,1/2)$ without changing the alphabet order.
For every fixed noninteger $\alpha>1$, the fixed-order consequence in
Section~\ref{sec:fixed-real-closure} also gives the unconditional
additive rate $\Theta(k/\log k)$. These fixed-order constants are
not uniform as the leading order moves toward an integer.
\end{corollary}

\begin{corollary}[Complete alphabet dependence for positive analytic targets]
\label{cor:analytic-additive-classification}
Suppose $f$ is continuous on $[0,1]$, analytic on a neighborhood of
zero, and positive on $(0,1]$. Let $r\ge0$ be the first index with
$f^{(r)}(0)\ne0$, assuming $f(0)\ge0$. The ambient alphabet size $k$ is known; when $f(0)>0$, zero
coordinates are included in the sum. For fixed
$\beta,\varepsilon>0$ and failure probability $1/3$:
\begin{enumerate}
\item If $r=0$, no observations are needed for all sufficiently large $k$.
\item If $r=1$, $O_{f,\beta,\varepsilon}(1)$ observations suffice.
For every $k$ at which the range of $H_{f,\beta}$ exceeds
$2\varepsilon$, at least one observation is needed; otherwise its
range midpoint is a zero-sample estimator.
\item If $r\ge2$, the complexity is
$\Theta_{f,\beta,\varepsilon}(k^{1-1/r})$.
\end{enumerate}
\end{corollary}
\begin{proof}
Positivity and analyticity imply that $r$ is finite and its leading
coefficient is positive. For $r\ge2$ apply the preceding transfer
and the proved integer rate.

For $r=1$, $h(x)=f(x)/x$ extends to a positive continuous function
on $[0,1]$. Consequently $F_f(p)\ge\min h>0$ uniformly in $k$.
Lemma~\ref{lem:dimension-free-additive}, with a fixed sufficiently
small absolute tolerance, therefore estimates its logarithm with the
requested accuracy. A zero-sample randomized estimator cannot attain
success $2/3$ at two targets farther than $2\varepsilon$ apart,
because their acceptable output intervals are disjoint. A range
midpoint suffices if the range has width at most $2\varepsilon$.

For $r=0$, write $a=f(0)>0$ and $M=\max|f|$. For every $t>0$,
at most $1/t$ coordinates of a probability vector exceed $t$.
If $\omega_f(t)=\sup_{0\le x\le t}|f(x)-a|$, it follows that
\[
 \sup_{p\in\Delta_k}|F_f(p)/k-a|
 \le\omega_f(t)+\frac{2M}{kt}.
\]
First let $k\to\infty$ and then $t\downarrow0$ to obtain uniform
convergence to $a$. Thus the deterministic output
$-\beta^{-1}\log(ka)$ eventually has error at most $\varepsilon$.
\end{proof}

\subsubsection{Why arbitrary smooth convex functions need additional data}

\begin{proposition}[Smooth convex unit-Lipschitz targets can be arbitrarily hard]
\label{prop:flat-function-hardness}
Fix $s>0$ and $A\ge(s+1)/s$, and define
\[
 f_s(0)=0,\qquad
 f_s(x)=\frac{e^{-A/x^s}}{As e^{-A}}\quad(0<x\le1).
\]
Then $f_s$ is $C^\infty$, convex and increasing on $[0,1]$,
$\|f_s'\|_\infty=1$, and $f_s^{(j)}(0)=0$ for every $j\ge0$.
Nevertheless, for every fixed $\beta,\varepsilon>0$,
estimating $-\beta^{-1}\log\sum_i f_s(p_i)$ with failure
probability at most $1/3$ requires
\[
 \Omega_{s,A,\beta,\varepsilon}(k^{2s})
\]
observations as $k\to\infty$.
\end{proposition}
\begin{proof}
Every derivative at a positive $x$ is a finite sum of real powers of
$x$ times $e^{-A/x^s}$, so its limit at zero is zero. Direct
differentiation gives
\[
 f_s''(x)=\frac{e^{-A/x^s}}{As e^{-A}}\,
 As x^{-s-2}\{As x^{-s}-(s+1)\}\ge0.
\]
Thus the derivative is increasing, and the chosen normalization gives
$f_s'(1)=1$.

First take even $k$, let $p_i=1/k$, and let $q$ have $k/2$ entries
$(1+c k^{-s})/k$ and $k/2$ entries $(1-c k^{-s})/k$, where $c>0$
is fixed. For sufficiently large $k$, $q$ is a probability vector.
Taylor expansion of $(1\pm c k^{-s})^{-s}$ at one gives
\[
 -Ak^s\{(1\pm c k^{-s})^{-s}-1\}
 =\pm As c+O(k^{-s}).
\]
Therefore
\[
 \frac{F_{f_s}(q)}{F_{f_s}(p)}\longrightarrow\cosh(As c),
 \qquad
 \chi^2(q,p)=c^2k^{-2s}.
\]
Choose $c$ so that $\beta^{-1}\log\cosh(As c)>2\varepsilon$.
The two targets are then separated by more than $2\varepsilon$
for all sufficiently large $k$. For $n$ iid observations,
\[
 \TV(q^{\otimes n},p^{\otimes n})
 \le\frac12\sqrt{(1+c^2k^{-2s})^n-1}.
\]
If $n\le c_0 k^{2s}$ with a sufficiently small constant $c_0$,
this is less than $1/3$. Binary testing consequently has average
error greater than $1/3$, ruling out the proposed estimator.
For arbitrary $k$, use the largest even alphabet size below $k$
and pad with zeros. Since $f_s(0)=0$, the target is unchanged.
\end{proof}

\begin{proposition}[A single smooth convex target with exponential complexity]
Set
\[
 f(0)=0,\qquad
 f(x)=\frac{\exp(-e^{2/x})}{2e^2\exp(-e^2)}\quad(0<x\le1).
\]
This function is positive on $(0,1]$, increasing, convex,
$C^\infty$ and $1$-Lipschitz, with all derivatives zero at the
origin. For every fixed $\beta,\varepsilon>0$, the sample complexity
of its logarithmic additive functional is
$\Omega_{\beta,\varepsilon}(k^2e^{4k})$ along even alphabet sizes.
\end{proposition}
\begin{proof}
Flatness follows by repeated differentiation. For the unnormalized
function $u(x)=\exp(-e^{2/x})$,
\[
 \frac{u''(x)}{u(x)}
 =\frac{2e^{2/x}}{x^4}\{2(e^{2/x}-1)-2x\}>0
 \qquad(0<x\le1).
\]
The normalization makes $f'(1)=1$, proving the stated Lipschitz
bound. For even $k$, compare uniform $p$ with a vector having half
of its coordinates $(1+t_k)/k$ and half $(1-t_k)/k$, where
$t_k=c e^{-2k}/k$. Then
\[
 e^{2k/(1\pm t_k)}=e^{2k}\mp2c+o(1),\qquad
 \frac{F_f(q)}{F_f(p)}\longrightarrow\cosh(2c),\qquad
 \chi^2(q,p)=c^2e^{-4k}/k^2.
\]
Choose $c$ with $\log\cosh(2c)>2\beta\varepsilon$ and apply
the identical two-point likelihood calculation above.
\end{proof}

The proposition supplies smooth examples, rather than a discontinuous
support functional: convexity, a common Lipschitz constant, and even the
entire zero Taylor series do not imply a common polynomial sample rate.
Theorem~\ref{thm:real-leading-order-transfer} gives a complete reduction
when a finite leading power exists. It does not justify a minimax
formula for arbitrary positive continuous $f$.

\begin{corollary}[Nonpower additive functionals with matching bounds]
\label{thm:additive-new}
Fix an integer $r\ge2$, a continuous function
$h:[0,1]\to(0,\infty)$, and $\beta,\eps>0$.
For $f(x)=x^rh(x)$, estimating
$H_{f,\beta}(p)=-\beta^{-1}\log\sum_i f(p_i)$ over $\Delta_k$
with accuracy $\eps$ and success probability $2/3$ requires and suffices
to use $\Theta_{r,h,\beta,\eps}(k^{1-1/r})$ samples.
\end{corollary}
\begin{proof}
Apply Theorem~\ref{thm:real-leading-order-transfer} with $\alpha=r$
and $a=h(0)>0$. The already established integer iid rate has order
$k^{1-1/r}$ for every fixed accuracy and confidence, and remains of
this order under a fixed replication of the alphabet. Both inequalities
in that theorem therefore have the displayed order.
\end{proof}
This includes $x^2+x^{5/2}$, $x^2e^x$, and $e^x-1-x$.
The constants need not be uniform in a changing function or order.

\clearpage

\clearpage
\section{Approximation duality and the limits of scalar reductions}
\label{app:duality}

\subsection{Exact multiplicative approximation and its count dual}
\label{sec:exact-multiplicative-dual}

The following duality preserves the requested logarithmic accuracy exactly.
It concerns approximation and likelihood overlap.  A minimax conclusion
requires the additional variance, target-concentration, and sample-reduction
conditions stated below.

\begin{theorem}[Exact multiplicative approximation duality]
\label{thm:exact-multiplicative-dual}
Let $\Theta$ be a compact metric space, let $f\in C(\Theta)$ satisfy
$f_*:=\min_\Theta f>0$, and let $V$ be a finite-dimensional linear
subspace of $C(\Theta)$ containing the constants. Define
\[
 d(f,V)=\inf_{v\in V}\left\|\frac{f-v}{f}\right\|_\infty.
\]
Then
\begin{equation}\label{eq:weighted-hb-dual}
 d(f,V)=\max\left\{\int f\,d\omega:
   \int v\,d\omega=0\ (v\in V),\quad
  \int f\,d|\omega|\le1\right\}.
\end{equation}
Here $\omega$ ranges over finite signed Borel measures.
The infimum and maximum are attained, and $0\le d(f,V)<1$.
If $d(f,V)>0$, an optimal $\omega$ can be chosen with
$\int f\,d|\omega|=1$.  Its Jordan parts have a common positive
mass $m$, and the probability laws $\mu_+=\omega^+/m$ and
$\mu_-=\omega^-/m$ satisfy
\begin{equation}\label{eq:exact-jordan-ratio}
 \int v\,d\mu_+=\int v\,d\mu_-\quad(v\in V),\qquad
 \frac{\int f\,d\mu_+}{\int f\,d\mu_-}
       =\frac{1+d(f,V)}{1-d(f,V)}.
\end{equation}
\end{theorem}

\begin{proof}
Equip $C(\Theta)$ with the equivalent norm
$\|g\|_f=\|g/f\|_\infty$.  Since $V$ is finite-dimensional, it
is closed, and distance from $f$ to $V$ is attained.  Hahn--Banach
applied to the quotient $C(\Theta)/V$ represents this distance as
the supremum of $\Lambda(f)$ over norm-one-or-smaller continuous
linear functionals $\Lambda$ that annihilate $V$.  The Riesz
representation theorem identifies $\Lambda(g)=\int g\,d\omega$.
Its dual norm is
\[
 \sup_{|g|\le f}\left|\int g\,d\omega\right|
       =\int f\,d|\omega|.
\]
This proves \eqref{eq:weighted-hb-dual}.  The feasible measures
have ordinary total variation at most $1/f_*$; hence their set
is weak-star compact and the objective is continuous, proving
attainment.  A sufficiently small positive constant approximant
has relative error strictly less than one, so $d(f,V)<1$.
When $d(f,V)>0$, homogeneity makes the weighted variation
constraint active.  Because $1\in V$, $\omega(\Theta)=0$, so
$\omega^+(\Theta)=\omega^-(\Theta)=m>0$.  Writing
$A_\pm=\int f\,d\omega^\pm$ gives
$A_++A_-=1$ and $A_+-A_-=d(f,V)$.  Their ratio and the
annihilation conditions give \eqref{eq:exact-jordan-ratio}.
\end{proof}

\begin{corollary}[The same exact logarithmic threshold]
Fix $\beta,\varepsilon>0$, and set
$\eta=\tanh(\beta\varepsilon)$ and
$\kappa=e^{-2\beta\varepsilon}=(1-\eta)/(1+\eta)$.
There exists $v\in V$ satisfying
\[
 (1-\eta)f\le v\le(1+\eta)f
\]
if and only if there exists $q\in V$ satisfying
\[
 \kappa f\le q\le f.
\]
These equivalent conditions hold exactly when $d(f,V)\le\eta$.
If $d(f,V)>\eta$, the laws in
\eqref{eq:exact-jordan-ratio} have matching $V$-moments and
\[
 \frac1\beta\log
  \frac{\int f\,d\mu_+}{\int f\,d\mu_-}>2\varepsilon.
\]
\end{corollary}
\begin{proof}
Use $q=v/(1+\eta)$ in the forward direction and
$v=(1+\eta)q$ in the reverse direction.  The separation follows
from $\log((1+t)/(1-t))=2\operatorname{arctanh}(t)$.
\end{proof}

This is an exact correspondence between an approximation band and
separated prior \emph{expectations}.  To obtain a lower bound for
estimating a realized functional, one must additionally establish
that the functional concentrates around those prior expectations.
For product random-weight constructions, this is a separate moment
calculation; for general priors on matrices, it is not automatic.

\begin{theorem}[A finite-dimensional regularized dual]
\label{thm:regularized-weighted-dual}
Under the preceding assumptions, let $V_0\subset C(\Theta)$ be
finite-dimensional and contain the constants, let
$\phi_1,\ldots,\phi_D\in C(\Theta)$, and let $t>0$. Put
\[
 \mathcal D_t(f)=\inf_{v\in V_0,\,c\in\mathbb R^D}
 \left\{\left\|\frac{f-v-\sum_{j=1}^D c_j\phi_j}{f}
                \right\|_\infty+t\|c\|_2\right\}.
\]
Then
\begin{equation}\label{eq:regularized-dual}
 \mathcal D_t(f)=\max_{\omega}
 \left\{\int f\,d\omega:
  \int f\,d|\omega|\le1,\quad
  \int v\,d\omega=0\ (v\in V_0),\quad
  \left\|\left(\int\phi_j\,d\omega\right)_{j=1}^D\right\|_2
       \le t\right\}.
\end{equation}
\end{theorem}

\begin{proof}
The weighted norm has the dual representation used above. Thus the
primal objective is the infimum over $(v,c)$ of the supremum over
$\int f\,d|\omega|\le1$ of
\[
 \int f\,d\omega-\int v\,d\omega
       -\sum_j c_j\int\phi_j\,d\omega+t\|c\|_2.
\]
The measure ball is convex and weak-star compact; for fixed $(v,c)$
the integrand is continuous and affine in $\omega$, and for fixed
$\omega$ it is continuous and convex in $(v,c)$. Sion's minimax
theorem~\cite{Sion58} interchanges infimum and supremum.  The infimum over $v$
is finite precisely when $\omega$ annihilates $V_0$. The infimum
over $c$ of $t\|c\|_2-c\cdot a$ equals zero for $\|a\|_2\le t$
and minus infinity otherwise. This proves the identity. Compactness
also proves dual attainment.
\end{proof}

Unlike in the unregularized case, an optimal signed measure need not
saturate its weighted variation bound: scaling may violate the
moment-norm constraint. Nevertheless, if its objective is $z>0$
and $s=\int f\,d|\omega|\le1$, its Jordan ratio equals
$(s+z)/(s-z)\ge(1+z)/(1-z)$.  Normalizing Jordan parts changes the
size of their moment discrepancy by the inverse of their common
mass. That mass must be retained in any likelihood estimate.

\subsubsection{The Poisson basis gives the same quadratic form on both sides}

Let $\theta=(\theta_e)_{e=1}^q$ range over a compact subset
$\Theta\subset[0,\infty)^q$, and observe independent
$N_e\sim\operatorname{Poi}(\nu_e\theta_e)$, where $\nu_e>0$.
Fix a reference $\theta^*$ with every $\theta_e^*>0$; the
reference need not belong to $\Theta$. Write
$\lambda_e^*=\nu_e\theta_e^*$.  For $\ell\ge0$ define
\[
 C_\ell(N;\lambda)=\sum_{j=0}^{\ell}
   \binom{\ell}{j}(-\lambda)^{\ell-j}(N)_j,
 \qquad H_\ell(N;\lambda)=\frac{C_\ell(N;\lambda)}
                              {\sqrt{\ell!\lambda^\ell}}.
\]
For a multi-index $a\in\mathbb N_0^q$, put
\[
 H_a(N)=\prod_e H_{a_e}(N_e;\lambda_e^*),\qquad
 \phi_a(\theta)=\prod_e
   \frac{\nu_e^{a_e/2}(\theta_e-\theta_e^*)^{a_e}}
        {\sqrt{a_e!(\theta_e^*)^{a_e}}}.
\]
Then
\begin{equation}\label{eq:charlier-common}
 \mathbb E_\theta H_a(N)=\phi_a(\theta),\qquad
 \mathbb E_{\theta^*}H_a(N)H_b(N)=\mathbf1_{\{a=b\}}.
\end{equation}
Indeed, the generating function
\[
 \sum_{\ell\ge0}C_\ell(N;\lambda)\frac{z^\ell}{\ell!}
          =e^{-\lambda z}(1+z)^N
\]
has mean $e^{(\mu-\lambda)z}$ under $\operatorname{Poi}(\mu)$.
The expected product of two such generating functions under
$\operatorname{Poi}(\lambda)$ is $e^{\lambda zw}$, proving both
identities by coefficient comparison and then independence.

For a finite polynomial $Q(\theta)=\sum_a c_a\phi_a(\theta)$,
the statistic $\widehat Q=\sum_a c_aH_a(N)$ is unbiased and
\begin{equation}\label{eq:reference-poly-variance}
 \operatorname{Var}_{\theta^*}(\widehat Q)=\sum_{a\ne0}c_a^2.
\end{equation}
If $L_\theta=dP_\theta/dP_{\theta^*}$ denotes the likelihood ratio,
direct summation of the Poisson series gives
\begin{align}
 K(\theta,\vartheta)
 &:=\mathbb E_{\theta^*}L_\theta L_\vartheta\\
 &=\exp\left\{\sum_e\frac{\nu_e
        (\theta_e-\theta_e^*)(\vartheta_e-\theta_e^*)}
                        {\theta_e^*}\right\}
   =\sum_{a\in\mathbb N_0^q}\phi_a(\theta)\phi_a(\vartheta).
 \label{eq:common-poisson-kernel}
\end{align}
For every finite signed measure $\omega$ on $\Theta$, uniform
absolute convergence on the compact square gives
\begin{equation}\label{eq:signed-likelihood-energy}
 \mathbb E_{\theta^*}\left(\int L_\theta\,d\omega(\theta)\right)^2
 =\iint K(\theta,\vartheta)\,d\omega(\theta)d\omega(\vartheta)
 =\sum_a\left(\int\phi_a\,d\omega\right)^2.
\end{equation}
Thus the Euclidean norm of polynomial coefficients is exactly
reference sampling noise, whereas its dual norm is exactly the
corresponding truncated likelihood overlap. The two sides use
the same feature map, not separately chosen transforms.

If only $|a|\le L$ is retained, put
\[
 R=\max_{\theta\in\Theta}\sum_e
      \frac{\nu_e(\theta_e-\theta_e^*)^2}{\theta_e^*}.
\]
Minkowski's inequality and the multinomial theorem imply
\begin{equation}\label{eq:kernel-tail-explicit}
 \sum_{|a|>L}\left(\int\phi_a\,d\omega\right)^2
 \le\|\omega\|_{\rm TV}^2\sum_{j>L}\frac{R^j}{j!}
 \le\|\omega\|_{\rm TV}^2e^R\frac{R^{L+1}}{(L+1)!}.
\end{equation}
In particular the finite-dimensional dual is an actual likelihood
certificate only after this tail is bounded. Weighted variation
control alone gives $\|\omega\|_{\rm TV}\le1/f_*$, which may be
too costly for sharp entropy rates.

For arbitrary $\theta$, factorial substitution has the exact
variance formula
\begin{equation}\label{eq:global-derivative-variance}
 \operatorname{Var}_\theta(\widehat Q)
  =\sum_{a\ne0}\frac{\theta^a}{a!\nu^a}
                   \bigl(\partial^aQ(\theta)\bigr)^2.
\end{equation}
To verify this identity for a multivariate polynomial, expand its
factorial estimator into monomials, apply the overlap identity in
Lemma~\ref{lem:factorial-energy-new} independently in every coordinate,
and collect terms with common multi-index $a$. The coefficient is
$\theta^a/(a!\nu^a)$ times the square of $\partial^aQ(\theta)$;
the $a=0$ term is the square of the mean and is subtracted.
Consequently a reference norm bound is not a uniform variance bound.
One must either control \eqref{eq:global-derivative-variance}
on the relevant parameter region or use a norm incorporating all
required references.  Removing a polynomial term from the random
statistic is legitimate only if its aggregate value is known.

\subsubsection{Scope of the approximation duality}

\paragraph{Scalar powers and additive functions.}
For $I=[a,b]$ with $a>0$ and $f(x)=x^\alpha$, the unregularized
duality applies for every real $\alpha>1$, with the same exact
accuracy threshold $\tanh((\alpha-1)\varepsilon)$. Integer
powers lie in a finite polynomial space, so their deterministic
approximation error vanishes; the remaining constraint is count
noise. General positive continuous $f$ uses the same theorem.
For $x^\alpha$ on an interval including zero, however, $f_*=0$,
so this theorem cannot simply be invoked. Restricting to
$[a,b]$ and letting $a\downarrow0$ requires uniform weighted
measure bounds, a controlled affine correction using known mass,
and tightness. None follows from the displayed duality alone.

\begin{proposition}[Pointwise relative approximation fails at zero]
Let $b>0$ and let $\alpha>0$ be noninteger. For every polynomial
$Q$,
\[
 \sup_{0<x\le b}\left|\frac{Q(x)}{x^\alpha}-1\right|\ge1.
\]
\end{proposition}
\begin{proof}
For a nonzero $Q$, its first nonzero monomial has integer degree
$r\ne\alpha$. If $r<\alpha$, $|Q(x)/x^\alpha|$ is unbounded
near zero. If $r>\alpha$, $Q(x)/x^\alpha\to0$, so the
displayed supremum is at least one. The zero polynomial has
error exactly one.
\end{proof}
Thus known-mass cancellations and aggregate approximation cannot
be discarded.

A valid global formulation instead takes $\Theta=\Delta_k$ and
$f(p)=\sum_i p_i^\alpha\ge k^{1-\alpha}>0$ for $\alpha>1$.
The exact same duality then applies to polynomials on the simplex,
including additive polynomials with known affine aggregate.
It likewise applies to any positive additive target
$f(p)=\sum_i f_0(p_i)$ with a positive minimum over the chosen
compact parameter class. A rate conclusion from this abstract formulation requires a quantitative
evaluation of the constrained approximation problem. The constructive
iid proof supplies the needed certificates for powers directly.

\paragraph{Symmetry does not force additivity.}
For the iid count model, permutation averaging of a polynomial
statistic preserves its degree and does not increase its worst-case
bias or variance over the simplex. Indeed, averages preserve each
pointwise bias band for a symmetric target, and
$\operatorname{Var}(m^{-1}\sum_j Z_j)\le
m^{-1}\sum_j\operatorname{Var}(Z_j)$ by Cauchy--Schwarz.
Relabeling counts identifies the latter variances with those at
permuted parameters.  Therefore a symmetric polynomial suffices
for a symmetric global approximation/noise certificate.
Such a polynomial need not be additive. For example, on
$\Delta_4$ the symmetric polynomial $(\sum_i p_i^2)^2$ is not
of the form $c+\sum_i g(p_i)$ for a twice differentiable $g$.
Writing $p_4=y=1-p_1-p_2-p_3$, the mixed derivative
$\partial_{p_1}\partial_{p_2}$ of an additive expression is
$g''(y)$, whereas that of the displayed polynomial is
\[
 8(p_1-y)(p_2-y)+4\sum_i p_i^2.
\]
At fixed $y$ and $p_1=p_2=t$, this equals
$32t^2-16t+12y^2+4(1-y)^2$, which varies with $t$ on the
interior of the simplex. Hence permutation symmetrization does
not justify restricting the global primal to additive polynomials.

\clearpage
\section*{AI methodology}
\addcontentsline{toc}{section}{AI methodology}
OpenAI ChatGPT/Codex assisted with mathematical exposition, proof review,
literature comparison, and LaTeX preparation.


\begin{thebibliography}{JVHW15}
\addcontentsline{toc}{section}{References}
\bibitem[AK26]{AK26}
Arman Adibi and Piotr Krysta.
Tight sample bounds for R\'enyi and min-entropy estimation.
arXiv:2607.16966, 2026.
\url{https://arxiv.org/abs/2607.16966}.

\bibitem[AOST17]{AOST17}
Jayadev Acharya, Alon Orlitsky, Ananda Theertha Suresh, and Himanshu Tyagi.
Estimating R\'enyi entropy of discrete distributions.
\emph{IEEE Transactions on Information Theory}, 63(1):38--56, 2017.
\href{https://doi.org/10.1109/TIT.2016.2620435}{doi:10.1109/TIT.2016.2620435}.
\url{https://arxiv.org/abs/1408.1000}.

\bibitem[FS18]{FS18}
Kazuto Fukuchi and Jun Sakuma.
Minimax optimal additive functional estimation with discrete distribution.
arXiv:1812.00001, 2018.
\url{https://arxiv.org/abs/1812.00001}.

\bibitem[JVHW15]{JVHW15}
Jiantao Jiao, Kartik Venkat, Yanjun Han, and Tsachy Weissman.
Minimax estimation of functionals of discrete distributions.
\emph{IEEE Transactions on Information Theory}, 61(5):2835--2885, 2015.
\href{https://doi.org/10.1109/TIT.2015.2412945}{doi:10.1109/TIT.2015.2412945}.
\url{https://arxiv.org/abs/1406.6956}.

\bibitem[OS17]{OS17}
Maciej Obremski and Maciej Skorski.
R\'enyi entropy estimation revisited.
In \emph{Approximation, Randomization, and Combinatorial Optimization.
Algorithms and Techniques (APPROX/RANDOM 2017)}, volume 81 of LIPIcs,
pages 20:1--20:15, 2017.
\url{https://doi.org/10.4230/LIPIcs.APPROX-RANDOM.2017.20}.

\bibitem[PW19]{PW}
Yury Polyanskiy and Yihong Wu.
Dualizing Le Cam's method for functional estimation, with applications
to estimating the unseens.
arXiv:1902.05616, 2019; revised 2021.
\url{https://arxiv.org/abs/1902.05616}.

\bibitem[R\'en61]{Renyi61}
Alfr\'ed R\'enyi.
On measures of entropy and information.
In \emph{Proceedings of the Fourth Berkeley Symposium on Mathematical
Statistics and Probability}, volume 1, pages 547--561.
University of California Press, 1961.
\url{https://projecteuclid.org/euclid.bsmsp/1200512181}.

\bibitem[VV11]{VV11}
Gregory Valiant and Paul Valiant.
Estimating the unseen: An $n/\log(n)$-sample estimator for entropy and
support size, shown optimal via new CLTs.
In \emph{Proceedings of the 43rd Annual ACM Symposium on Theory of
Computing (STOC)}, pages 685--694, 2011.
\url{https://doi.org/10.1145/1993636.1993727}.

\bibitem[Sio58]{Sion58}
Maurice Sion.
On general minimax theorems.
\emph{Pacific Journal of Mathematics}, 8(1):171--176, 1958.
\url{https://doi.org/10.2140/pjm.1958.8.171}.

\bibitem[WY16]{WY16}
Yihong Wu and Pengkun Yang.
Minimax rates of entropy estimation on large alphabets via best
polynomial approximation.
\emph{IEEE Transactions on Information Theory}, 62(6):3702--3720, 2016.
\href{https://doi.org/10.1109/TIT.2016.2548468}{doi:10.1109/TIT.2016.2548468}.
\url{https://arxiv.org/abs/1407.0381}.

\bibitem[WY19]{WY19}
Yihong Wu and Pengkun Yang.
Chebyshev polynomials, moment matching, and optimal estimation of the unseen.
\emph{The Annals of Statistics}, 47(2):857--883, 2019.
\url{https://doi.org/10.1214/17-AOS1665}.

\bibitem[WY20]{WY20}
Yihong Wu and Pengkun Yang.
Polynomial methods in statistical inference: Theory and practice.
\emph{Foundations and Trends in Communications and Information Theory},
17(4):402--586, 2020.
\url{https://doi.org/10.1561/0100000095}.
\end{thebibliography}
\end{document}